\documentclass[10pt,a4paper]{article}
\usepackage{a4wide}

\usepackage{multicol}

\usepackage{xfrac}

\usepackage[affil-it]{authblk}

\usepackage{amssymb, amsfonts, amsbsy, latexsym}
\usepackage{graphicx}
\usepackage{tikz}
\usepackage{color}
\usepackage{colortbl}
\usepackage[english]{babel}
\usepackage{url}
\usepackage{amsmath}
\usepackage{amsthm}

\newtheorem{theorem}{Theorem}
\newtheorem{definition}[theorem]{Definition}
\newtheorem{proposition}[theorem]{Proposition}
\newtheorem{remark}[theorem]{Remark}
\newtheorem{lemma}[theorem]{Lemma}
\newtheorem{example}[theorem]{Example}

\newtheorem{algorithm}[theorem]{Algorithm}
\newtheorem{approximation problem}[theorem]{Approximation problem}

\newtheorem{corollary}[theorem]{Corollary}
\newtheorem{claim}[theorem]{Claim}

\newcommand{\ip}[2]{\left\langle#1,#2\right\rangle}
\newcommand{\abs}[1]{\left|#1\right|}
\newcommand{\norm}[1]{\left\|#1\right\|}

\def\bT{\breve{T}}
\def\bQ{\breve{Q}}
\def\bbQ{{\bf\breve{Q}}}

\def\hf{\hat{f}}
\def\hs{\hat{s}}

\def\Re{{\rm Re}}
\def\Im{{\rm Im}}

\def\w{\omega}
\def\d{\delta}
\def\e{\epsilon}

\def\l{\lambda}
\def\k{\kappa}
\def\a{\alpha}
\def\b{\beta}

\def\x{\omega}
\def\y{\nu}

\def\bw{\boldsymbol{\omega}}
\def\bx{{\bf x}}
\def\bD{{\bf D}}
\def\bg{{\bf g}}

\def\NN{\mathbb{N}}

\def\cF{\mathcal{F}}

\def\cH{\mathcal{H}}

\def\cW{\mathcal{W}}
\def\cS{\mathcal{S}}

\def\CC{\mathbb{C}}
\def\NN{\mathbb{N}}
\def\ZZ{\mathbb{Z}}

\def\RR{\mathbb{R}}
\def\SS{\mathbb{S}}

\usepackage{xcolor}

\definecolor{gr}{RGB}{150,150,255}

\begin{document}

%\title{A wavelet Plancherel theory \\ with application to sparse continuous wavelet transform}

\title{A Wavelet Plancherel Theory \\ with Application to Multipliers and Sparse Approximations}

\author{Ron Levie \ \ \ \ \ \ \ \ \ \ \ \  \ \ \ \  \ \ \ \ \ \ \ \ \ \ \ \ \ \ \ \ \ \ \ \ \     Nir Sochen \\
{\small levie@math.lmu.de \ \ \ \ \ \ \ \ \ \ \ \ \ \ \ \ \ \ \ \ \  \ \ \ \ \ \ \ \ \ \ \ \ \ \ \ \ \ sochen@tauex.tau.ac.il}}
\affil{Ludwig Maximilian University of Munich \ \ \ \ \ \ \ \ \ \ Tel Aviv University \ \ \ \ \ \ \ \ \ \ \ \ \ }
\date{ }

\maketitle

\begin{abstract}
We introduce an extension of continuous wavelet theory that enables an efficient implementation of multiplicative operators in the coefficient space.  In the new theory, the signal space is embedded in a larger abstract signal space -- the so called window-signal space. There is a canonical extension of the wavelet transform to an isometric isomorphism between the window-signal space and the coefficient space. Hence, the new framework is called a wavelet-Plancherel theory, and the extended wavelet transform is called the wavelet-Plancherel transform.  Since the wavelet-Plancherel transform is an isometric isomorphism, any operation in the coefficient space can be pulled-back to an operation in the window-signal space. It is then possible to improve the computational complexity of methods that involve a multiplicative operator in the coefficient space, by performing all computations directly in the window-signal space. As one example application, we show how continuous wavelet multipliers (also called Calder\'{o}n-Toeplitz Operators), with polynomial symbols, can be implemented with linear complexity in the resolution of the 1D signal. As another example, we develop a framework for efficiently computing greedy sparse approximations to signals based on elements of continuous wavelet systems.
\end{abstract}

\textbf{keywords.} Continuous wavelet, Plancherel theorem, wavelet multiplier, sparse decomposition, matching pursuit


\section{\textcolor{black}{Introduction}}
\label{Introduction}

In this paper we study continuous wavelet transforms based on square integrable representations \cite{gmp}. Such transforms are based on the following recipe, which we write more formally in Definition \ref{GCWT}. Let $\cS$ be a separable Hilbert space that we call the \emph{signal space}. 
Let $G$  be a locally compact topological group, and let $\pi:G\rightarrow {\cal U}(\cS)$ be a unitary representation of $G$, mapping $G$ to unitary operators in $\cS$. Given some choice of a signal $f\in\cS$ that we call the \emph{window}, the \emph{wavelet system} is defined to be the set of signals $D=\{\pi(g)f\}_{g\in G}$. Consider the Haar measure $d\mu(g)$ of $G$. We call $G$ \emph{phase space}, and call $L^2(G)=L^2(G;d\mu(g))$ the \emph{coefficient space}. 
The \emph{general continuous wavelet transform} is defined to be
\begin{equation}
V_{f}:\cS\rightarrow L^2(G) \quad ,\quad V_f[s](g)= \ip{s}{\pi(g)f},
\label{eq:1}
\end{equation}
under the assumption that $g\mapsto \ip{s}{\pi(g)f}$ is indeed in $L^2(G)$. The operator $V_f$ is also called the \emph{analysis operator}. Given some conditions that will be specified in Section \ref{Continuous wavelet transforms}, we have the following wavelet reconstruction formula
\[s=CV_f^*V_f[s]= C\int V_f[s](g) \ \pi(g)f\  d\mu(g),\]
where $C>0$ is some constant. Here, the \emph{synthesis operator} $V_f^*$ is the adjoint of $V_f$
(see Remark \ref{remark:GCWT_reco}).

The analysis operator $V_f$ maps signals from $\cS$ to their representation as functions over phase space $G$, and the synthesis operator $V_f^*$ builds signals from a choice of coefficients in $L^2(G)$.
Some examples of general continuous wavelet transforms are the short time Fourier transform (STFT) \cite{Time_freq}, where $G$ is the time-frequency plane, the  1D continuous  wavelet transform (CWT) \cite{Cont_wavelet_original,Ten_lectures}, where $G$ is the time-scale space, and the Shearlet transform \cite{Shearlet}, where $G$ is the position-orientation-scale space.

The analysis operator $V_f:\cS\rightarrow L^2(G)$ is not invertible, and $V_f^*$ is only the pseudo-inverse of $V_f$ (up to the constant $C$). This poses sometimes a difficulty when developing signal processing methods based on wavelet transforms. For motivation, consider the comparison to classical Fourier analysis. There, it is often useful to work interchangeably with the time and the frequency representations of the signal. For example, the statement ``time convolution is equivalent to frequency multiplication'' has a precise meaning. In-fact, every time operation is equivalent to some corresponding frequency operation, and vice-versa, since the Fourier transform is an isomorphism. The equivalent philosophy does not hold in wavelet analysis, as $V_f$ is not onto $L^2(G)$. Hence, in general, operations on functions in $L^2(G)$ cannot be represented, or \emph{pulled-back}, to equivalent operations in $\cS$.   
In this paper we introduce the \emph{wavelet-Plancherel theory}, which aims at resolving this shortcoming. By extending the signal space $\cS$ to the so-called \emph{window-signal space} $\cW\otimes\cS$ (Definition \ref{def:WS}), and canonically extending the wavelet transform $V_f$ to the so-called \emph{wavelet-Plancherel transform} $V$ (see Subsection \ref{A wavelet Plancherel transform}), we can show that $V$ is an isometric isomorphism (Theorem \ref{Th:WP}). The new formulation allows computing more efficiently some operations, defined in phase space,  by pulling them back to operations formulated in the window-signal space.

In this paper we are interested in computational pipelines that involve the multiplication of $V_f[s]$ with a function $F:G\rightarrow\CC$. Define the multiplicative operator $\mathbf{F}:L^2(G)\rightarrow L^2(G)$  by $\mathbf{F}Q(g)=F(g)Q(g)$. One example are continuous frame multipliers $M_F$ \cite{New_mult2}, where $F$ is called the \emph{symbol}. A multiplier with symbol $F$ is defined to be
\[
    M_F = V_f^*\mathbf{F}V_f. 
\]
When $V_f$ is the 1D CWT, $M_F$ is called a Calder\'{o}n-Toeplitz Operator \cite{wave_mult0}. 
Some papers derive theoretical results for multipliers, like boundedess, compactness,  and convergence of multipliers as the symbol converges \cite{New_mult2}. Other works deal with computational aspects that focus on the coefficient space representation of the multiplier, like showing that the adjoint and the inverse of a multiplier are multipliers
\cite{New_mult2, mult_inv}.
Another line of work uses multipliers in applications, e.g., in audio analysis 
\cite{ex3}
and improving signal to noise
\cite{ex4}.
 This paper can be seen as belonging to a fourth line of work. The goal is to derive closed form formulas for multipliers, written in terms of the signal $s$, bypassing the need to map the signal to phase space in the computational pipeline.  %Hence, the multipliers that we study can be computed with complexity corresponding to the resolution of the signal domain.

For this, the multiplier $M_F$ is extended to an operator over the window-signal space $\cW\otimes\cS$ by
\begin{equation}
    \label{eq:mult0}
    \tilde{M}_F = V^*\mathbf{F}V.
\end{equation}
We call (\ref{eq:mult0}) the \emph{pull-back} of the operator $\mathbf{F}$ to the window-signal space.
We show %in Subsection \ref{Closed form formulas for phase space filters}
 that $\tilde{M}_F$ has a simple formulation in terms of the window-signal space, that does not require an explicit computation of the wavelet-Plancherel transform (Propositions \ref{prop:SPWT_pull_obs} and \ref{pullDGWT}). It is then possible to project the resulting formula of $\tilde{M}_F$ from the window-signal space to the signal space. This results in a formula of $M_Fs$ in terms of the signal $s$, that does not require an explicit computation of the wavelet transform $V_f[s]$.
When working with discrete signal of resolution $N$, and polynomial symbols $F$, we show in Subsection \ref{Closed form formulas for phase space filters} that the wavelet-Plancherel implementation of $M_F$ is of computational complexity $O(N)$.

We are also interested in pipelines that start with a signal $s$, and return the value $\norm{\mathbf{F}V_f[s]}_{L^2(G)}$. Such a computation is involved in a search algorithm for large wavelet coefficients, that we propose in this paper. The idea is to choose $F$ as the characteristic function $\chi_{G_0}$ of compact domains $G_0\subset G$. Hence, if $\norm{\mathbf{F}V_f[s]}_{L^2(G)}$ is large,  $V_f[s]$ must have a large coefficient in $G_0$. By considering a sequence of refined partitions of $G$, the computation of the above norms allows pinpointing large wavelet coefficients via a bisection search. We use the resulting search algorithm in a sparse approximation method that we call \emph{wavelet Plancherel phase-space projection pursuit} (WP4),  based  on a wavelet-Plancherel implementation (see Subsection \ref{Sparse approximations with continuous wavelet systems}). When working with discrete signals of resolution $N$, the method allows pinpointing large wavelet coefficients in phase space with resolution $N^2$ with $O(N\log(N))$ operations (see Section \ref{Greedy sparse continuous wavelet transform}).

\subsection{Main contribution}
We summarize our contribution as follows.
\begin{itemize}
	\item 
	We develop the wavelet-Plancherel theory for generic wavelet transforms based on square integrable representations, in which the wavelet transform is canonically extended to an isometric isomorphism from the window-signal space to a space of functions on phase space.
	\item
	We develop closed form formulas for the pull-back of phase space multiplicative operators for a large class of generic wavelet transforms -- semi-direct product wavelet transforms (SPWT) \cite{MyRef}.
	\item
	In the 1D continuous wavelet transform, we show that carrying all computations in the window-signal space leads to fast multiplicative operator computations for a large class of function $F$. For polynomial or characteristic functions $F$, we have a reduction in computational complexity -- $O(N)$ instead of $O(N^2)$ of the naive method. 
	\item
	We utilize the fast implementation of multiplication by characteristic functions in a coefficient search method. The method takes $O(N\log(N))$ operations, instead of the naive $O(N^2\log(N))$. As an example, this search method is used in a matching pursuit algorithm which we call WP4.
\end{itemize}

\subsection{Outline}

The general wavelet-Plancherel theory is derived for a large class of transforms called semi-direct product wavelet transforms \cite{MyRef}, that include the STFT, 1D wavelet transform, and shearlet transform as special cases. Since the general formulation is somewhat long and involved, we begin this paper with a special case -- the wavelet-Plancherel theory of the 1D wavelet transform. Here, the formulas and claims are much simpler.
In Section \ref{The one-dimensional wavelet-Plancherel theory} we present this special case, and postpone the proofs to Sections \ref{A wavelet Plancherel theory} and \ref{Wavelet Plancherel phase space filtering} where we develop rigorously the general theory.

In Section \ref{Preliminaries} we review some general preliminaries. 
In Section \ref{Continuous wavelet transforms} we review the general theory of continuous wavelet transforms based on square integrable representations, which is the basis of the wavelet-Plancherel theory. Moreover, we recall the recently developed class of semi-direct product wavelet transforms (SPWT) \cite{MyRef}, whose special structure will allow the derivation of closed form pull-back formulas of phase space multiplicative operators.  Last, we introduce a new class of wavelet transforms, called \emph{diffeomorphism geometric wavelet transforms (DGWT)}, which will allow a simplification of the wavelet-Plancherel theory and the pull-back formulas. 

In Section \ref{A wavelet Plancherel theory} we derive our wavelet-Plancherel theory for general continuous wavelet transforms based on square integrable representations.
In Section \ref{Wavelet Plancherel phase space filtering} we derive closed form formulas for the pull-back of phase space multiplicative operators. The general formulas are developed for SPWTs, the structure of which allows a functional calculus machinery as the basis of calculations.
Last, for DGWTs, we show that the pull-back of phase space multiplicative operators are explicitly given as convolutions and multiplicative operators along integral lines in the window-signal space.

In Section \ref{Greedy sparse continuous wavelet transform} we show how to utilize the wavelet-Plancherel theory to design an efficient sparse continuous wavelet transform algorithm (WP4), for the special case of the 1D continuous wavelet transform. We moreover derive the computational complexity of the new method. \textcolor{black}{As a toy example}, we show the advantage in having increased resolution in phase space using our model, in a phase vocoder application.

\section{The one-dimensional wavelet-Plancherel theory}
\label{The one-dimensional wavelet-Plancherel theory}

In this section we summarize the wavelet-Plancherel theory in the special case of the 1D wavelet transform. 
The theory in this case is simpler than the general formulation. 
We explain the main ideas without going into the proofs, which are given in the general case in the rest of this paper.

\subsection{The 1D continuous wavelet transform}

The following material can be found, e.g., in \cite{Cont_wavelet_original,Ten_lectures}.
We formulate the wavelet transform for signals represented in the frequency domain. %(In freq domain, scale direction is filter direction.)
We consider a construction of the wavelet transform for analytic signals, namely, signals supported on the positive half frequency line $\left.\left[0,\infty\right.\right)$. Since we focus on a frequency domain representation, the signal space is $L^2(0,\infty)$. Signals supported on the whole frequency line can be treated by considering the positive and negative frequency lines separately.

\subsubsection{Definition of the wavelet transform}
The wavelet transform is based on three operations in $L^2(0,\infty)$.
\emph{Translation} by $g_1$ in the time domain takes the form of modulation in frequency
\begin{equation}
[\hat{L}(g_1)\hf](\w)=e^{-2\pi i g_1 \w}\hf(\w).
\label{eq:1D_mod}
\end{equation}
We call the parameter $g_1$ $time$.
\emph{Dilation} by $g_2$ in frequency is
\begin{equation}
[\hat{D}(g_2)\hf](\w)=e^{\frac{1}{2}g_2}\hf(e^{g_2}\w)
\label{eq:1D_dil}
\end{equation}
We call the parameter $g_2$ $scale$. The smaller the scale $g_2$ the larger the spread of $\hat{D}(g_2)\hf$.
The \emph{wavelet representation} is defined to be the mapping $\pi$ that assigns to each $(g_1,g_2)$ pair the translation-dilation operator 
\[\hat{\pi}(g_1,g_2)=\hat{L}(g_1)\hat{D}(g_2).\]

The \emph{wavelet transform} $V_{\hf}$, based on the mother wavelet $\hf\in L^2(0,\infty)$, is the mapping that transforms each signal $\hs\in L^2(0,\infty)$ to the function
\begin{equation}
V_{\hf}[\hs]: \RR^2\rightarrow\CC  , \quad V_{\hf}[\hs]: (g_1,g_2)\mapsto \ip{\hs}{\hat{\pi}(g_1,g_2)\hf}.
\label{eq:1D_wavelet}
\end{equation}
The mother wavelet $\hf$ is assumed to satisfy the \emph{admissibility condition}
\begin{equation}
\int_{0}^{\infty}\abs{\hf(\w)}^2 \frac{1}{\w}d\w < \infty .
\label{eq:1D_admiss}
\end{equation}
Any $\hf$ that satisfies (\ref{eq:1D_admiss}) is called an \emph{admissible wavelet}. \textcolor{black}{An admissible wavelet is also called a \emph{mother wavelet} in 1D wavelet analysis. In general wavelet analysis, an admissible wavelet is often called a \emph{window}, and since most of this paper deals with general wavelet analysis, we stick to the term window also in 1D continuous wavelet analysis.}

\subsubsection{Phase space}
\label{Phase space}

There is a special measure in the space $\RR^2$ of the $time$-$scale$ parameters. 
We denote by $G$ the space $\RR^2$ with the measure $d\mu(g)=e^{-g_2}dg_1dg_2$, and call $G$ \emph{phase space}.
The space
\[L^2(G) = L^2(\RR^2;d\mu(g))\]
is define to be the space of all measurable functions $F:\RR^2\rightarrow\CC$ with
\[\int_{\RR^2}\abs{F(g_1,g_2)}^2e^{-g_2}dg_1dg_2 <\infty.\]
The inner product in the Hilbert space $L^2(G)$ is
\[\ip{F}{H} = \int_{\RR^2}F(g_1,g_2)\overline{H(g_1,g_2)}e^{-g_2}dg_1dg_2.\]
The wavelet transform $V_{\hf}$ is an isometric embedding of $L^2(0,\infty)$ to $L^2(G)$ up to \textcolor{black}{the admissibility scalar (\ref{eq:1D_admiss}). Namely,
\[ \norm{\hs} = \sqrt{\int_{0}^{\infty}\abs{\hf(\w)}^2\frac{1}{\w} d\w}\norm{V_{\hf[\hs]}}.\]}

\subsubsection{Reconstruction formula and orthogonality relation}

The adjoint of the wavelet transform, $V_{\hf}^*:L^2(G)\rightarrow L^2(0,\infty)$, can be written as
\begin{equation}
V_{\hf}^*[F]=\int_{\RR^2}F(g_1,g_2)\hat{\pi}(g_1,g_2)\hf e^{-g_2}dg_1dg_2,
\label{eq:1D_V_f_adj}
\end{equation}
We have the reconstruction formula
\[V_{\hf}^*V_{\hf}[\hs] = \int_{0}^{\infty}\abs{\hf(\w)}^2\frac{1}{\w} d\w\ \hs.\]
More generally, for every pair of signals $\hs_1,\hs_2\in L^2(0,\infty)$ and admissible wavelets $\hf_1,\hf_2\in L^2(0,\infty)$
\begin{equation}
\ip{V_{\hf_1}[\hs_1]}{V_{\hf_2}[\hs_2]} = \overline{\int_{0}^{\infty}\hf_1(\w)\overline{\hf_2(\w)}\frac{1}{\w} d\w} \ip{\hs_1}{\hs_2}.
\label{eq:1D_orth}
\end{equation}
Equation (\ref{eq:1D_orth}) is called the \emph{orthogonality relation} of the wavelet transform.

The operator $P=V_{\hf}V_{\hf}^*:L^2(G)\rightarrow L^2(G)$ is the orthogonal projection from $L^2(G)$ upon the image space $V_{\hf}[L^2(0,\infty)]$ of the wavelet transform, \textcolor{black}{up to the constant $\int_{0}^{\infty}\abs{\hf(\w)}^2\frac{1}{\w} d\w$}.

\subsection{A wavelet-Plancherel theory}

In this subsection we motivate and define the wavelet-Plancherel transform.

\subsubsection{Motivation from signal processing in phase space}
Often, we would like to perform signal processing in phase space \cite{MyRefStoch}. Namely, we would like to transform a signal $\hs$ to phase space by $V_{\hf}[\hs]$, apply some linear operator $T$ on $V_{\hf}[\hs]$ to obtain $F=TV_{\hf}[\hs]$, and return $V_{\hf}^*F$ as output. Namely
\begin{equation}
Y\hs = V_{\hf}^*TV_{\hf}[\hs].
\label{eq:1D_SPPS}
\end{equation}
\textcolor{black}{One example of (\ref{eq:1D_SPPS}) is a multiplier \cite{wave_mult0,New_mult2}.}

Since phase space is two dimensional, and the domain of the signal space $L^2(0,\infty)$ is one dimensional, it would be beneficial to bypass the transformation to phase space in the pipeline (\ref{eq:1D_SPPS}), and somehow derive a closed form formula of $Y$ in the signal domain.
On the other hand,
$V_{\hf}$ maps $L^2(0,\infty)$ to a ``small'' subspace of $L^2(G)$. For most function $F\in L^2(G)$ there is no signal $\hs\in L^2(0,\infty)$ that gives $V_{\hf}[\hs]=F$. As a result, in general, $F$ is not the wavelet transform of $V_{\hf}^*F$
\[F\neq V_{\hf}V_{\hf}^*F.\]
Hence, in general, $T$ is not unitarily equivalent to $Y$
\[T\neq V_{\hf}YV_{\hf}^*.\]
It is only true that $PTP$ is unitary equivalent to $Y$.
Since there is no simple characterization of the space $V_f[L^2(0,\infty)]$, it is difficult in most cases to find a closed form formula for the pull-back $Y$ of $PTP$.
\textcolor{black}{We overcome this limitation using our wavelet-Plancherel theory, which extends the signal domain, and by which extends the wavelet transform to an isometric isomorphism onto $L^2(G)$.}
%The wavelet-Plancherel theory is a way to extend the signal domain, and by which to extend the wavelet transform to an isometric isomorphism into $L^2(G)$. In this extended version of the wavelet transform there are closed form formulas for pulling back useful operators from phase space to the extended signal space. As we show in this paper, these pull-back formulas reduce the computational complexity of signal processing in phase space when the extended signal space is discretized in the right way.

\subsubsection{The window-signal space} 

%To extend the wavelet transform we extend the signal space.
To extend the signal space, we first consider an extended space of admissible wavelets. Consider the Hilbert space $\cW$ of measurable functions $\hf:(0,\infty)\rightarrow\CC$ satisfying the admissibility condition (\ref{eq:1D_admiss}), with the inner product
\begin{equation}
\ip{\hf_1}{\hf_2}_{\cW}=\int_{0}^{\infty}\hf_1(\w)\overline{\hf_2(\w)} \frac{1}{\w}d\w .
\label{eq:1D_admiss2}
\end{equation}
\textcolor{black}{Namely, $\cW$ is a weighted $L^2$ space.}
Note that there are signals in $\cW$ which are not in $L^2(0,\infty)$. We call $\cW$ the \emph{window space}. To emphasize that $L^2(0,\infty)$ is the space of signals we denote it by $\cS$. Note that $\cW\cap\cS$ is the space of admissible wavelets.

%We extend the signals space $\cS$ as follows. 
We define the \emph{window-signal} space as the tensor product space $\cW\otimes\cS$. \textcolor{black}{For the definition of tensor product in the general case, see Subsection \ref{Tensor products}.} \textcolor{black}{In our case, we have}
\[\cW\otimes\cS = L^2(\RR^2; \frac{1}{\w'}d\w' d\w),\]
where $\w'$ denotes the window variable and $\w$ denotes the signal variable.  There is a way to map window-signal pairs from $\cW\times\cS$ to $\cW\otimes\cS$, namely, by the tensor product operator
\begin{equation}
[\hf\otimes\hs](\w',\w) = \overline{\hf(\w')}\hs(\w).
\label{eq:1D_simple}
\end{equation}
Any function $F\in \cW\otimes\cS$ of the form (\ref{eq:1D_simple}) is called a \emph{simple function}. Note that not every function in $\cW\otimes\cS$ is simple. The window signal space $\cW\otimes\cS$ is the linear closure of all simple functions. Namely,  every $F\in \cW\otimes\cS$  can be approximated by a finite sum of simple functions.
\textcolor{black}{For the theory for general wavelet transforms, see Subsection \ref{The window-signal space}.}

\subsubsection{The wavelet-Plancherel transform}

We extend the wavelet transform to a mapping $V:\cW\otimes\cS\rightarrow L^2(G)$ as follows.
For simple functions we define
\[V(\hf\otimes\hs) = V_{\hf}[\hs].\]
By the orthogonality relation (\ref{eq:1D_orth}), and by the fact that $\cW\otimes\cS$ is the linear closure of simple functions, $V$ is uniquely defined on the whole window-signal space $\cW\otimes\cS$ and must be an isometric embedding. As we show in Example \ref{Example:1Dwave453}, $V$ is actually an isometric isomorphism between $\cW\otimes\cS$ and $L^2(G)$. We call the isomorphism property the \emph{wavelet-Plancherel theorem}. \textcolor{black}{For the theory for general wavelet transforms, see Subsection \ref{A wavelet Plancherel transform}.}

It is interesting to derive closed form formulas for $V$. By (\ref{eq:1D_wavelet}),
the wavelet-Plancherel transform of a simple function $\hf\otimes \hs$ is given by 
\[V(\hf\otimes \hs)(g_1,g_2)=\int_{\RR} \overline{e^{-i g_1 \w}e^{\frac{1}{2}g_2}\hf(e^{g_2}\w)}s(\w)d\w.\]
Thus, by linear closure, for any $F\in\cW\otimes\cS$,
\begin{equation}
V(F)(g_1,g_2)=\int_{\RR} e^{i g_1 \w}e^{\frac{1}{2}g_2}F(e^{g_2}\w,\w)d\w.
\label{eq:1Dwav_slope}
\end{equation}

The pull-back of the classical synthesis $V_{\hf}^*:L^2(G)\rightarrow \cS$ to $V_{\hf}^*V:\cW\otimes\cS\rightarrow \cS$ is given by
\begin{equation}
V_{\hf}^*V F (\w) = \int_{\RR} F(\w',\w){f(\w')}\frac{1}{w'} d\w'.
\label{eq:1Deqg:rr6y}
\end{equation}
Indeed, for simple functions $[\hat{q}\otimes \hs](\w',\w)=\overline{\hat{q}(\w')}\hs(\w)$,
\begin{equation}
\begin{split}
[V_{\hf}^*V (\hat{q}\otimes \hat{s})](\w) & = [V_{\hf}^*V_{\hat{q}} \hs] (\w) = \ip{\hf}{\hat{q}}_{\cW}\hs(\w) \\
 & = \int \hf(\w')\overline{\hat{q}(\w')}\hs(\w)\frac{1}{w'} d\w'.
\end{split}
\label{eq:1DPullWave0}
\end{equation}
Identity (\ref{eq:1Deqg:rr6y}) gives an interpretation for the window signal space. \textcolor{black}{We interpret $V^*_{\hf}Q$ as a view of $Q\in L^2(G)$, since $Q$ contains information outside $V_{\hf}(\cS)$, and $V^*_{\hf}$ extracts the part of $Q$ in $V_{\hf}(\cS)$. The function $F\in\cW\otimes\cS$ can be hence interpreted as a rule for how different windows ``see'' $F$ as a signal. Namely, observed by the window $\hf$, the function $F$ takes the form of (\ref{eq:1Deqg:rr6y}).  To see this better, consider an orthonormal basis of windows $\{\hf_j\}_j$, and represent the function $F$ as
\[F= \sum_j \hf_j\otimes\hs_j,\]
where $\{\hs_j\}_j$ is a sequence of signals. The function $F$
is a ``mixed state'' that defines different views of signals. Viewed by $\hf_j$, $VF$ is seen as $V_{\hf_j}^* VF = \hs_j$. For the general formulation, see Remark \ref{remark:WP_invv}.}

\begin{figure}[!ht]
\centering 
\includegraphics[width=0.6\linewidth]{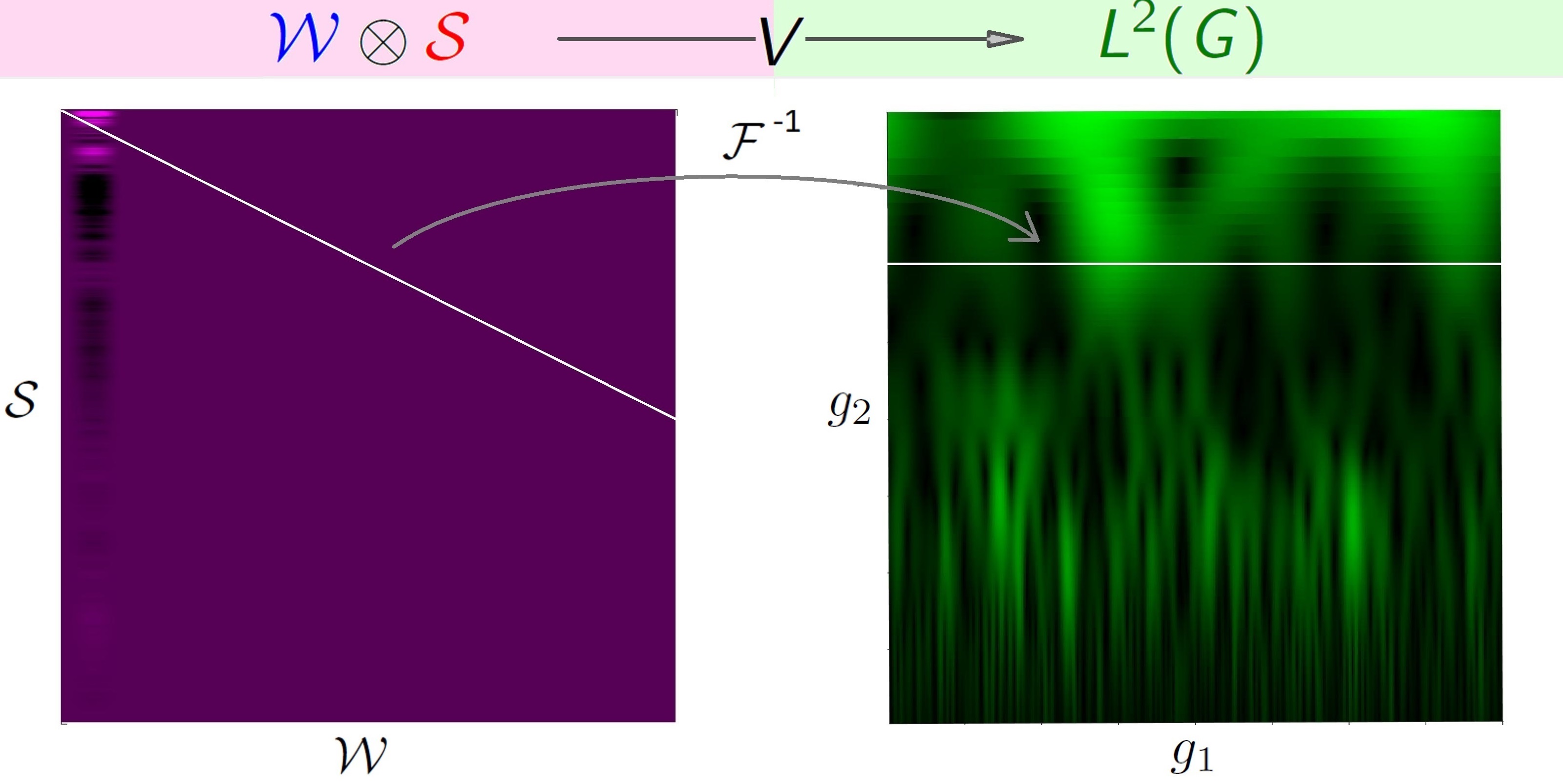}
\caption{Visualization of the wavelet-Plancherel transform and its inverse. Left: a simple function $\hf\otimes\hs\in\cW\otimes\cS$. Since mother wavelets are typically simple ``bump function'' near $\w'=0$, the function $\hf\otimes\hs$ has most of its energy concentrated near the axis $\w$, with $\w'=0$. Right: $V(\hf\otimes\hs)$. Each restriction of $V(\hf\otimes\hs)$ to a constant $g_2$ is  the inverse Fourier transform of the restriction of $\hf\otimes\hs$ to a slope ray. The inverse wavelet-Plancherel transform is  the Fourier transform between restrictions of $V(\hf\otimes\hs)$ to a constant $g_2$, and restrictions of  $\hf\otimes\hs$ to a slope rays.}
\label{fig:V} 
\end{figure}

\subsubsection{Inversion formula for the wavelet-Plancherel transform}
\label{Inversion formula for the wavelet-Plancherel transform}

In view of  Formula (\ref{eq:1Dwav_slope}), a useful change of variabes in $\cW\otimes\cS$ is the \emph{slope change of variable}. We define 
\begin{equation}
z=\frac{\w'}{\w} , \quad \w=\w
    \label{eq:slope3}
\end{equation}
where $z$ is called the \emph{slope}. We call sets of the form 
\begin{equation}
l_z=\{(z,\w)\ |\  \w>0\},
    \label{eq:slope_ray}
\end{equation}
 for fixed $z$, \emph{slope rays}. \textcolor{black}{We define the \emph{slope transform} $\mathcal{L}:\cW\otimes\cS\rightarrow\cW\otimes\cS$ as the unitary operator defined by
\[\mathcal{L}F(z,w) = F(z\w,\w), \quad {\rm with\ } \mathcal{L}^{-1}Q(w',w) = Q(\frac{\w'}{\w},\w).\]}
Equation (\ref{eq:1Dwav_slope}) 
\textcolor{black}{ now takes the form
\begin{equation}
V(F)(g_1,g_2)=\int_{\RR} e^{i g_1 \w}e^{\frac{1}{2}g_2}\mathcal{L}F(e^{g_2},\w)d\w.
\label{eq:1Dwav_slope2}
\end{equation}
}
This gives an interpretation for the wavelet-Plancherel transform. For $F\in \cW\otimes\cS$, the restriction of $VF$ to a fixed scale $g_2$ is the inverse Fourier transform of $F$ along slope rays, up to the normalization \textcolor{black}{$e^{\frac{1}{2}g_2}$}. %The normalization corresponds to the measure \textcolor{black}{$d\mu(g)$} of phase space. 
This interpretation gives an inversion formula for $V$. For $H\in L^2(G)$, the values of $V^*H$ along slope rays are the Fourier transforms of $H$ along time $g_1$ for fixed scales $g_2$. 
\textcolor{black}{ Namely, it is easy to see that for any $H\in L^2(G)$, 
\begin{equation}
    [\mathcal{L}V^*H](z,\w) = z^{-\frac{1}{2}}\frac{1}{2\pi}\int_{\RR} e^{-ig_1\w} H\big(g_1,\ln(z)\big) dg_1.
    \label{eq:IWP_CWT}
\end{equation}
In Figure \ref{fig:V} we visualize the wavelet-Plancherel transform and its inverse.}

\subsection{Wavelet-Plancherel phase space filtering}

We study \emph{phase space filters}, namely, multiplicative operator $\breve{H}$ in $L^2(G)$ of the form
\begin{equation}
  \breve{H} F(g_1,g_2)= H(g_1,g_2)F(g_1,g_2)  
  \label{eq:mult00}
\end{equation}
where $H\in L^{\infty}(G)$.
Wavelet-Plancherel phase space filtering is the theory of deriving closed form formulas for efficiently computing phase space filters directly in the window-signal space, or computing projected versions in the signal space.
\textcolor{black}{For general SPWTs, the theory is derived in Section \ref{Wavelet Plancherel phase space filtering}.}

\subsubsection{Phase space observables}
\label{Phase space observables}

We consider special operators in phase space as building blocks for phase space signal processing operators.
The \emph{phase space time observable}  is the operator $\breve{N}_1$ that multiplies every point in $G$ by its \textit{time} value. Namely, for every $F\in L^2(G)$
\[[\breve{N}_1 F](g_1,g_2) = g_1F(g_1,g_2)\]
The \emph{phase space scale observable} is the operator $\breve{N}_2$ that multiplies every point in $G$ by its \textit{scale} value
\[[\breve{N}_2 F](g_1,g_2) = g_2F(g_1,g_2).\]
\textcolor{black}{For general SPWTs, the theory is derived in Subsection \ref{The pull-back of phase space observables2}.}

\subsubsection{Approximating multiplicative operators using simple functions}
\label{Approximating multiplicative operators using simple functions}

We can formally approximate any multiplicative operator in phase space using functions of $\breve{N}_1$ and $\breve{N}_2$ as follows.
Consider the multiplicative operator $\breve{H}$ defined in (\ref{eq:mult00}) .
We can approximate $H\in L^{\infty}(G)$ by a finite sum of simple functions
\[H(g_1,g_2)\approx \sum_{j=1}^J h_1^j(g_1) h_2^j(g_2)\]
where $h_1^j,h_2^j\in L^{\infty}(\RR)$ for every $j=1,\ldots,J$ and $J\in\NN$.
Consider the operators $h_1^j(\breve{N}_1),h_2^j(\breve{N}_2)$ in $L^2(G)$ defined for every $j=1,\ldots,J$  by
\begin{equation}
h_1^j(\breve{N}_1) F(g_1,g_2)= h_1^j(g_1)F(g_1,g_2) , \quad h_2^j(\breve{N}_2) F(g_1,g_2)= h_2^j(g_2)F(g_1,g_2).
\label{eq:1D_func1}
\end{equation}
We thus have
\[\breve{H}F\approx \sum_{j=1}^J h_1^j(\breve{N}_1)h_2^j(\breve{N}_2)F,\]
\textcolor{black}{We formulate the approximation precisely in Subsection \ref{Phase space filters via trigonometric polynomials of observables}.}
Hence, to study general multiplicative operators in phase space it is enough to focus on multiplicative operators of the form $h_k(\breve{N}_k)$, $k=1,2$, as defined in (\ref{eq:1D_func1}).

\subsubsection{Phase space observables as building blocks for filters}
\label{Phase space observables as building blocks for filters}

In the next subsection we show that there is a closed form formula for the pull-back of phase space observables $\breve{N}_k$ to the window-signal space. Namely, the operators 
\[\breve{T}_k= V^*\breve{N}_kV  ,\quad k=1,2\]
have computationally tractable formulas in $\cW\otimes\cS$.  \textcolor{black}{The operators $\bT_k$ are at the core of our theory, since it is enough to know a formula for $\bT_k$ in order to derive a formula for the pull-back of any phase space filter (\ref{eq:mult00}) from $L^2(G)$ to $\cW\otimes\cS$. This is explained next.}

The pull-back expressions \textcolor{black}{of phase space filters} are based on the application of functions $h_k$ on the operators $\breve{N}_k$ and $\breve{T}_k$. 
The general theory of applying functions on operators is called \emph{functional calculus}. In Subsection \ref{Functional calculus0} we recall the general theory of functional calculus. There, we recall how to apply functions on any normal operator, not just multiplicative operators as in formula (\ref{eq:1D_func1}). It turns out that the pull-back $V^*h_k(\breve{N}_k)V$ of $h_k(\breve{N}_k)$ to $\cW\otimes\cS$ satisfies
\[V^*h_k(\breve{N}_k)V=h_k(V^*\breve{N}_kV) = h_k(\breve{T}_k).\]
Namely, it is enough to derive one formula for $\breve{T}_k= V^*\breve{N}_kV$, and then to use functional calculus  on $\breve{T}_k$ directly in $\cW\otimes\cS$. We also present in Subsections \ref{Scale filters}--\ref{General time filters} closed form formulas for functional calculus in $\breve{T}_k$.
Thus, $\breve{T}_1$ and $\breve{T}_2$ are basic building blocks of phase space filtering. \textcolor{black}{The general theory is developed in Section \ref{Wavelet Plancherel phase space filtering}, and, specifically for DGWTs, in Subsections \ref{Phase space filters via trigonometric polynomials of observables} and \ref{The pull-back of phase space filters for geometric wavelet transforms}.
}

\subsubsection{Pull-back of phase space observables}

In Section \ref{Wavelet Plancherel phase space filtering}, \textcolor{black}{and especially Propositions \ref{prop:SPWT_pull_obs} and \ref{pullDGWT},} we derive formulas for the pull-back of phase space observables for SPWTs and DGWTs. In Subsection \ref{A:The 1D wavelet transform} we show that for the 1D wavelet transform, the pull-back takes the following form.
For simple window-signal functions
\textcolor{black}{
\begin{equation}
\breve{T}_1 (\hf\otimes\hs) (\w',\w) = \hf(\w')\otimes \Big(  i\frac{\partial}{\partial\w}\hs(\w)\Big) -  \Big(i\w'\frac{\partial}{\partial\w'}\hf(\w')\Big) \otimes\Big( \frac{\hs(\w)}{\w}\Big)
\label{eq:1D_Pull_1}
\end{equation}
}
%where the operators $I,A,B$ and $C$ are defined as
%\[I \hat{f}(\w') =  \hat{f}(\w'), \quad
%A\hs(\w)= i\frac{\partial}{\partial\w}\hs(\w), \quad
%B\hf(\w') = i\w'\frac{\partial}{\partial\w'}\hf(\w'), \quad
%C\hs(\w)= \frac{\hs(\w)}{\w}.\]
%%
%%\begin{equation}
%%\breve{T}_1 (\hf\otimes\hs) = \hf\otimes(i\frac{\partial}{\partial\w}\hs) - \big((\cdot) i\frac{\partial}{\partial\w'}\hf\big)\otimes \Big(\frac{\hs}{(\cdot)}\Big)
%%\label{eq:1D_Pull_1}
%%\end{equation}
%%where $(\cdot)\hf$ is the function that maps $\w'$ to $\w'\hf(\w')$, and $\frac{\hs}{(\cdot)}$ maps $\w$ to $\frac{\hs(\w)}{\w}$.
Thus, by (\ref{eq:1D_simple}), 
for general $F\in \cW\otimes\cS$
\begin{equation}
[\breve{T}_1 F](\w',\w) = \Big(i\frac{\partial}{\partial\w} +\frac{\w'}{\w}i\frac{\partial}{\partial\w'}\Big)F(\w',\w) .
\label{eq:1D_Pull_2}
\end{equation}

For simple window-signal functions we have
\textcolor{black}{
\begin{equation}
\breve{T}_2 (\hf\otimes\hs) (\w',\w) = -\hf(\w')\otimes \Big(\ln(\w)\hs(\w)\Big) +\Big( \ln(\w')\hf(\w') \Big)\otimes \hs(\w).
\label{eq:1D_Pull_3}
\end{equation}
}
%%\begin{equation}
%%\breve{T}_2 (\hf\otimes\hs) = \hf\otimes\Big(-\ln(\cdot)\hs\Big) + \Big(\ln(\cdot)\hf\Big)\otimes\hs.
%%\label{eq:1D_Pull_3}
%%\end{equation}
%where
%\[D\hs(\w) = -\ln(\w)\hs(\w), \quad E\hf(\w) = \ln(\w)\hf(\w).\]
Thus, by (\ref{eq:1D_simple}), for general $F\in \cW\otimes\cS$
\begin{equation}
[\breve{T}_2 F](\w',\w) = \Big(\ln(\w')-\ln(\w)\Big)F(\w',\w).
\label{eq:1D_Pull_22}
\end{equation}

\subsection{Closed form formulas for phase space filters}
\label{Closed form formulas for phase space filters}

As explained in Subsection \ref{Phase space observables as building blocks for filters}, to complete the wavelet-Plancherel phase space filtering theory we need to derive formulas for $h(\breve{T}_k)$. We focus on different classes of functions $h$.

\subsubsection{Polynomial filters}
\label{Polynomial filters}

Low order polynomial implementations of phase space filtering are efficient since they represent all signals as combinations of 1D functions, even though the window-signal space and phase space are 2D.

We first explain how to efficiently compute $h(\breve{N}_2)$ applied on simple functions $\hf\otimes\hs$ if $h$ is a low order polynomial.
The idea is to repeat formulas (\ref{eq:1D_Pull_1}) and (\ref{eq:1D_Pull_3}) for the different powers in the polynomial. For example,
\textcolor{black}{
\[\begin{split}
\breve{T}_1^2 (\hf\otimes\hs) =& \breve{T}_1\Big( \hf(\w')\otimes i\frac{\partial}{\partial\w}\hs(\w)\Big) - \breve{T}_1\Big(i\w'\frac{\partial}{\partial\w'}\hf(\w')\otimes \frac{1}{\w} \hs(\w)\Big)\\
  =&  \hf(\w')\otimes(-\frac{\partial^2}{\partial\w^2}\hs(\w)) - \big(i\w' \frac{\partial}{\partial\w'}\hf(\w')\big)\otimes(i\frac{1}{\w}\frac{\partial}{\partial\w}\hs(\w))\\
  &-\big(i\w' \frac{\partial}{\partial\w'}\hf(\w')\big)\otimes \Big(i\frac{1}{\w}\frac{\partial}{\partial\w}\hs(\w) -i\frac{1}{\w^2}\hs(\w) \Big) \\
	& -\big(\w' \frac{\partial}{\partial\w'}\hf(\w') + \w^{\prime 2} \frac{\partial^2}{\partial\w^{\prime 2}}\hf(\w')\big)\otimes \Big(\frac{1}{\w^2}\hs(\w)\Big)
\end{split}
\]
}

%\[\begin{split}
%\breve{T}_1^2 (\hf\otimes\hs) =& \breve{T}_1\Big(\hf\otimes(i\frac{\partial}{\partial\w}\hs)\Big) - \breve{T}_1\Big(\big((\cdot) i\frac{\partial}{\partial\w'}\hf\big)\otimes \Big(\frac{\hs}{(\cdot)}\Big)\Big)\\
%  =&  \hf\otimes(-\frac{\partial^2}{\partial\w^2}\hs) - \big((\cdot) i\frac{\partial}{\partial\w'}\hf\big)\otimes(i\frac{\partial}{\partial\w}\hs)\\
%  &-\big((\cdot) i\frac{\partial}{\partial\w'}\hf\big)\otimes \Big(i\frac{1}{(\cdot)}\frac{\partial}{\partial\w}\hs -i\frac{1}{(\cdot)}\hs \Big) \\
%	& +\big((\cdot) i\frac{\partial}{\partial\w'}\hf + (\cdot)^2 i\frac{\partial^2}{\partial\w^{\prime 2}}\hf\big)\otimes \Big(\frac{\hs}{(\cdot)}\Big)
%\end{split}
%\]
%Here,
%\[[i\frac{\partial}{\partial\w}\frac{\hs}{(\cdot)}](\w') = i\frac{\hs'(\w')}{\w'} -i\frac{\hs(\w)}{\w^{2}}\]
%and
%\[
%\begin{split}
%[(\cdot) i\frac{\partial}{\partial\w'}(\cdot) i\frac{\partial}{\partial\w'}\hf](\w') = &
%- \w'\frac{\partial}{\partial\w'}[\w' \hf'(\w')] =
%- \w'\hf'(\w') -\w'^2\hf''.
%\end{split}
%\]

In general, repeated application of (\ref{eq:1D_Pull_1}) and (\ref{eq:1D_Pull_3}) on a simple function results in a linear combination of simple functions, where the number of simple functions increases exponentially in the order of the polynomial. \textcolor{black}{Namely, for a polynomial $p(g_1,g_2)$ of degree $L$, there are at most $J=\frac{2^{L+1}-1}{L-1}$ windows $\{\hf_j\}_{j=1}^J$ and signals $\{\hs_j\}_{j=1}^J$ such that 
$p(\bT_1,\bT_2) = \sum_{j=1}^J \hf_j\otimes\hs_j$.}

\textcolor{black}{The above approach can be used for implementing  multipliers $V_{\hf}^* p(\breve{N}_1,\breve{N}_2)V_{\hf}$ that have low order polynomial symbols $p$ directly in $\cW\otimes\cS$. %Here, we pull back phase space filters using the standard continuous wavelet transform  . 
Indeed, by (\ref{eq:1DPullWave0}), we have
\[V_{\hf}^* p(\breve{N}_1,\breve{N}_2)V_{\hf}=\sum_{j=1}^J\ip{\hf}{ \hf_j}\otimes\hs_j.\]
This implementation involves only computations on univariate functions. If windows and signals are discretized with $N$ samples in the frequency line, derivatives are approximated by finite differences, and $J\ll N$ is a constant, the implementation is of $O(N)$ computational complexity.} %For example, for $V_{\hf}^* \breve{N}_1V_{\hf}$, we have
%\[V_{\hf}^*\breve{N}_1 V_{\hf} \hs (\w)= V^*_{\hat{f}}V\big(\breve{T}_1 (\hf\otimes \hs)\big)(\w)=\ip{\hf}{\hf}_{\cW}i\frac{\partial}{\partial\w}\hs(\w) - \ip{\hf}{B\hf }_{\cW} \frac{\hs(\w)}{\w}\Big.\]

%Polynomial filters in $\breve{N}_2$, and bi-variate polynomials filters in $(\breve{N}_1,\breve{N}_2)$ can be computed similarly.

\subsubsection{Scale filters}
\label{Scale filters}

Scale filters $h(\breve{N}_2)$ have a simple formula in $\cW\otimes\cS$. Namely, under the slope change of variable, $h(\breve{N}_2)$ is a multiplicative operator that multiplies by the function of scale $h(z)=h(\frac{\w'}{\w})$, 
\textcolor{black}{
\begin{equation}
h(\breve{N}_2) F(\w',\w) = h(\frac{\w'}{\w}) F(\w',\w).
\label{eq:SF00}
\end{equation}
}
\textcolor{black}{This is shown in Appendix \ref{A:The 1D wavelet transform}.
Alternatively, if $h$ is a polynomial and $F=\hf\otimes\hs$ is a simple function, $h(\breve{N}_2)$ can be computed as explained in Subsection \ref{Polynomial filters}, by repeated applications of the formula (\ref{eq:1D_Pull_3}). In Figure \ref{fig:WP_filt} we visualize the  pull-back formula of time and scale filters.}

\subsubsection{Exponential time filters}

As a building block of general time filters, we derive a formula for $\exp(it\bT_1)$. The idea is to approximate general time filters as linear combinations of time exponentials $\exp(it\bT_1)$, using a Fourier expansion.
We show in Subsection \ref{A:The 1D wavelet transform} that time exponentials are translations in $\cW\otimes\cS$ along slope rays. Namely,
\[\exp(it\bT_1) F(z\w,\w) = F(z(\w-t),\w-t).\]
or equivalently 
\begin{equation}
\exp(it\bT_1) F(\w',\w) = F(\frac{\w'}{\w}(\w-t),\w-t).
\label{eq:1DslopeTrans}
\end{equation}

\subsubsection{General time filters}
\label{General time filters}

General time filters can be written as a Fourier expansion in $\exp(it\bT_1)$. As a result, a general time filter is a convolution along slope rays of the form
\begin{equation}
h(\bT_1) F(z\w,\w) = [F(z(\cdot),\cdot)*\hat{h}](\w).
\label{eq:1DslopeConv}
\end{equation}
To approximate (\ref{eq:1DslopeConv}) in practice, we translate $F$ along the slope via (\ref{eq:1DslopeTrans}) and consider the linear combination
\begin{equation}
h(\bT_1) F(z\w,\w) \approx \sum_j F\big(z(\w-t_j),\w-t_j\big)\hat{h}(t_j).
\label{eq:1DslopeConvD}
\end{equation}

\begin{figure}[!ht]
\centering 
\includegraphics[width=0.9\linewidth]{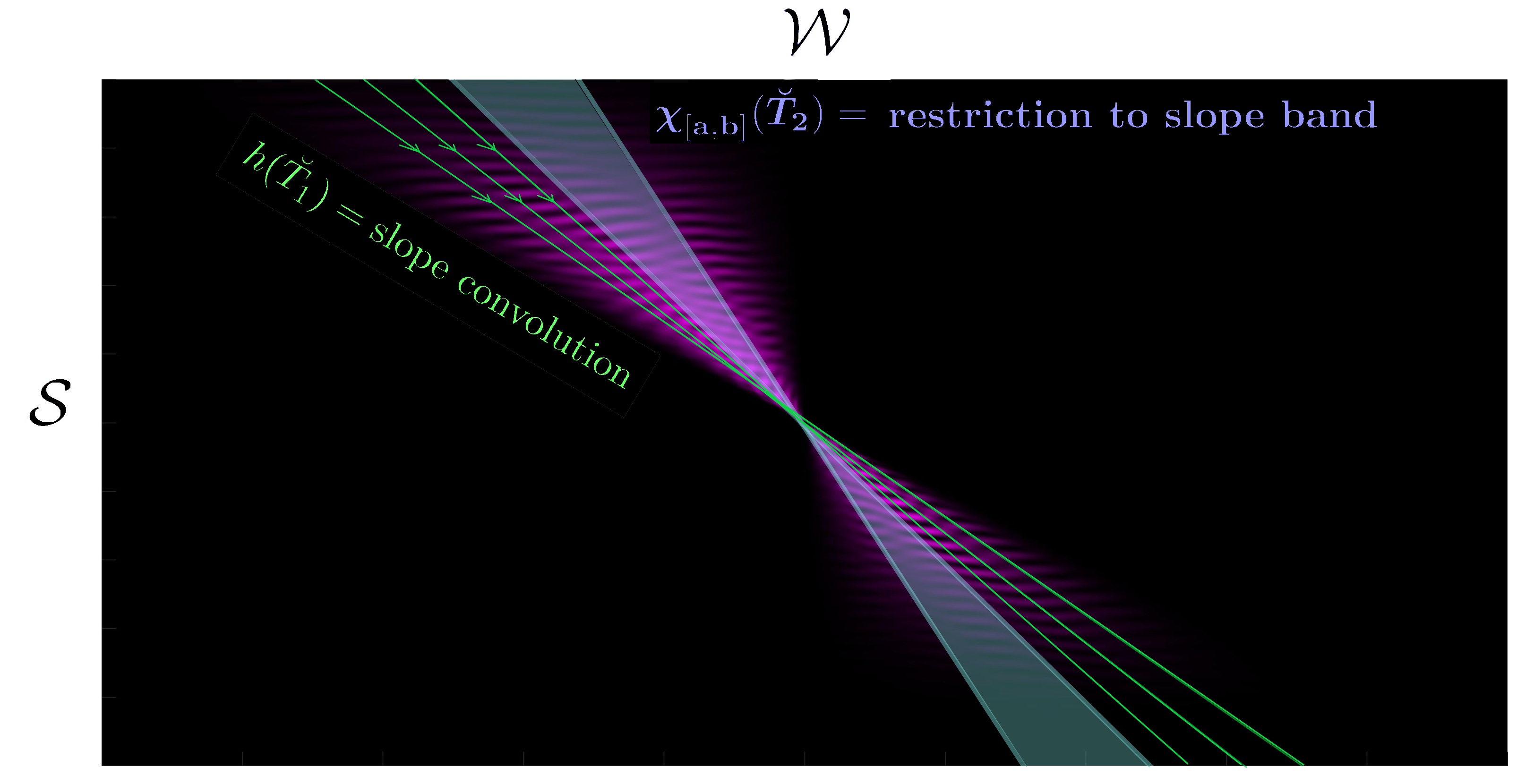}
\caption{Visualization of time and scale  filters. Magenta: a signal $F\in \cW\otimes\cS$. Green: visualization of time filters $h(\bT_1)$ as convolution operators along slope rays. Purple-blue: visualization of scale filters, with $h=\chi_{[a,b]}$ the characteristic function of the interval $[a,b]$. The operator $\chi_{[a,b]}(\bT_2)$ restricts $F$ to a band of slope rays.}
\label{fig:WP_filt} 
\end{figure}

\subsection{Efficient implementation of phase space filters}

At this point it is still not clear how the wavelet-Plancherel formulation of filters improves the computational complexity \textcolor{black}{of methods that compute $V_f[s]$ and apply phase space filters directly in $L^2(G)$}. Indeed, both $\cW\otimes\cS$ and $L^2(G)$ are function spaces over 2D domains. In this subsection we explain how to implement wavelet-Plancherel phase space filtering efficiently.

\subsubsection{Implementation for low order polynomial filters}
\label{Implementation for low order polynomial filters}

For low order polynomial filters $\breve{H}$ we discretize the window and the signal using $N$ samples in a regular grid in the frequency line. \textcolor{black}{We approximate derivatives with finite differences.} As explained in Subsection \ref{Polynomial filters}, applying the filter $V^*\breve{H}V$ on a simple function in $\cW\otimes\cS$, and representing the output as a sum of simple functions, takes $O(N 2^{M})$ operations, where $M$ is the order of the polynomial. Similarly, $V_{\hf}^*\breve{H}V_{\hf}$ takes $O(N 2^{M})$ operations. 

The term $2^M$ in the above complexity is due to the repeated application of discrete versions of formulas (\ref{eq:1D_Pull_1}) and (\ref{eq:1D_Pull_2}). However, $V^*\breve{H}V$ can be implemented more efficiently by deriving analytic expressions for the different monomials of $\bT_1$ and $\bT_2$.

Since both $\cW$ and $\cS$ are discretized to spaces of resolution $N$, the resolution of $\cW\otimes\cS$ is $N^2$.
A naive implementation of $\breve{H}$ in a discrete version of phase space with $N^2$ samples costs $O(N^2)$ operations, which is significantly higher than $O(N 2^{M})$ when the order $M$ is low. 

An implementation of $V_{\hf}^*\breve{H}V_{\hf}$ based on a discrete wavelet transform with a discretization of $O(N)$ points in $G$ costs $O(N)$ operations in a multi-resolution analysis implementation. In FFT based implementations, $V_{\hf}^*\breve{H}V_{\hf}$ takes $O(N\log(N))$ operations (see Subsection \ref{Standard discretization of the 1D wavelet transform}).

\subsubsection{Spline sequence discretization of the window-signal space}

It is possible to obtain a low computational complexity formulation of phase space filters using the following philosophy. Both spaces, $\cW\otimes\cS$ and $L^2(G)$, are function spaces over a 2D domain.  The window and signal \textcolor{black}{ information is mixed together via $V_{\hf}[\hs](g) = \ip{\hf}{\hat{\pi}(g)\hf}$ when represented in $L^2(G)$, and there is no reasonable discretization of $G$ that somehow separates the signal data from the window data. On the other hand, the window and signal data in $\cW\otimes\cS$ are more easily separated in view of the tensor product structure, and can hence be discretized separably.} This allows us to encode the window direction with only a few parameters, \textcolor{black}{since, in principle, windows are simple probes that measure the content of the signal, and need not be very complicated themselves}. We spend most of the resources encoding the signal direction, which contains all of the information.

In Subsection \ref{Discretization of the wavelet Plancherel method} we present the following formulation accurately.
We discretize the window direction using splines. Let us consider the special case of piece-wise linear continuous functions, also called \emph{linear splines}. \textcolor{black}{Linear splines are encoded by saving in memory the locations of the knots -- the points of discontinuity of the derivative --  and the values of the function at the knots.}  Such a spline can only be discontinuous at the first and last knots, if the value of the function is non-zero there.
The mother wavelet is defined as a linear spline with just a few knots. 

The signal direction is discretized finely, using $N$ samples on a regular grid in the frequency line. We discretize $\cW\otimes\cS$ as the tensor product of the spline discretization along $\cW$, and the grid-based discretization along $\cS$. Since each $\hs\in\cS$ is discretized as a sequence, and each $\hf\in\cW$ as a spline, the function $\hf\otimes\hs$ is discretized as a sequence of splines. Thus, general functions in the discrete $\cW\otimes\cS$ are sequences of splines.
We call this discretization of $\cW\otimes\cS$ the space of spline sequences, \textcolor{black}{and give a detailed construction in Subsection \ref{Efficient discretization of the window-signal space}.}

To be able to use the above discretization for phase space filtering it is important to show that filters of spline sequences are still spline sequences.

\subsubsection{Invariance of spline sequences under scale filters}

\textcolor{black}{By (\ref{eq:SF00}),} scale filters are multiplication by functions of slope. We consider examples in which applying the scale filter on the spline sequences gives another spline sequence. One example are characteristic functions of scale intervals, like the ones used for the search algorithm presented in Subsection \ref{Coefficient search in continuous wavelet systems}. \textcolor{black}{The invariance property is proven in Subsection \ref{Invariance of the discrete window-signal space to search operations}.}
General piecewise constant scale filters can be considered by extending the linear spline space to the space of piecewise linear functions which are not necessarily continuous.

\subsubsection{Invariance of spline sequences under time trigonometric polynomials}

By Subsection \ref{General time filters}, time filters are convolutions along slope rays. By (\ref{eq:1DslopeConvD}), time filters are computed by a linear combination of translated version of the spline sequence along slope rays. When applied on spline sequences, this results in a spline sequence, \textcolor{black}{as proven in Subsection \ref{Invariance of the discrete window-signal space to search operations}.}

\subsubsection{Phase space signal processing with spline sequences}
When the goal is to pull-back phase space filters via $V_{\hf}$, the projection upon the window $\hf$ is done by integrating all splines in the spline sequence along the window direction.
When the goal is to compute the norm of the filtered signal, like in the search algorithm of Subsection \ref{Coefficient search in continuous wavelet systems}, we compute the norm in $\cW\otimes\cS$, instead of $L^2(G)$, since $V$ is an isometric isomorphism.

%We show in Subsection \ref{Complexity of the WP4 algorithm} that the overall computations in the search algorithm is $O(N\log(N))$ operations.

\subsection{\textcolor{black}{Sparse approximations with continuous wavelet systems}}
\label{Sparse approximations with continuous wavelet systems}
 
It is well known that certain classes of signals can be effectively approximated using discrete wavelet systems, keeping only a sparse number of coefficients. For example, 1D wavelets are optimal for piecewise smooth functions \cite{Nonlinear_approximation}, and Shearlets are optimal for a class of 2D piecewise-smooth signals (cartoon-like images) \cite{optimal_Shearlet}.
In this paper we show how to compute sparse approximations to signals directly using the continuous wavelet system, instead of first discretizing it to a frame. 

We focus on greedy methods, like matching pursuit, where the atoms of the sparse approximation are chosen one by one, to decrease the approximation error as much as possible at each step. An important ingredient in a greedy algorithm is a coefficient search method, that extracts atoms having large coefficients. In discrete dictionaries we can simply compute all coefficients and extract the largest one. However, in continuous dictionaries, it is not possible to compute the whole continuum of coefficients. Using the wavelet-Plancherel theory, we develop a new efficient method for extracting large coefficients of continuous wavelet systems.

\subsubsection{Motivation for sparse continuous wavelet decompositions}

Continuous wavelet systems are advantageous over discrete wavelet system in situations where the wavelet transform is treated as a feature extraction method. For example, in the continuous 1D wavelet transform the atoms are treated as features at different times and scales, or as local frequencies. In some signal processing tasks, like phase vocoder (see Subsection \ref{Sampling resolution in phase space}), it is important to extract the local frequencies of a signal accurately. To compare the continuous and discrete methods, note that the continuous 1D wavelet transform is based on a continuous set of dilations and translations of a mother wavelet, whereas a discrete wavelet transform is based on a discrete set of translations and dilations. The spacing between different discrete wavelet atoms in the \textcolor{black}{time-frequency space} is \textcolor{black}{usually} big. This means that features in the signal with certain time-scale locations cannot be accurately pinpointed by the atoms of the discrete wavelet system. 
This effect is especially pronounced in high local frequencies, since discrete wavelet transforms are based on an exponential grid in the frequency direction (see, e.g., the dyadic wavelet transform \cite{Ten_lectures}). \textcolor{black}{In Subsection \ref{Standard discretization of the 1D wavelet transform} we detail this phenomenon.} Failing to extract the correct local frequencies of a signal degrades signal processing methods that manipulate these local frequencies, as is illustrated in Subsection \ref{Sampling resolution in phase space}.

%The above problem is avoided in continuous wavelet systems, where in the case of the 1D wavelet transform, features of any time-scale location can be represented.
%However, continuous wavelet systems consist of a continuum of atoms, and it is not clear how to search for large continuous wavelet coefficients.
%In this paper we introduce a framework in which it is possible to efficiently search for large continuous wavelet coefficients. We use this coefficient search method in a matching pursuit algorithm, to extract sparse continuous wavelet approximations to signals. 

\subsubsection{Coefficient search in continuous wavelet systems}
\label{Coefficient search in continuous wavelet systems}

%In the following we study the continuous 1D wavelet transform, and call the wavelet transform $V_f[s]$ of a signal $s$ its phase space representation. 
Given a signal $\hs\in L^2(0,\infty)$, in principle, if the whole function $V_{\hf}[\hs]$ was given to us, we could have searched for big coefficients directly in phase space using a bisection method as follows. We consider a big rectangle $G_1$ in phase space, that contains most of the energy of $V_{\hf}[\hs]$. We partition phase space into four sub rectangles $G_1^1,\ldots,G_1^4$. For each $G^k_1$, we consider the restriction $V_{\hf}[\hs]\big|_{G^k_1}$, padded by zero outside $G^k_1$, and denote it by $P_{G_1^k}V_{\hf}[\hs]$. The projection operator $P_{G_1^k}$ is called a  band-pass filter. We compute each of the norms $\norm{P_{G_1^k}V_{\hf}[\hs]}$, $k=1,\ldots,4$,  and denote by $G_2=G^k_1$ the argmax. By repeating this process, each time on the subdomain with the greatest norm, we can pinpoint big wavelet coefficients in logarithmic time with respect to the resolution in phase space. The resulting coefficient might not be the global maximum. However, the method is also not a local method, as all of the wavelet coefficients are accounted for in the search. Thus, we interpret this method as ``somewhere between a global and a local method''.

Clearly this approach is not practical, since computing $\norm{P_{G_1^k}V_{\hf}[\hs]}$ requires knowing a continuum of wavelet coefficients, which we want to avoid. If it was possible to pull back the computation of $P_{G_1^k}V_{\hf}[\hs]$ from phase space to the signal domain $\cS$, the computations in the search could have been carried out entirely on the signal domain data, using the fact that $V_{\hf}$ is an isometry. \textcolor{black}{This of course is impossible, since the wavelet transform is not an isomorphism from the signal domain $L^2(0,\infty)$ to a reasonable space of functions on phase space. More accurately, %band-passed phase space representations are in general not in the image space of the wavelet transform. Namely, 
 in the generic case, there is no signal $\hs_k$ with phase space representation $V_{\hf}[\hs_k]=P_{G_1^k}V_{\hf}[\hs]$.}

We hence reformulate the problem in the wavelet-Plancherel setting, were the concept of pulling back the operator $P_{G_1^k}$ from $L^2(G)$ to $\cW\otimes\cS$ is well defined. The resulting search algorithm is called wavelet Plancherel phase-space projection pursuit, or WP$^4$. See Figure \ref{fig:WP4}  for visualization of the bisection method.

\begin{figure}[!ht]
\centering 
\includegraphics[width=0.3\linewidth]{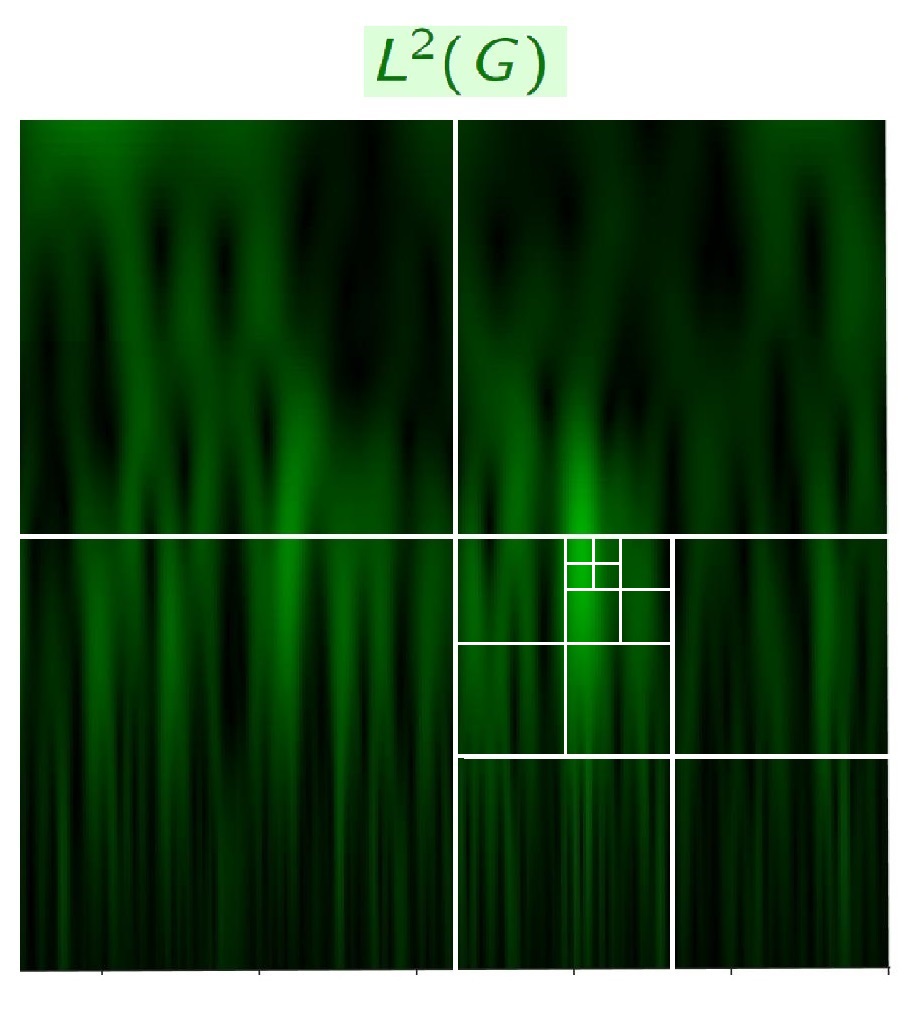}
\caption{Visualization of the WP4 bisection search method. At each step, the square with the largest norm of $V(\hf\otimes\hs)$ in phase space is chosen, until converging to a large wavelet coefficient.}
\label{fig:WP4} 
\end{figure}

\subsubsection{Coefficient search via the wavelet-Plancherel theory}
\label{Coefficient search via a wavelet-Plancherel theory}

%\subsubsection{Computational complexity of the coefficient search algorithm}

Given a signal $\hs\in\cS$ and a window $\hf\in\cW$, we pull-back the norm computations in the search algorithm to a window-signal formulation by
\[\norm{ V^*P_{G_n}V (\hf\otimes \hs)}_{\cW\otimes\cS}=\norm{P_{G_n}V_{\hf}[\hs]}_{L^2(G)},\]
where $G_n\subset G$ is the domain at step $n$ of the algorithm. 
We call scale filters based on characteristic functions of intervals \emph{scale-pass filters}. We similarly define \emph{time-pass} filters.
We compute $P_{G_n}$ as a composition of a scale-pass filter with a time-pass filter.

We discretize $\cW\otimes\cS$ as spline sequences, and the WP4 search algorithm is implemented by repeated applications of scale-pass and time-pass filters. We approximate the time-pass filters by trigonometric polynomials, and scale-pass filters are implemented directly. This is explained in more detail in Section \ref{Greedy sparse continuous wavelet transform}. 
\textcolor{black}{We present the pseudo-code of WP4 in Algorithm \ref{Wavelet Plancherel phase-space projection pursuit (WP4)}, and give the implementation details in Appendix \ref{Implementation of the Wavelet-Plancherel search algorithm}.}
In Subsection \ref{Complexity of the WP4 algorithm} we show that the overall complexity of one coefficient search is $O(N\log N)$ operations for signals of resolution $N$. Observe that even though the coefficient search in phase space locates large wavelet coefficients in a resolution of $N\times N$ pixels in $G$, the search algorithm takes only $O(N\log N)$ operations.

Note that WP4 has lower computational complexity than naive grid-based coefficient search methods that utilize FFT. Specifically, for a discretization of time and scale based on $N$ samples each, the naive method takes $O(N^2\log(N))$ operations (see Subsection \ref{Standard discretization of the 1D wavelet transform}), while our method takes $O(N\log(N))$ operations. Our lower computational complexity is comparable to discrete methods based on painless reconstruction wavelet frames \cite{Painless}, which take $O(N\log(N))$ operations. As opposed to the discrete method, which implies a resolution of $O(N)$ in phase space, our method localizes an atom in a resolution of $O(N^2)$ in phase space. %For this reason, our method is more suitable for time-scale feature extraction.

\section{Preliminaries}
\label{Preliminaries}

\textcolor{black}{In this section we recall basic facts from representation theory  and functional calculus, which are the building blocks for the wavelet-Plancherel theory. For more details on representation theory see, e.g., \cite{repre_HA}, and for functional calculus in Hilbert spaces \cite{Spectral_Hilbert}.}

Let $\SS$ be $\RR,\ZZ,e^{i\RR}$ or $e^{2\pi i \ZZ/N}$  for some $N\in\NN$, and for each case of $\SS$ let $\hat{\SS}$ be respectively $\RR,e^{i\RR},\ZZ$ or $e^{2\pi i \ZZ/N}$. We denote the Fourier transform $L^2(\SS)\rightarrow L^2(\hat{\SS})$ by $\cF$, and for $f\in L^2(\SS)$ we denote $\hf=\cF f$. The convention is that the normalization in the  Fourier transform is in the exponent, e.g., $e^{2\pi i \w x}$, and for $f\in L^2(\RR)\cap L^1(\RR)$
\[[\cF f](\w) = \int_{\RR} f(x)e^{-2\pi i \w x}dx. \]

\subsection{Representations of topological groups}

The function system of a wavelet transform is constructed by applying a group representation on a basic function, called a window. %This yields a family of transformed versions of the window. 
In the following we review some definitions from representation theory.
For a linear vector space $V$, the group of invertible linear operator $V\rightarrow V$, with composition as the group product, is denoted by ${\cal GL}(V)$.
A \emph{representation} $\pi$ of a group $G$ in the linear space $V$ is a homomorphism $\pi:G\rightarrow {\cal GL}(V)$. A subspace $W$ of $V$ is called \emph{invariant} under the representation $\pi$ if for every $v\in W$ and $g\in G$, $\pi(g)v\in W$. 
A representation $\pi$ on the vector space $V$ is called \emph{irreducible} if the only invariant subspaces of $V$ under $\pi$ are $\{0\}$ and $V$. For a Hilbert space $\cH$, a \emph{unitary representation} of $G$ is a homomorphism $\pi$ from $G$ to the subgroup ${\cal U}(\cH)\subset{\cal GL}(\cH)$ consisting of unitary operators. 
A \emph{topological group} is a group $G$ which is also a topological space, such that the group product and inversion are continuous maps (as mappings $G\times G\rightarrow G$ and $G\rightarrow G$ respectively). A unitary representation $\pi$ of a topological group $G$ is called \emph{strongly continuous} if it is continuous from $G$ to ${\cal U}(\cH)$ endowed with the strong topology.

Any locally compact Hausdorff topological group has a unique Radon measure, up to constant, that is invariant under left translations. Namely, there is a Radon measure $d\mu(g)$ in $G$ such that for every $g'\in G$ and every measurable function $F:G\rightarrow\RR_+$ the following equality holds
\[\int_G F(g)d\mu(g) = \int_G F(g'g)d\mu(g).\]
This measure is called the \emph{left Haar measure} of $G$.

\subsection{Functional calculus}
\label{Functional calculus0}

Functional calculus is the theory of applying functions on unitary or self-adjoint operators. In the finite dimensional case, normal operators have an eigen-decomposition, and applying a function on the operator is defined by applying the function on the eigenvalues, keeping the eigenvectors intact. In the infinite dimensional case, there is an extension of the notion of an eigen-decomposition, namely the spectral theorem.

A self-adjoint or unitary operator $T$ in an infinite dimensional separable Hilbert space $\cH$ does not admit an eigen-decomposition in general. Instead, it admits a more subtle notion of spectral decomposition, called a projection-valued Borel measure, or PVM. In case $T$ is self-adjoint, its spectrum is a subset of $\RR$, and in case $T$ is unitary, its spectrum is a subset of $e^{i\RR}$. \textcolor{black}{We denote $\SS=\RR$ if $T$ is self-adjoint, and $\SS=e^{i\RR}$ if $T$ is unitary. Hence, $\SS$ contains the spectrum of $T$.}
Informally, the PVM corresponding to $T$ is a mapping $P$, that maps sets $B\subset\SS$ to the projection $P(B)$ upon the ``eigenspace'' corresponding to the eigenvalues in $B$.  For a precise formulation see Definition \ref{defi_PVM} in Appendix A. 

The spectral theorem links $T$ with their PVMs. The theorem states that given $T$, there is a PVM $P$
such that
	\begin{equation}
T=\int_{\SS}\lambda \ dP(\lambda).
\label{eq:spectral_theorem0}
\end{equation}
The integral in (\ref{eq:spectral_theorem0}) is defined informally as
\[\lim_{ \scriptsize
\begin{array}{c}
	{\rm diam}({\bf x})\rightarrow 0 \cr
	[x_0,x_n]\rightarrow \SS
\end{array}
} \sum_{k=0}^{n-1} x_k P\big([x_k,x_{k+1})\big).  \]
Here, ${\bf x}=\{x_0,\ldots,x_n\}$ denotes a Riemann partition of $\SS$, and ${\rm diam}({\bf x})$ denote the maximal diameter of intervals in ${\bf x}$. For a precise formulation see Theorem \ref{spectral theorem}. The projections $P(B)$ are also called the spectral projections of $T$.

In analogy to the application of functions on finite dimensional normal operator, a measurable function $\phi:\SS\rightarrow\CC$ of a self-adjoint or unitary operator $T$ is defined to be the normal operator
\begin{equation}
\phi(T)=\int_{\SS}\phi(\lambda) dP(\lambda).
\label{eq:functional_calculus11}
\end{equation}
See Remark \ref{Band_limit_poly} of Appendix A for more details. 

Let $T$ be a self-andjoint or unitary operator in the Hilbert space $\cH$, with PVM $P$.  Let $U$ be an isometric isomorphism from the Hilbert space $\cS$ to $\cH$. We call $U^*TU$ the pull-back of $T$ to $\cS$. The spectral decomposition of $U^*TU$ is given by the PVM $U^*P(\cdot)U$, namely
\begin{equation}
U^*TU =\int_{\SS}\lambda \ dU^*P(\lambda)U.
\label{eq:pull_spectral}
\end{equation} 
Intuitively, in the conjugation of $T=\int_{\SS}\lambda \ dP(\lambda)$ with $U$, by ``linearity of the integral'', the conjugation is applied on each ``infinitesimal element'' $dP(\lambda)$. 
As a result, for a measurable function $\phi:\SS\rightarrow\CC$, we have
\begin{equation}
\phi(U^*TU) =U^*\phi(T)U.
\label{eq:pull_func_calc}
\end{equation}

\section{Continuous wavelet transforms}
\label{Continuous wavelet transforms}

In this section we review basic facts and definitions from the general theory of continuous wavelet transforms. We then review the class of semi-direct product wavelet transforms \cite{MyRef}, and introduce the class of geometric wavelet transforms.

\subsection{General continuous wavelet transforms}

General wavelet transform frameworks generalize the two classical examples of the STFT and the continuous 1D wavelet transform. The classical general formulation of general wavelet transforms was developed in \cite{gmp0}\cite{gmp}, where a phase space $G$ is a locally compact topological group, and $\pi$ is a square integrable representation.

\begin{definition}
\label{GCWT}
Let $G$ be a locally compact topological group, and consider the Hilbert space $L^2(G)$ based on the left Haar measure of $G$. Let $\pi$ be an irreducible strongly continuous unitary representation of $G$ in a \textcolor{black}{separable} Hilbert space of signals $\cS$. 
An element $f\in\cS$ such that the function $G\ni g\mapsto \ip{f}{\pi(g)f}$ is in $L^2(G)$, is called a \emph{window}. Given a window $f$, the mapping $V_f:\cS\rightarrow L^{\infty}(G)$ defined for $s\in\cS$ and $g\in G$ by
\begin{equation}
V_f[s](g)=\ip{s}{\pi(g)f}
\label{eq:voice1}
\end{equation}
is called the \emph{wavelet transform} based on the window $f$.
\end{definition}

An irreducible strongly continuous unitary representation, for which there exists a window $f$, is also called a \emph{square integrable representation}.
The following properties of the wavelet transform can be found in \cite{gmp}.

\begin{proposition}
\label{prop_wavelet1}
Let $f_1,f_2\in\cS$ be windows, and let $s_1,s_2\in\cS$.
\begin{enumerate}
	\item 
	The wavelet transform $V_{f_1}$ is a scalar times an isometric embedding of $\cS$ into $L^2(G)$.
	\item
	There exists a unique (up to multiplication by scalar), densely defined positive (self-adjoint) operator $K$ in $\cS$, with densely defined inverse, called the \emph{Duflo-Moore operator}, such that the domain of $K$ is the set of windows, and
	\begin{equation}
	\ip{V_{f_1}[s_1]}{V_{f_2}[s_2]}_{L^2(G)} = \ip{s_1}{s_2}_{\cS}\ip{Kf_2}{Kf_1}_{\cS}.
	\label{eq:ortho_reco}
	\end{equation}
\end{enumerate}
\end{proposition}

\begin{remark}
\label{remark:GCWT_reco}
$ $
\begin{enumerate}
	\item 
	\label{reconstruction1}
Equation (\ref{eq:ortho_reco}) can be read of as a weak reconstruction formula. Namely, by taking $s_1=s$ and two windows $f_1,f_2$, against an arbitrary $s_2$, we get
\begin{equation}
s=\frac{1}{\ip{Kf_2}{Kf_1}}\int_G V_{f_1}[s]\pi(g)f_2 \ d\mu(g),
\label{eq:inversionR}
\end{equation}
where \ref{eq:inversionR} is a weak vector integral, also called Pettis integral \cite{Weak_Integral}.
\item
By taking $f_1=f_2=f$, equation (\ref{eq:ortho_reco}) shows that for any window $f$, $V_f$ is an isometric embedding of $\cS$ into $L^2(G)$, up to a global normalization that depends on $f$.
\item
The wavelet transform $V_{f}$ is also called the \emph{analysis operator} corresponding to the window $f$. The adjoint $V_{f}^*$ is called the \emph{synthesis operator} corresponding to $f$. For $F\in L^2(G)$, we have
\begin{equation}
V_{f}^*(F) = \int F(g) \pi(g)f d\mu(g)
\label{eq:wave_inversion1}
\end{equation}
where the integral is defined in the weak sense as in (\ref{eq:inversionR}).
The reconstruction formula (\ref{eq:inversionR}) can be written in the form
$s=\frac{1}{\ip{K f_2}{K f_1}}V^*_{f_2}V_{f_1}[s]$.
\end{enumerate}
\end{remark}

The convolution of two $L^2(G)$ functions $F$ and $Q$, is defined by
	\[[F*Q](g)= \int F(q^{-1} g)Q(q)d\mu(q).\]
	The convolution $F*Q$ is interpreted as ``filtering'', or ``blurring,'' $Q$ using the kernel $F$.
 Given a window $f\in \cS$, its ambiguity function is defined to be $V_f[f]$. Consider the image space of the wavelet transform, $V_f[\cS]= \{V_f[s]\ |\ s\in\cS\}$.
The following two propositions show that $V_f[\cS]$ is a reproducing kernel Hilbert space, with the translations of the ambiguity function $V_f[f]$ as the reproducing kernels (see, e.g., \cite{Fuhr_wavelet}).	
	
\begin{proposition}
\label{Amb_is_kernel}
Let $f\in\cS$ be a window with $\norm{Kf}=1$, where $K$ is the Duflo-Moore operator. Consider the image space of the wavelet transform, $V_f[\cS]$. Then
	$V_fV_f^*$ is the orthogonal projection $L^2(G)\rightarrow V_f[\cS]$. Moreover, for any $Q\in L^2(G)$, $V_fV_f^*[Q] =V_f[f]*Q$.
\end{proposition}

The following is a result of Proposition \ref{Amb_is_kernel}.
	\begin{corollary}
	\label{Amb_is_kernel2}
	Let $f\in\cS$ be a window with $\norm{K f}=1$. Then the image space of the wavelet transform, $V_f[\cS]$, is a reproducing kernel Hilbert space with kernel $V_f[f]$. Precisely, for any $Q\in V_f[\cS]$, $Q= V_f[f]*Q$.
	\end{corollary}
	
In view of Corollary \ref{Amb_is_kernel2}, the space $V_f[\cS]$ is \textcolor{black}{interpreted as} ``blurry,'' \textcolor{black}{since we interpret convolution as blurring, and any $Q\in V_f[\cS]$ is the blurring of some function in $L^2(G)$}. The ambiguity function $V_f[f]$ is interpreted as the ``point spread function'' of $V_f[\cS]$, posing a lower bound on the resolution of $V_f[\cS]$ functions.

\subsection{Examples of general continuous wavelet transforms}
\label{Examples of general continuous wavelet transforms}

\subsubsection{The short time Fourier transform} 

For references on the short time Fourier transform see  {\color{black}e.g.} \cite{Time_freq}.
Consider the signal space $L^2(\RR)$. Let $L:\RR\rightarrow {\cal U}(L^2(\RR))$ be the translation in $L^2(\RR)$. Namely, for $g_1\in\RR$ and $f\in L^2(\RR)$, $[L(g_1)f](t)=f(t-g_1)$. Let $M:\RR\rightarrow {\cal U}(L^2(\RR))$ be the modulation in $L^2(\RR)$. Namely, for $g_2\in\RR$ and $f\in L^2(\RR)$, $[M(g_1)f](t)=e^{ig_2 t}f(t)$. Note that $L$ and $M$ are representations of the group $\{\RR,+\}$. We have the commutation relation
\begin{equation}
[L(g_1), M(g_2)]:=L(g_1)^*M(g_2)^*L(g_1)M(g_2) = e^{-ig_2 t} I.
\label{eq:Heisy_comm}
\end{equation}
where $I$ is the identity operator. This shows that the set of unitary operators
\[J=\{g_3 L(g_1)M(g_2)\ |\ g_3\in e^{i\RR}, g_1,g_2\in\RR\}\]
is a group, with composition as the group product. We can treat $J$ as a group of tuples $\RR\times\RR\times e^{i\RR}$, with group product derived from (\ref{eq:Heisy_comm}). The group $J$ is called the {  \emph{(reduced) Heisenberg group}}. The mapping 
\[\pi(g_1,g_2,g_3)=g_3 L(g_1)M(g_2)\]
is a square integrable representation \textcolor{black}{called the Schr\"odinger representation \cite[Example 2.27]{Fuhr_wavelet}}, with Dulfo-Moore operator $K=I$. The resulting wavelet transform is called the   \emph{short time Fourier transform (STFT)}. \textcolor{black}{We note that, typically, the short-time-Fourier transform is defined, for some window $f\in L^2(\RR)$, as the restriction of $V_f[s]$ to the subset in phase space $J_0=\{(g_1,g_2,0)\ |\ g_1,g_2\in\RR\}$. Since $g_1$ acts on signals by multiplication,  $V_f\big|_{J_0}: L^2(\RR)\rightarrow L^2(J_0)$ is also an isometry (see \cite[Example 2.27]{Fuhr_wavelet}). We ignore this subtlety in this paper and call the transform $V_f$, based on the full reduced Heisenberg group $J$, the STFT.}

\subsubsection{The 1D wavelet transform} 
\label{The 1D wavelet transform}

For references on the 1D wavelet transform see   \color{black}e.g. \cite{Cont_wavelet_original,Ten_lectures}.
Consider the signal space $L^2(\RR)$, and the translation $L$. Let $D:\RR\rightarrow {\cal U}(L^2(\RR))$ be the dilation in $L^2(\RR)$. Namely, for $g_2\in\RR$ and $f\in L^2(\RR)$, $[D(g_2)f](t)=e^{-\frac{1}{2}g_2}f(e^{-g_2}t)$. Consider the reflection $R:\{1,-1\}\rightarrow {\cal U}(L^2(\RR))$, namely $[R(g_3)f](t)=f(g_3 t)$. The mappings $D,R$ are representations of the groups $\{\RR,+\},\{\{1,-1\},\cdot\}$ respectively. The set of transformations
\begin{equation}
A=\{ L(g_1)D(g_2)R(g_3)\ |\ g_3\in \{1,-1\}, g_1,g_2\in\RR\}
\label{eq:Affine_group1}
\end{equation}
is closed under compositions. We can treat $A$ as a group of tuples $\RR\times\RR\times \{1,-1\}$, with group product derived from  the compositions of operators in (\ref{eq:Affine_group1}). The group $A$ is called the {  \emph{1D affine group}}. The mapping 
\begin{equation}
\pi(g_1,g_2,g_3)=L(g_1)D(g_2)R(g_3)
\label{eq:wave_rep}
\end{equation}
is a square integrable representation, with Dulfo-Moore operator $K$ defined by
\[ [\cF K \cF^*\hf](\w) = \frac{1}{\sqrt{\abs{\w}}}\hf(\w).\]
 The resulting wavelet transform is called the {  \emph{continuous 1D wavelet transform}}. In this paper we simply call it the {  \emph{wavelet transform}}. \textcolor{black}{We note that the standard way to parametrize the 1D affine group, for example in \cite{Ten_lectures}, is with two parameters: time $\alpha=g_1$ and scale $\beta=g_3 e^{g_2}$. With this, the 1D wavelet representation takes the form
 \[ \pi'(\alpha,\beta)f(x) = \frac{1}{\sqrt{\abs{\beta}}}f\Big(\frac{x-\alpha}{{\beta}}\Big).\]}

\subsubsection{The rotation-dilation wavelet transform}

For references on the  rotation-dilation wavelet transform see {\color{black}e.g.} \cite{geometric}.
Consider the signal space $L^2(\RR^2)$, and the translations $L_x, L_y:\RR\rightarrow {\cal U}(L^2(\RR^2))$, along the $x$ and $y$ directions,
\[[L_x(g_1^1)f](x,y) = f(x-g_1^1,y) \quad , \quad [L_y(g_1^2)f](x,y) = f(x,y-g_1^2).\] 
 Let $R:e^{i\RR}\rightarrow {\cal U}(L^2(\RR^2))$, be the rotation, namely 
{  
\[[R(g_2)f](x,y) = f\Big(\left(
\begin{array}{cc}
	{\rm Re}(g_2) & -{\rm Im}(g_2)\\
	{\rm Im}(g_2) & {\rm Re}(g_2)
\end{array}
\right)^{-1} \left(\begin{array}{c}
	x \\
	y
\end{array}\right)
\Big).\] 
}
Let $D:\RR\rightarrow {\cal U}(L^2(\RR^2))$ be the (isotropic) dilation, $[D(g_3)f](x,y)=e^{-\frac{1}{2}g_3}f(e^{-g_3}x,e^{-g_3}y)$. 
The mappings $L_x,L_y,D$ are representations of the groups $\{\RR,+\}$, and $R$ is a representation of $\{e^{i\RR}, \cdot\}$. The set of transformations
\begin{equation}
A_R=\{ L_x(g^1_1)L_y(g^2_1)R(g_2)D(g_3)\ |\ g_2\in e^{i\RR}, g^1_1,g_1^2,g_3\in\RR\}
\label{eq:Affine_group2}
\end{equation}
is closed under compositions. We can treat $A_R$ as a group of tuples $\RR^3\times e^{i\RR}$, with group product derived from  the compositions of operators in (\ref{eq:Affine_group2}). The group $A_R$ is called the {  \emph{rotation-dilation affine group}}. The mapping 
\[\pi(g^1_1,g_1^2,g_2,g_3)=L_x(g^1_1)L_y(g^2_1)R(g_2)D(g_3)\]
is a square integrable representation. The corresponding Dulf-Moore operator is given by \cite{geometric}
\[[\cF K\cF^*\hf](\w_1,\w_2) = \frac{1}{\abs{(\w_1,\w_2)}}\hf(\w_1,\w_2).\]
We call in this paper the resulting wavelet transform the {  \emph{rotation-dilation wavelet transform}}.

\subsubsection{The Shearlet transform} 

For references on the Shearlet transform see   {\color{black}e.g.} \cite{Shearlet}.
Consider the signal space $L^2(\RR^2)$, and the translations $L_x, L_y$.
 Let $S:\RR\rightarrow {\cal U}(L^2(\RR^2))$, be the shear, namely 
{  
\[[S(g_2)f](x,y) = f\Big(\left(
\begin{array}{cc}
	1 & g_2\\
	0 & 1
\end{array}
\right)^{-1} \left(\begin{array}{c}
	x \\
	y
\end{array}\right)
\Big).\] 
}
Let $D:\RR\rightarrow {\cal U}(L^2(\RR^2))$ be the anisotropic dilation, 
\[[D(g_3)f](x,y)=e^{-\frac{3}{4}g_3}f(e^{-g_3}x,e^{-\frac{1}{2}g_3}y).\]
Last, let $R:\{1,-1\}\rightarrow {\cal U}(L^2(\RR^2))$ be the reflection $[R(g_4)f](x,y)=f(g_4 x, g_4 y)$.
The mappings $L_x,L_y,S,D$ are representations of the groups $\{\RR,+\}$, and $R$ is a representation of $\big\{\{1,-1\}, \cdot\big\}$. The set of transformations
\begin{equation}
A_S=\{ L_x(g^1_1)L_y(g^2_1)S(g_2)D(g_3)R(g_4)\ |\ g_4\in \{1,-1\}, g^1_1,g_1^2,g_2,g_3\in\RR\}
\label{eq:Affine_group3}
\end{equation}
is closed under compositions. We can treat $A_S$ as a group of tuples $\RR^4\times \{1,-1\}$, with group product derived from  the compositions of operators in (\ref{eq:Affine_group3}). The group $A_S$ is called the {  \emph{Shearlet group}}. The mapping 
\[\pi(g^1_1,g_1^2,g_2,g_3,g_4)=L_x(g^1_1)L_y(g^2_1)S(g_2)D(g_3)R(g_4)\]
is a square integrable representation. The corresponding Duflo-Moore operator is given by \cite{Affine_uncertainty2}
\[[\cF K\cF^*\hf](\w_1,\w_2)=\frac{1}{\abs{\w_1}}\hf(\w_1,\w_2).\]
The resulting wavelet transform is called the {  \emph{Shearlet transform}}. Note that this construction of the Shearlet transform, differs slightly from the standard construction, in that it models dilations by multiplying the variable by $e^{g_3}$, instead of by $g_3$, as in the standard formulation.

\subsection{Semi-direct product wavelet transforms}

\textcolor{black}{In the above examples of wavelet transforms, the phase space has a physical interpretation. For example, in the STFT, $G$ is the time-frequency plane, and in the 1D wavelet transform it is the time-scale space. There is a systematic approach for formulating general phase spaces with physical interpretations \cite{Adjoint,MyRef}. A general class of wavelet transforms for which this is possible is called semi-direct product wavelet transforms \cite{MyRef}. The physical interpretation of phase space is given by an explicit construction, which we use later on to develop the pull-back formulas of phase space filters. We next recall the definition of semi-direct product wavelet transforms, starting from their motivation. For more details see \cite{MyRef}.}

%The examples from Subsection \ref{Examples of general continuous wavelet transforms} are instances of semi-direct product wavelet transforms \cite{MyRef}.  %, which we recall in this subsection. Denote for each of the 1D examples 
%\[\pi(g) = \pi_1(g_1)\ldots\pi_M(g_m)\]
% the definition of $\pi$ as a composition of one parameter representations. For the 2D examples, denote ${\bf g}_1=(g_1^1,g_1^2)$, and $\pi(g) = \pi_1({\bf g}_1)\pi_2(g_2)\ldots\pi_M(g_m)$.
\textcolor{black}{In each example from Subsection \ref{Examples of general continuous wavelet transforms}, the general wavelet transform is interpreted as a procedure of measuring a set of ``physical quantities'' on the signal. 
%Let us illustrate this point of view on the 1D wavelet transform, which measures the content of signals at different times and scales. 
In the 1D wavelet transform, the window $f$ is interpreted as a time-scale atom, sitting at time zero and scale zero. The dilation operator $D(g_2)$ is interpreted as an operation that increases (or decreases) scale by $g_2$, and thus $D(g_2)f$ is a time-scale atom at the time-scale pair $(0,g_2)$. Similarly, $L(g_1)$ increases (or decreases) time by $g_1$, and $\pi(g)f=L(g_1)D(g_2)f$ is a time-scale atom at the time-scale pair $(g_1,g_2)$. Last, the inner product $V_f[s](g)=\ip{s}{\pi(g)f}$ is interpreted as a probing of the content of the signal $s$ at the time-scale pair $(g_1,g_2)$.  Similarly, for the STFT, $V_f[s](g_1,g_2,g_3)$ is interpreted as the content of $s$ at the time-frequency pair $(g_1,g_2)$, for the rotation-dilation wavelet, $V_f[s](g_1^1,g_1^2,g_2,g_3)$ is the content of the image $s$ at the position-angle-scale triplet $(g^1_1,g_1^2,g_2,g_3)$, and for the Shearlet transform, $V_f[s](g_1^1,g_1^2,g_2,g_3,g_4)$ is the content of $s$ at the position-slope-scale triplet $(g_1^1,g_1^2,g_2,g_3)$.}

%In general, there is a systematic approach under which it is possible to interpret a general wavelet transform as the measurement of a signal at some tuple of physical quantities \cite{MyRef}.
For the general treatment, we assume that the elements of $G$ are tuples of numerical values, with a certain structure to the way different tuples interact.
We begin by recalling the definition of physical quantities \cite{MyRef}.

\begin{definition}
A \emph{physical quantity} is one of the following numerical Lie groups
\[\{\RR,+\} \quad,\quad \{e^{i\RR},\cdot\} \quad,\quad \{\ZZ,+\} \quad,\quad \{e^{2\pi i \ZZ/N},\cdot\} \quad (N\in\NN).\]
\end{definition}

%\textcolor{black}{The group rule of $G$ describes how the different physical quantities interact.
%In semi-direct product wavelet transforms, the group rule is based on a semidirect product.}

%The group $G$ is the set of tuples of the physical quantities underlying the corresponding wavelet transform. The group structure of $G$ encapsulates how different values of the physical quantities interact. Namely, the group product $g'\bullet g$ describes how the values of the physical quantities in $g$ change, when transformed by the amounts of the physical quntities in $g'$. To define the interaction between different group elements, we rely on the notion of semidirect product.

\subsubsection{Semi-direct product groups of physical quantities}

A topological group $G$ is called a \emph{semi-direct product} of a normal subgroup $N\triangleleft G$ and a subgroup $H\subset G$, if $G=NH$ and $N\cap H=\{e\}$ \textcolor{black}{\cite{abs_alg}}. This is denoted by $G=N\rtimes H$. If $G=N\rtimes H$, then each element $g\in G$ can be written in a unique way as $nh$ where $n\in N$, $h\in H$. Thus we can identify elements of $G$ with ordered pairs, or coordinates $(n,h)\in N\times H$. In the coordinate representation, the group multiplication takes the form
\[(n,h)(n',h')\sim nhn'h'=n\ hn'h^{-1}\ hh'\sim (n\ hn'h^{-1},hh').\]
\textcolor{black}{Here, $\sim$ denotes the relation between points in $G$ and their coordinate representations.}
Since $N$ is a normal subgroup, the element
\begin{equation}
  A_h(n'):=hn'h^{-1} 
  \label{SP_A0}
\end{equation}
 is in $N$. Moreover, for a Lie group $G$, $A_h$ is a smooth group action of $H$ on $N$, and a smooth automorphism of $N$ for each $h\in H$.

When $N,H$ are isomorphic to physical quantities, $G=N\times H$ is interpreted as the group of ordered pairs of the physical quantities $quantity_1,quantity_2$. Each coordinate of $G$ is interpreted as the ``physical dimension'' of the corresponding physical quantity, and the value at this coordinate corresponds to the value of the physical quantity.
Thus, the semidirect product structure allows us to make the following philosophical argument: ``physical quantities may change their values when transformed, but they retain their dimensions.'' Indeed, multiplying a group element $g$ of $G$ with another, may change the values of the coordinates of $g$, but may not change the ordered pair structure itself.

\begin{example}
In the case of the affine group, we have $G=N\rtimes H$, where $N\cong \{\RR,+\}$ is the subgroup of translations and $H\cong\{\RR,+\}\times\{-1,1\}$ is the subgroup of dilations and reflections. The group product takes the following form in coordinates
\[(n,h_1,h_2)\bullet(n',h'_1,h_2')=(n+h_2e^{h_1}n',h_1+h_1',h_2h_2').\]
Namely, $A_{(h_1,h_2)}(n')=h_2e^{h_1}n'$.
\end{example}

\subsubsection{General semi-direct product wavelet transforms}

In the systematic approach for treating general wavelet transforms as measurements of physical quantities, $G$ is a nested semi-direct product of physical quantities.
The following definition is taken from \cite{MyRef} and consists of a list of assumptions on the semi-direct product structure. On the one hand, these assumptions allow  a systematic derivation of a general wavelet-Plancherel phase space filtering theory. Particularly, the assumption that the coordinates in phase space are numerical allows us to express the physical quantities of the transform using normal operators which can be pulled back to the window-signal space with closed form expressions. On the other hand, the assumptions are general enough to encompass most of the common examples of general wavelet transforms.

\begin{definition}[Semi-direct product wavelet transform (SPWT)]
\label{ass_voice2}
A Semi-direct product wavelet transform is constructed by, and assumed to satisfy, the following.
\begin{enumerate}
\item
The group $G$ is a nested semi-direct product group, namely
\begin{equation}
\begin{split}
&G = H_0\\
&H_0=  (N_0\times N_1)\rtimes H_1 \\
%&H_0  =  N_1\rtimes H_1\\
&H_m  =  N_{m+1} \rtimes H_{m+1} \quad , \quad m=1,\ldots, M-2 \\
&H_{M-1}=N_M .
\end{split}
\label{eq:general_semi2}
\end{equation}
Here, $N_0$ is the center of $G$. 
For $m=0,\ldots,M$,
 $N_m$ is a group direct product of physical quantities, $G_m^1\times \ldots\times G_m^{K_m}$, where $K_m\in\NN$. 
For each $m=0,\ldots,M$, \textcolor{black}{ we assume that $G_m^1= \ldots= G_m^{K_m}$  are the same physical quantity.}
\item
We denote elements of $N_m$ in coordinates by 
\[{\bf g}_m=(g_m^1,\ldots,g_m^{K_m}) \sim g_m^1g_m^2\ldots g_m^{K_m},\] 
where $g_m^k\in G_m^k$ for $k=1,\ldots,K_m$.
We denote elements of $H_m$ in coordinates by 
\[{\bf h}_m=({\bf g}_{m+1},\ldots,{\bf g}_M)\sim g_{m+1}g_{m+2}\ldots g_M,\]
where $g_{m'}\in N_{m'}$ for $m'=m+1,\ldots,M$. 
For the center, we also denote $Z=N_0$, and $K_z=K_0$, and denote elements of $Z$ in coordinates by ${\bf z}=(z^1,\ldots,z^{K_z})$. 
  \item
  We consider the representations $\pi_m({\bf g}_m)=\pi_m^1(g_m^1) \circ\ldots \circ \pi_m^{K_m}(g_m^{K_m})$ of $N_m$, $m=0,\ldots,M$ in $\cS$. We assume  that
	$\pi(g)=\pi_0({\bf g}_0)\circ\ldots \circ \pi_M({\bf g}_M)$ is a square integrable representation of $G$. Namely, $\pi$ is a strongly continuous irreducible representation, and there is a vector $f\in\cS$ such that $V_f[f]\in L^2(G)$.  %
	\item
	The \emph{semi-direct product wavelet transform} based on $\pi$ and on the window $f\in\cS$, satisfying $V_f[f]\in L^2(G)$, is $V_f:\cS\rightarrow L^2(G)$, defined in (\ref{eq:voice1}).
\end{enumerate}
\end{definition}

When there is no concrete interpretation of the physical quantity underlying the coordinate $N_m$, e.g.$time$, $scale$ or $angle$, we call it by the generic name $quantity_m$.

\subsubsection{Notations and properties}
	%\begin{remark}
	\label{ass_inversion2}
	Consider a SPWT. In the following we give a list of notations and properties from \cite{MyRef} that will be used repeatedly in subsequent sections.
	\begin{enumerate}
	\item
	We denote by $e$ the unit element of a generic group, and by ${\bf e}$ a tuple of unit elements. We denote by $\bullet$ the group product of a generic group, e.g.$g\bullet g'$, or in short $gg'$. We also denote by $\bullet$ the group product in coordinates, e.g., ${\bf g}'\bullet{\bf g} \sim g'\bullet g$.
	\item
	We denote by abuse of notation elements of $G$ with coordinate representation \newline
	$({\bf e},\ldots,{\bf e},{\bf g}_m,{\bf e},\ldots,{\bf e})$, by $g_m$. We sometimes do not distinguish between the group $N_m$, and the subgroup of $G$ consisting of elements of the form $g_m$. We denote elements of $G$ with coordinate representation $({\bf e},\ldots,{\bf e},{\bf h}_{m})$, by $h_m$.
	Note that 
	\[h_m \sim {\bf h}_m=({\bf g}_{m+1},{\bf h}_{m+1})\sim g_{m+1}h_{m+1}.\]
		\item	
	As a result of the semi-direct product structure, the group multiplication has the following form in coordinates
	\begin{equation}
	\begin{split}
	g\bullet g' = & ({\bf z},{\bf g}_1,\ldots,{\bf g}_M)\bullet ({\bf z}',{\bf g}_1',\ldots,{\bf g}_M') \\
	 =    & \Big(({\bf z},{\bf g}_1) \bullet A_{z,1}\big({\bf h}_1;({\bf z}',{\bf g}_1')\big) \ ,\  {\bf g}_{2} \bullet A_{2}({\bf h}_{2};{\bf g}_{2}')    \ ,\ \ldots  \\
	 & \quad \quad\quad \quad \ldots \ , \ {\bf g}_{M-1} \bullet A_{M-1}({\bf h}_{M-1};{\bf g}_{M-1}')\ ,\ {\bf g}_M \bullet {\bf g}_M'\Big)
	\end{split}
	\label{eq:vvvvvvvvvvv}
	\end{equation}
	where $A_m$ are smooth automorphisms with respect to ${\bf g}_m'$ if $m\geq 2$, and with respect to $({\bf z}',{\bf g}_1')$ for $m=(z,1)$, \textcolor{black}{as defined in (\ref{SP_A0})}. Moreover, $A_m$ with respect to ${\bf h}_m$  are smooth group actions of $H_m$ on $N_m$  for $m\geq 2$, and on $Z\times N_1$ for $m=(z,1)$. Here, ${\bf z}$, ${\bf g}_m$, ${\bf h}_m$ ($m=1,\ldots,M$) are coordinates corresponding to $g$, and ${\bf z}'$, ${\bf g}'_m$ ($m=1,\ldots,M$) are coordinates corresponding to $g'$.	
	\item
	There is a convenient way to represent the automorphisms $A_m({\bf h}_m; \ \cdot\ )$ as matrix operators. Let us denote in short $A_m=A_m({\bf h}_m; \ \cdot\ )$ the automorphism $N_m\rightarrow N_m$ for some $m=1,\ldots,M$. 
	For each $k=1,\ldots,K_m$, consider the projections $P_k:N_m\rightarrow G_m^k$, defined in coordinates by
	\[{\bf P}_k(g_m^1,\ldots,g_m^{K_m})=(e,\ldots,e,g_m^k,e,\dots,e)\]
	 Since $N_m$ is a direct product, $P_k$ are homomorphisms. For each pair $k,k'=1,\ldots,K_m$ consider the homomorphism $a_{k,k'}=P_{k}A_mP_{k'}:G_m^{k'}\rightarrow G_m^k$. The automorphism $A$ can now be written in coordinates as the homomorphism valued invertible matrix ${\bf A}$ with entries $a_{k,k'}$. 
	The multiplication of ${\bf A}$ by ${\bf g}_m$ is defined by
	\[\forall k=1,\ldots,K_m\ , \quad [{\bf A}{\bf g}_m]_k = a_{k,1}(g_m^1)\bullet \ldots \bullet a_{k,K_m}(g_m^{K_m}).\]
	Here, invertibility means that there is another homomorphism valued matrix ${\bf A}^{-1}$ with ${\bf A}{\bf A}^{-1}={\bf A}^{-1}{\bf A}={\bf I}$. We denote these matrices in the extended notation by ${\bf A}_m({\bf h}_m)$.
	
	For example, if $G^1_m=\{\RR,+\}$, any automorphism is an invertible matrix in $\RR^{K_m \times K_m}$.
	If $G_m=\{\ZZ,+\}$, any automorphism is an invertible matrix in $\ZZ^{K_m \times K_m}$, with inverse in $\ZZ^{K_m \times K_m}$. %, e.g
	Another example is the case where $N_m=G^1_m$ consists of one coordinate.  In this case, for $N_m=G_m=\{\RR,+\}$, ${\bf A}_m({\bf h}_m)$ is a multiplication by a nonzero scalar. For $N_m=G_m=\{\ZZ,+\}$,  ${\bf A}_m({\bf h}_m)$ is a multiplication by $\pm 1$. For $N_m=G_m=\{e^{i\RR},+\}$, ${\bf A}_m({\bf h}_m)$ is a multiplication of the exponent by $\pm 1$. Last, for $N_m=G=\{e^{2\pi i \ZZ/N}\}$, ${\bf A}_m({\bf h}_m)$ is a multiplication of the exponent by a co-prime of $N$.

\item
The matrices ${\bf A}_m({\bf h}_m)$, with respect to ${\bf h}_m$, are smooth group actions of $H_m$ on $N_m$. As a result
\[{\bf A}_m({\bf h}_m)={\bf A}_m({\bf g}_{m+1})\ldots{\bf A}_m({\bf g}_{M}).\]

		\item
		There is a formula for the group inverse of $g\in G$ in coordinates. 
The inverse of $({\bf z},{\bf g}_1,{\bf h}_1)$ is given by 
\begin{equation}
({\bf z},{\bf g}_1,{\bf h}_1)^{-1}=\big({\bf A}_{z,1}({\bf h}_1^{-1})({\bf z}^{-1},{\bf g}_1^{-1}),{\bf h}_1^{-1}\big).
\label{eq:inverse_coo0}
\end{equation}
The inverse of $({\bf g}_m,{\bf h}_m)$ in $H_{m-1}$ is given by
\begin{equation}
({\bf g}_m,{\bf h}_m)^{-1}=\big({\bf A}_m({\bf h}_m^{-1}){\bf g}_m^{-1},{\bf h}_m^{-1}\big).
\label{eq:inverse_coo1}
\end{equation}

\end{enumerate}

\subsubsection{The $quantity_m$ transforms}

\textcolor{black}{The assumption that the coordinates of a SPWT are scalars is essential in our wavelet-Plancherel phase space filtering theory. This fact allows us to measure the physical quantities of the transform using normal operators. More accurately, it is possible to define multiplicative operators in phase space that multiply each point by the value of one physical quantity coordinate. These operators are the so-called \emph{observables of the $quantity_m$ transforms}, introduced in \cite{MyRef}. The special nested semi-direct product structure of SPWT is also essential in the wavelet-Plancherel theory, as it allows to faithfully represent the group rule as canonical commutation relations on the observables (see Definition \ref{def:quantity_trans}). This is the key property which allows explicit pull-back formulas which generalize the 1D CWT case  (\ref{eq:1D_Pull_2}) and (\ref{eq:1D_Pull_22}).}

The $quantity_m$ transforms formalize the notion that $\pi(g)$ translates the values of the physical quantities of windows. 
To illustrate the concept, consider the STFT, with group structure
\[({\rm phase\ rotation}\times{\rm translations})\rtimes{\rm modulations} =(e^{i\RR}\times\RR)\rtimes \RR = (Z\times N_1)\rtimes N_2.\]
 The operator $\pi_1(g_1)$ translates $f$ in the \textit{time} domain, and $\cF\pi_2(g_2)\cF^{-1}$ translates $\cF f$ in the \textit{frequency} domain. In the point of view of the SPWT theory, what makes $\cF$ a \textit{frequency} transform is the fact that it maps $\pi_2(g_2)$ to a translation in $L^2(N_2)=L^2(\RR)$, and what makes the identity operator a \textit{time} transform is the fact that it maps $\pi_1(g_1)$ to a translation in $L^2(N_1)=L^2(\RR)$. 

Next, we recall the extension of this idea for general SPWTs. In a general SPWT, when translating an element $g\in G$ by multiplying from the left with an element $g^{\prime-1}_m\in G_m$, not only the ${\bf g}_m$ coordinate of $g$ is affected. Indeed, by the semi-direct product structure of $G$, all of the coordinates ${\bf g}_1,\ldots, {\bf g}_m$ are affected. This is taken into account in the following construction \cite{MyRef}.

We want to prove an equivalence between $\pi_m(g_m)$ and translations in some space. The space $L^2(N_m)$ is a good candidate for this translation space. However, for dimensional balance considerations, the space $L^2(N_m)$ is not appropriate in general. For example, consider the representation $\pi_1(g_1)$ that translates $L^2(\RR^2)$ functions along the $x$ axis. In this case $N_1=\RR$. However, there is no reasonable transform from $L^2(\RR^2)$ to $L^2(N_1)=L^2(\RR)$ that transforms $\pi_1(g_1)$ to translation. To solve this problem, we consider a manifold $Y_m$ with Radon measure, and take the translation space to be $L^2(Y_m\times N_m)$. In \cite{MyRef}, the existence and construction of such a space $Y_m$ is studied. We call $Y_m\times N_m$ the \emph{$quantity_m$ domain}.

For the definition of the $quantity_m$ transforms, we rely on calculations of tuples of operators and tuples of vectors.
Let $M\in\NN$.
A tuple of operators ${\bf T}=(T_1,\ldots, T_M)$, applied on a single vector $f$, is defined to be ${\bf T}f=(T_1 f,\ldots, T_M f)$. The composition of a tuple of operators ${\bf T}$ with an operator $U$ is defined by ${\bf T}U=(T_1U,\ldots, T_MU)$ and $U{\bf T}=(UT_1,\ldots, UT_M)$.

We are interested in $quantity_m$ transforms of the form $\Psi_m: \cS \rightarrow L^2(Y_m\times N_m)$ where $\cS$ is the signal space.  For every $1\leq k\leq K_m$, we define the multiplicative operator $\bQ_m^k$ in $L^2(Y_m\times N_m)$ by
	\begin{equation}
	\bQ_m^k F(y_m,g_m)=g_m^kF(y_m,g_m),
	\label{eq:mult_one1}
	\end{equation}
	with domain ${\cal D}(\bQ_m^k)$ consisting of all functions $F\in L^2(Y_m\times N_m)$ such that (\ref{eq:mult_one1}) is also in 
	$L^2(Y_m\times N_m)$. Note that $\bQ_m^k$ is unitary in case $G_m^1$ is $e^{i\RR}$ or $e^{2\pi i\ZZ/N}$, or self-adjoint in case $G_m^1$ is $\RR$ or $\NN$.
	 The operators $\bQ_m^k$ multiply each point in $Y_m\times N_m$ by the value of its physical quantity $G_m^k$. It is customary in quantum mechanics to call such operators \emph{observables of $qunatity_m$} \cite{HAPS,quantum_measure}.
	Consider the multiplicative operator ${\bf\bQ}_m$, defined in $\bigcap_{k=1}^{K_m}{\cal D}(\bQ_m^k)$, mapping to $L^2(Y_m\times N_m)^{K_m}$, by
	\[{\bf\bQ}_m F = ( \bQ_m^1 F , \ldots , \bQ_m^{K_m} F ).\]
	Note that the entries of ${\bf\bQ}_m$ are either unitary or essentially self-adjoint\footnote{A symmetric operator is called essentially self-adjoint, if there exists a unique extension of the operator to a self-adjoint operator.}.

The following definition imposes the group structure of $G$ on the $quantity_m$ transforms through observables.
\begin{definition}
\label{def:quantity_trans}
Let $1\leq m \leq M$, and
 $Y_m$ be a manifold with Radon measure. The isometric isomorphism 
$\Psi_m: \cS \rightarrow L^2(Y_m\times N_m)$ is called a $quantity_m$ transform if the following two conditions are met.
\begin{itemize}
	\item 
	Consider the left $N_m$ translation $L_m(g_m'):L^2(Y_m\times N_m)\rightarrow L^2(Y_m\times N_m)$ defined by
	\[[L_m(g_m')F](y_m,g_m) = F(y_m,g_m^{\prime-1}g_m).\]
	Then
	\begin{equation}
	\Psi_m\pi_m(g_m')\Psi_m^* = L_m(g_m').
	\label{eq:quant_trans1}
	\end{equation}
	\item
	Define the representation
	\[\rho^m(g')=\Psi_m\pi(g')\Psi_m^*.\]
	Then 
	\begin{equation}
	\rho^m(g')^*{\bf\bQ}_m \rho^m(g') = {\bf g}'_m\bullet {\bf A}_m({\bf h}'_m){\bf\bQ}_m,
	\label{eq:quant_trans2}
	\end{equation}
	where the operator tuple ${\bf g}'_m\bullet {\bf A}_m({\bf h}'_m){\bf\bQ}_m$ is understood in the functional calculus sense, as the application of the function ${\bf g}'_m\bullet {\bf A}_m({\bf h}'_m)(\cdot)$ on ${\bf\bQ}_m$ elementwise (see Remark \ref{R3e3}). Equation (\ref{eq:quant_trans2}) is called the \emph{canonical commutation relation} of $quantity_m$.
\end{itemize}
\end{definition}

\begin{remark}
\label{R3e3}
Equation (\ref{eq:quant_trans2}) is understood as follows.
Note that for any $k=1,\ldots,K_m$ we have
\[\int dP_m^k(\l_m^k)=I\]
where $P_m^k(\l_m^k)$ are the spectral projections of $\bQ_m^{k}$.
Thus, by commutativity of \newline $dP_m^1(\l_m^1),\ldots, dP_m^{K_m}(\l_m^{K_m})$ we can write the spectral decomposition of ${\bf\bQ}_m$ by
\[{\bf\bQ}_m = \Big(\int \l_m^1 dP_m^1(\l_m^1)\ ,\ \ldots\ ,\  \int \l_m^{K_m}dP_m^{K_m}(\l_m^{K_m})\Big)\]
\[=\Big(\int\ldots\int \l_m^1  dP_m^1(\l_m^1)\ldots dP_m^{K_m}(\l_m^{K_m})\ \ ,\ \ \ldots\ \ ,\ \  \int\ldots\int \l_m^{K_m} dP_m^1(\l_m^1)\ldots dP_m^{K_m}(\l_m^{K_m})\Big)\]
\[= \int\ldots\int (\l_m^1,\ldots, \l_m^{K_m}) dP_m^1(\l_m^1)\ldots dP_m^{K_m}(\l_m^{K_m}).\]
We now define the tuple of operators
\[{\bf g}'_m\bullet {\bf A}_m({\bf h}'_m){\bf\bQ}_m = \int\ldots\int {\bf g}'_m\bullet {\bf A}_m({\bf h}_m)(\l_m^1,\ldots, \l_m^{K_m}) dP_m^1(\l_m^1)\ldots dP_m^{K_m}(\l_m^{K_m}),\]
\textcolor{black}{where $(\l_m^1,\ldots, \l_m^{K_m})$ is seen as a group element of $N_m$, and ${\bf g}'_m\bullet {\bf A}_m({\bf h}_m)(\l_m^1,\ldots, \l_m^{K_m})$ as a tuple of numerical values.}
\end{remark}

Note that if $G_m^1$ is $\RR$ or $\ZZ$, then the operators in the tuple (\ref{eq:quant_trans2}) are essentially self-adjoint, as the sum of commuting essentially self-adjoint operators. Similarly, in case $G_m^1$ is $e^{i\RR}$ or $e^{2\pi i\ZZ/N}$, the operators in the tuple (\ref{eq:quant_trans2}) are unitary.

For later calculations, it is important to demand the existence of the $quantity_m$ transforms. In addition, it is required to restrict the group structure of $G$, to have nontrivial ${\bf A}({\bf h_m})$ automorphisms only for $quantity_m$ coordinates that are equal to $\RR$ of $\ZZ$. This is summarized in the following definition.

\begin{definition}
\label{simply_dilatet}
A SPWT is called simply dilated, if for every $m=1,\ldots,M$ the following two conditions are met.
\begin{enumerate}
	\item
	There exists a $quantity_m$ transform.
	\item
	If $G^1_m$ is $e^{i\RR}$ or $e^{2\pi i \ZZ/N}$, then ${\bf A}_m({\bf h}_m)$ is the unit matrix ${\bf I}$, mapping every ${\bf g}_m$ to itself.
\end{enumerate} 
\end{definition}

Note that in Definition \ref{simply_dilatet} the center $N_0$ of $G$ is not required to have a $quantity_0$ transform. Indeed, such a transform does not exist, as the restriction of $\pi$ to $N_0$ is a scalar representation times the identity, which makes (\ref{eq:quant_trans1}) and (\ref{eq:quant_trans2}) impossible. For this reason, the STFT is a wavelet transform that measure only times and frequencies, and the phase shift subgroup $N_0$ is not considered to be a transformation group of a physical quantity.

\subsection{Geometric wavelet transforms}

In this section we define a class of SPWT that are based on affine transformations of a window in $L^2(\RR^N)$, called geometric wavelet transforms.
Geometric wavelet transforms are a special case of the class of wavelet transforms introduced in
\cite{geometric}. \textcolor{black}{We study geometric wavelet transforms for two reasons. First, the pull-back of phase space filters can be written explicitly via first order partial differential equations, and second, geometric wavelet transforms are ubiquitous in general wavelet transforms.}

In the following we analyze $L^2(\RR^N)$ functions through their spatial domain and frequency domain representations $f$ and $\hf$ respectively.
We treat points in the spatial domain as column vectors in $\RR^N$, and treat points in the frequency domain as row vectors in $\RR^N$. We do not distinguish between linear operators ${\bf A}\in {\cal GL}(\RR^N)$ and their matrix representations in the standard basis. For example, for time points $\bx\in\RR^N$ and frequency points $\bw\in\RR^N$, ${\bf A}\bx$ and $\bw {\bf A}$ are defined by matrix multiplication.

\begin{definition}
Let $H_1$ be a nested semi-direct product group of the subgroups $N_m$, $m=2,\ldots,M$, where each $N_m$ is a group direct product of physical quantities, as in Definition \ref{ass_voice2}. Consider a representation $\bD:H_1\rightarrow {\cal GL}(\RR^N)$. For $h_1=g_2\ldots g_M$, we denote the decomposition $\bD(h_1)^{-1}=\bD_M(g_M)^{-1}\ldots \bD_2(g_2)^{-1}$, where each ${\bf D}_m:N_m\rightarrow {\cal GL}(\RR^N)$ is a representations of the subgroup $N_m$. Denote the group $N_1=\{\RR^N,+\}$. Consider the group $G=\RR^N\rtimes H_1=N_1\rtimes H_1$, with group product defined in coordinates by
\[(\bx,{\bf h}_1)\bullet(\bx',{\bf h}_1') = (\bx+{\bf D}({\bf h}_1)\bx',{\bf h}_1\bullet{\bf h}_1').\]
Let $\bw_0\in\RR^N$ and suppose that $U=\bw_0{\bf D}(H_1)= \{\bw_0 \bD(h_1)\ |\ h_1\in H_1\}$ is an open dense subset of $\RR^N$. Suppose that $\pi:G\rightarrow {\cal U}(L^2(\RR^N))$, defined by
\[[\pi(g)f](\bx)=\abs{{\rm det}\big(\bD({\bf h}_1)^{-1}\big)}^{\frac{1}{2}} f\big(\bD({\bf h}_1)^{-1}(\bx-{\bf g}_1)\big)\]
is a square integrable representation of $G$ in $L^2(\RR^N)$.
If the resulting wavelet transform $V_f$ is a simply dilated SPWT, we call $V_f$ a \emph{geometric wavelet transform}.
\end{definition}

A geometric wavelet transform takes the following form in the frequency domain %: Denote 
\begin{equation}
{\hat{\pi}}(g)\hf(\bw) :=\cF\pi(g)\cF^*\hf(\bw) = \abs{{\rm det}\big(\bD({\bf h}_1)\big)}^{\frac{1}{2}}e^{i\bw\cdot {\bf g}_1}\hf\big(\bw \bD({\bf h}_1)\big).
\label{eq:geo_wav_freq}
\end{equation}

The following theorem is a special case of Corollary 1 of Theorem 1 in \cite{geometric}.
\begin{theorem}
\label{Geo_Duflo}
There exists a measurable function $\kappa:\RR^N\rightarrow \RR_+$ such that the
 Duflo-Moore operator $K$ is given by
\[[\cF K \cF^* \hf](\w) = \kappa(\w)\hf(\w).\]
\end{theorem}

Consider a geometric wavelet transform.
We denote $Y_m=N_2\times\ldots\times N_{m-1}\times N_{m+1}\times\ldots\times N_M$, and denote generic elements of $Y_m$ in coordinates by $y_m\sim {\bf y}_m=({\bf g}_2,\ldots,{\bf g}_{m-1},{\bf g}_{m+1},\ldots,{\bf g}_M)$. We denote
\[\bD^m(y_m) = \bD_2({\bf g}_2)\ldots \bD_{m-1}({\bf g}_{m-1})\bD_{m+1}({\bf g}_{m+1})\ldots \bD_M({\bf g}_M).\]
Next, we introduce a property, that if satisfied, there is a constructive definition of all of the $qunatity_m$ transforms of a geometric wavelet transform.

\begin{definition}
Consider a geometric wavelet transform.
Let $U\subset \RR^N$ be the domain of the signal space in frequency, inheriting the topology of $\RR^N$.
If for each $m=2,\ldots,M$, there exists a point $\bw_m\in\RR^N$ such that the  mapping
\[C_m:Y_m\times N_m \rightarrow U \quad , \quad C_m(y_m,g_m) = \bw_m\bD^m(y_m)\bD_m(g_m^{-1})\]
is a diffeomorphism of the manifold  ${\textstyle G}/ _{\textstyle \RR^N} \cong Y_m\times N_m$ to $U$, we call the wavelet transform a \emph{diffeomorphism geometric wavelet transform}, or \emph{DGWT}.
\end{definition}

In the following we consider the signal space $\cS$ of a geometric wavelet transform as the frequency domain $L^2(U)$, using (\ref{eq:geo_wav_freq}) for the representation. By Theorem \ref{Geo_Duflo}, $\cW$ is the weighted frequency space $L^2(U,\k^2(\w)d\w)$.

\begin{definition}
Consider a DGWT.
Consider the standard Riemannian structure on $Y_m\times N_m$, namely the direct product of Riemannian structures of the physical quantities.
Let $J(y_m,g_m)$ be the Jacobian of the substitution $C_m$. The \emph{$quantity_m$ diffeomorphism} is defined to be
\[\Gamma_m: L^2(U)\rightarrow L^2(Y_m\times N_m) \quad , \quad  [\Gamma_m \hf](y_m,g_m) = \hf\big(C_m(y_m,g_m)\big).\]
The \emph{diffeomorphism $quantity_m$ transform} is defined to be
\begin{equation}
\begin{split}
 & \Psi_m: L^2(U)\rightarrow L^2(Y_m\times N_m) \  ,  \\
  &  [\Psi_m \hf](y_m,g_m) = \sqrt{\abs{{\rm det}(J(y_m,g_m))}}\hf\big(C_m(y_m,g_m)\big).  
\end{split}
\label{eq:quantity_trans}
\end{equation}
\end{definition}
Note that $\Psi_m$ is an isometric isomorphism. 
In the following we show that $\Psi_m$ is a $quantity_m$ transform.

\begin{remark}
The construction also works under the direct product Riemannian structure weighted by the Haar measure of $G$, represented in the coordinates
$(y_m,g_m)\sim y_m g_m$, with a corresponding Jacobian. Note that $\Gamma_m$ is the same in this construction, and $\Psi_m$ has a different normalizing Jacobian.
\end{remark}

\begin{proposition}
\label{T_m_multiplicative_in_freq}
Consider a DGWT, where $\cS$ is considered to be the frequency domain $L^2(U)$.
Then we have the following.
\begin{itemize}
	\item The inverse Fourier transform is a \emph{position transform}, \textcolor{black}{namely, the $quantity_1$ transform of the position parameter $\mathbf{g}_1$ (see Definition \ref{def:quantity_trans}).}
	\item For every $m=2,\ldots,M$, the diffeomorphism $quantity_m$ transform $\Psi_m$ is a  $quantity_m$ transform.
\end{itemize} 
%Consider a DGWT, where $\cW,\cS$ are considered to be frequency domains \newline
%$L^2(U,\k^2(\w)d\w)$ and $L^2(U)$ respectively.
%The inverse Fourier transform is a position transform (a $quantity_1$ transform). Moreover, for every $m=2,\ldots,M$, the diffeomorphism $quantity_m$ transform $\Psi_m$ is a  $quantity_m$ transform.
\end{proposition}

The proof of this proposition is in Appendix B.

\subsection{Examples of semi-direct product wavelet transforms}
\label{Examples_SPWT}

The following examples can be checked directly (see \cite{MyRef}) or as instances of Proposition \ref{T_m_multiplicative_in_freq}. It is also noteworthy to mention the papers \cite{allW1}\cite{allW2}, in which all square integrable quasi-regular representation in $L^2(\RR^2)$ and $L^2(\RR^3)$ respectively were classified and explicitly given. Many of the representations listed in these papers are SPWT, and the papers can serve as a rich library of wavelet transforms.

\subsubsection{The short time Fourier transform}

The STFT is a simply dilated SPWT. The semi-direct product structure of the Heisenberg group is 
\[J= (e^{i\RR }\times \RR)\rtimes\RR = ({\rm phase\ rotations} \times {\rm translations}) \rtimes {\rm modulations},\]
where $e^{i\RR}$ is the center of $J$. We have
\[A_{z,1}=\left(\begin{array}{cc}
	I(g_2) &A_z(g_2) \\
	0(g_2)  &  A_1(g_2)
\end{array}\right)\]
where 
\[\begin{array}{lccclll}
	I(g_2)&:&\{e^{i\RR},\cdot\}\rightarrow\{e^{i\RR},\cdot\} & \ \ \ ,\ \ \  & I(g_2)z&=&z, \\
	A_z(g_2)&:&\{\RR,+\}\rightarrow\{e^{i\RR},\cdot\} & \ \ \ ,\ \ \  & A_z(g_2)g_1&=&e^{ig_2g_1}, \\
	0(g_2)&:&\{e^{i\RR},\cdot\}\rightarrow\{\RR,+\} & \ \ \ ,\ \ \  & 0(g_2)z&=&0,\\
	A_1(g_2)&:&\{\RR,+\}\rightarrow\{\RR,+\} & \ \ \ ,\ \ \  & A_1(g_2)g_1&=&g_1.
\end{array}\]
In the application of $A_{z,1}$ on $(z,g_1)$, the elements of the first row are multiplied, and the elements of the second row are summed.
 The $time$ transform is the unit operator, and the $frequency$ transform is the Fourier transform. The Duflo-Moore operator is $K=I$.

\subsubsection{The 1D wavelet transform} 

The 1D wavelet transform is a DGWT. The semi-direct product structure of the 1D affine group is 
\[A= \RR\rtimes(\RR\times\{1,-1\}) = \RR\rtimes(\RR\rtimes\{1,-1\}) = {\rm translations} \rtimes ({\rm dilations}\rtimes {\rm reflections}),\]
with trivial center. We have 
\[A_1(g_2,g_3):\RR\rightarrow\RR \quad , \quad A_1(g_2,g_3)g_1=g_3e^{g_2}g_1.\]
The set $U$ is $\RR\setminus\{0\}$.
The $time$ transform is the unit operator. The $scale$ transform 
$\Psi_2: L^2(\RR)\rightarrow L^2(\RR)$ is given by
\[[\cF\Psi_2 \cF^*\hf](g_2)=e^{-\frac{1}{2}g_2}\hf(e^{-g_2}).\]

\subsubsection{The rotation-dilation wavelet transform}.

The rotation-dilation wavelet transform is a DGWT. The semi-direct product structure of its group is 
\[A_R= (\RR\times\RR)\rtimes(e^{i\RR}\times \RR) = (\RR\times\RR)\rtimes(e^{i\RR}\rtimes  \RR) = {\rm translations} \rtimes ({\rm rotations}\rtimes {\rm dilations}),\]
with trivial center. We have 
\[A_1(g_2,g_3):\RR^2\rightarrow\RR^2 \quad , \quad A_1(g_2,g_3)\bg_1={\bf R}(g_2){\bf D}(g_3)\bg_1\]
where ${\bf R},{\bf D}$ are the rotation and dilation matrices
\[{\bf R}=  
 \left(
\begin{array}{cc}
	{\rm Re}(g_3) & -{\rm Im}(g_3)\\
	{\rm Im}(g_3) & {\rm Re}(g_3)
\end{array}
\right)
\quad ,\quad  
{\bf D}=\left(
\begin{array}{cc}
	e^{g_2} & 0\\
	0 & e^{g_2}
\end{array}
\right) .\]
The set $U$ is $\RR^2\setminus\{{\bf 0}\}$.

The $position$ transform is the unit transform, 
\[[\Psi_1^xf](g_1^1,g_1^2)=f(g_1^1,g_1^2)\]

Since rotations and dilations commute, we can model the $angle$ and the  $scale$ transforms using one transform.
The $angle-scale$ transform $\boldsymbol{\Psi_2}: L^2(\RR^2)\rightarrow L^2(e^{i\RR}\times \RR)$ is given by
\begin{equation}
[\cF\Psi_2 \cF^*\hf](g_2,g_3)=e^{-g_3}\hf(e^{-g_3}\Re(g_2), e^{-g_3}\Im(g_2)).
\label{eq:angle-scale}
\end{equation}

\subsubsection{The Shearlet transform}

The Shearlet transform is a DGWT. The semi-direct product structure of the Shearlet group is 
\[\begin{split}
A_S =& (\RR\times\RR)\rtimes\big(\RR\rtimes (\RR\times\{1,-1\})\big)= (\RR\times\RR)\rtimes\big(\RR\rtimes (\RR\rtimes\{1,-1\})\big) \\
 = & {\rm translations} \rtimes \big({\rm shears}\rtimes {(\rm anisotropic\ dilations}\rtimes{\rm reflections})\big),
\end{split}\]
with trivial center. We have 
\[A_2(g_3,g_4):\RR\rightarrow\RR\quad , \quad A_2(g_3,g_4)g_2=e^{\frac{1}{2}g_3}g_2\]
\[A_1(g_2,g_3,g_4):\RR^2\rightarrow\RR^2 \quad , \quad A_1(g_2,g_3,g_4)\bg_1=g_4{\bf D}(g_3){\bf S}(g_2)\bg_1\]
where ${\bf D}(g_3),{\bf S}(g_2)$ are the anisotropic dilation and shear matrices
\[{\bf D}(g_3)=\left(
\begin{array}{cc}
	e^{g_2} & 0\\
	0 & e^{\frac{1}{2}g_2}
\end{array}
\right) 
\quad ,\quad  
{\bf S}(g_2)=  
 \left(
\begin{array}{cc}
	1 & g_2\\
	0 & 1
\end{array}
\right).\]
The set $U$ is $\{(\x,\y)\in\RR^2\ |\ \x\neq 0\}$.

The $position$ transform is the unit transform, 
\[[\Psi_1^xf](g_1^1,g_1^2)=f(g_1^1,g_1^2).\]
The $slope$  transform 
$\Psi_2: L^2(\RR^2)\rightarrow L^2(\RR^2)$ is given by
\begin{equation}
[\cF\Psi_2 \cF^*\hf](y_2,g_2)=e^{y_2}\hf(e^{y_2},-g_2 e^{y_2}).
\label{eq:slope_t}
\end{equation}
The $anisotropic\ scale$  transform $\Psi_3: L^2(\RR^2)\rightarrow L^2(\RR^2)$ is given by
\begin{equation}
[\cF\Psi_3 \cF^*\hf](y_3,g_3)=e^{-\frac{3}{4}g_3}\hf(e^{-g_3},y_3e^{-\frac{1}{2}g_3}).
\label{eq:scale_t}
\end{equation}

\section{A wavelet Plancherel theory}
\label{A wavelet Plancherel theory}

The wavelet transform $V_f:\cS\rightarrow L^2(G)$ is an isometric embedding of $\cS$ into $L^2(G)$, and is generally not surjective.
Recall that our goal is to pull-back \textcolor{black}{phase space filters} from phase space to the signal space. As one example, consider a window $f$, a signal $s$, a domain $B\subset G$ and the projection (band-pass filter) $P_B: L^2(G)\rightarrow L^2(G)$ defined for any $F\in L^2(G)$ by 
\[P_B F(g) = \left\{ \begin{array}{ccc}
	F(g) & ,  & g\in B \\
	0 & ,  & g\notin B \\
\end{array}\right. .\]
We want to know if there is some signal $q\in\cS$ such that $V_f[q]= P_B V_f[s]$. Since $V_f[\cS]$ is a space of continuous functions, and band passed continuous functions are in general not continuous, such a $q$ does not exist in general. One might think that considering approximations to band-pass filters, based on multiplying $F$ with a continuous function, might alleviate this problem. However, such an approach also fails. Indeed, in view of the reproducing kernel property of $V_f[\cS]$ (Corollary
\ref{Amb_is_kernel2}), narrow band pass filters cannot be approximated in the ``blurry'' space $V_f[\cS]$.

We thus need to extend the wavelet theory to accommodate \textcolor{black}{the pull-back of phase space filters.} We do this by embedding the signal space $\cS$ in a larger space, and canonically extending the wavelet transform to a mapping between the larger signal space and $L^2(G)$. Considering Proposition \ref{prop_wavelet1}, it is natural to choose the larger signal space as the tensor product of some window space with the signal space, where in the window space we define the inner product via $\ip{K f_2}{K f_1}$. Indeed, using Proposition \ref{prop_wavelet1}, it is almost a triviality that $V_f$ extends to an isometry in the tensor product space. In this section we define this extension accurately. In subsequent sections, we show how to derive closed form formulas for the pull-back of band-pass filters in phase space, or more generally of multiplicative operators in phase space. Moreover, we show that the image space of the wavelet-Plancherel transform is invariant under such multiplicative operators.

\subsection{Tensor products}
\label{Tensor products}

The wavelet Plancherel theory relies on the notion of tensor product.
There are many equivalent (or isomorphic) ways to define the tensor product of two spaces. We adopt a less abstract but convenient definition, based on orthonormal basis of Hilbert spaces \textcolor{black}{(see, e.g., \cite[Chapter 3.4]{tens1})}. In our definition we complex conjugate the coefficients of one of the spaces,  to accommodate comparison with the wavelet orthogonality relation. 
\textcolor{black}{In \cite[Page 9]{Fuhr_wavelet} the same approach to tensor product was considered, but based on the Hilbert–Schmidt operator definition of tensor product.}
Since it is more standard to define tensor products without conjugation, we go through some of the technical details.

Let $\cW,\cS$ be two separable Hilbert spaces. Let $\{\phi_n\}_{n\in\NN}$, $\{\eta_n\}_{n\in\NN}$ be orthonormal basis of $\cW$ and $\cS$ respectively.
We define the Hilbert space $\cW\otimes\cS$ formally as the space with the orthonormal basis
\[\{\phi_n\otimes \eta_m\}_{n,m\in\NN}.\]
We define the inner product of $\cW\otimes\cS$ for basis elements by
\[\ip{\phi_n\otimes \eta_m}{\phi_k\otimes \eta_l} = \delta_{(n,m),(k,l)}.\]
Using this, we can determine uniquely the inner product for general $\cW\otimes\cS$ vectors using linearity and continuity.

This definition by itself seems to carry no information, as all separable Hilbert spaces have orthogonal bases. To add some essence to the tensor product space, we need to define how vectors of the spaces $\cW,\cS$ are related to the space $\cW\otimes\cS$.
There is a
sesquilinear form, dependent on the basis $\{\phi_n\}_{n\in\NN}$ and $\{\eta_n\}_{n\in\NN}$, embedding $\cW\times\cS$ into $\cW\otimes\cS$
\[(f,s)\mapsto f\otimes s.\]
This embedding is defined for $f=\sum_{n\in\NN}c_n \phi_n$ and $s=\sum_{m\in\NN}d_m \eta_m$ by
\[f\otimes s = \sum_{m,n\in\NN}\overline{c_n}d_m \ \phi_n\otimes \eta_m.\]
Note the difference from the standard \textcolor{black}{\textcolor{black}{Hilbert space basis definition (e.g., in \cite[Chapter 3.4]{tens1})}}, in which $c_n$ is not conjugated. 
The image of $\cW\times\cS$ under $\otimes$ is not a linear subspace of $\cW\otimes\cS$. For example, for non-zero linearly independent $f_1,f_2 \in \cW$ and $s_1, s_2 \in \cS$, the sum $f_1\otimes s_1 + f_2\otimes s_2$ is not an image of any $\cW \times \cS$ pair under $\otimes$. 
The space $\cW\otimes\cS$ is the linear closure of the image of $\otimes$.
We call a vector in $\cW\otimes\cS$ of the form $f\otimes s$ a simple tensor.
The inner product of simple tensors is given by
\begin{equation}
\ip{f_1\otimes s_1}{f_2\otimes s_2}_{\cW\otimes\cS} = \overline{\ip{f_1}{f_2}_{\cW}}\ip{s_1}{s_2}_{\cS}.
\label{eq:simple_inner}
\end{equation}

The tensor product operator allows us to embed vectors from $\cW\times\cS$ to $\cW\otimes\cS$. By this, we can define operators in $\cW\otimes\cS$ based on operators in the spaces $\cW,\cS$.
The tensor product of two linear operators $T_1\in {\cal GL}(\cW), T_2\in{\cal GL}(\cS)$ is defined on simple tensors by
\[T_1\otimes T_2(f\otimes s) = (T_1 f)\otimes(T_2 s).\]
Denote the space of finite linear combinations of simple tensors by $[\cW\otimes\cS]_0$, which is a dense subspace of $\cW\otimes\cS$. The operator $T_1\otimes T_2$ has a natural extension to $[\cW\otimes\cS]_0$ via linearity. 
Sometimes operators in  $[\cW\otimes\cS]_0$ can be extended to operators in  $\cW\otimes\cS$. For example, bounded operators can be extended using boundedness, and the fact that $[\cW\otimes\cS]_0$ is dense in $\cW\otimes\cS$. Symmetric operators $T$ in $[\cW\otimes\cS]_0$ can sometimes be extended to self-adjoint operators in $\cW\otimes\cS$ by explicitly finding a self-adjoint operator $T'$ in $\cW\otimes\cS$ such that the restriction of $T'$ to $[\cW\otimes\cS]_0$ is equal to $T$. \textcolor{black}{Such a case will be discussed in Subsection \ref{The pull-back of phase space observables}, and especially in Proposition \ref{prop:SPWT_pull_obs}.}

The tensor product definition depends on the basis of $\cW,\cS$, and different bases lead to different tensor product operators $\otimes:\cW\times\cS\rightarrow \cW\otimes\cS$. However, there is a natural isometric isomorphism between different tensor product spaces via the change of basis.
Namely , let  $\{\phi_n\}_{n\in\NN}$, $\{\phi'_n\}_{n\in\NN}$ be two basis of $\cW$, and $\{\eta_n\}_{n\in\NN}$, $\{\eta'_n\}_{n\in\NN}$ two basis of $\cS$. Let $\mathbf{\Phi}$ and $\mathbf{\Psi}$ be the change of basis matrices, namely
\[{\Phi}_{nn'}=\ip{\phi_n}{\phi'_{n'}} \quad, \quad {\Psi}_{mm'}=\ip{\eta_m}{\eta'_{m'}}.\]
We define, by abuse of notation, for simple tensors 
\begin{equation}
(\mathbf{\Phi}\otimes\mathbf{\Psi})f\otimes s 
:= f\otimes' s,
\label{eq:canonical_transform}
\end{equation}
namely
\[
{    (\mathbf{\Phi}\otimes\mathbf{\Psi})\sum_{n,m}\overline{\ip{f}{\phi_{n}}}\ip{s}{\eta_{m}} \phi_{n}\otimes\eta_{m}
= \sum_{n',m'}\overline{\ip{f}{\phi_{n'}'}}\ip{s}{\eta_{m'}'} \phi_{n'}'\otimes\eta_{m'}'},
\]
where $\otimes$,$\otimes'$ denote the tensor product definitions with the basis  $\{\phi_n\}_{n\in\NN}$, $\{\eta_n\}_{n\in\NN}$ and  $\{\phi'_n\}_{n\in\NN}$, $\{\eta'_n\}_{n\in\NN}$ respectively.
Then \textcolor{black}{the extension} $\mathbf{\Phi}\otimes\mathbf{\Psi}:\cW\otimes\cS \rightarrow \cW\otimes'\cS$ is an isometric isomorphism explicitly given by
\[\mathbf{\Phi}\otimes\mathbf{\Psi}\sum_{m,n\in\NN}b_{mn}\phi_m\otimes\eta_n=\sum_{m,n\in\NN}b_{mn}\phi_m\otimes'\eta_n 
= \sum_{m',n'\in\NN}\Big(\sum_{m,n\in\NN}\overline{\Phi_{nn'}}\Psi_{mm'}b_{mn}\Big)\phi'_{m'}\otimes'\eta'_{n'}.\]
This definition is consistent with tensor products of operators, in the sense that for bounded operators $T_1,T_2$
\begin{equation}
(\mathbf{\Phi}\otimes\mathbf{\Psi}) (T_1\otimes T_2)    = (T_1\otimes' T_2).
\label{eq:Tens_consist}
\end{equation}
 Identity (\ref{eq:Tens_consist}) is also true for unbounded operators in the space $[\cW\otimes\cS]_0$.

\begin{example}
\label{Ex_tens_L2R}
Let ${\cal M},{\cal M}'$ be two smooth manifolds with Radon measure.
	In this example we show how the space $L^2({\cal M})\otimes L^2({\cal M}')$ is isomorphic to $L^2({\cal M}\times{\cal M}')$. %, if we base our construction on bases comprising of real-valued $L^2({\cal M})$ and $L^2({\cal M}')$ functions.  
	\textcolor{black}{We consider a definition of a tensor product operator of functions $\phi,\eta$, given by
	\[\overline{\phi(y)}\eta(x).\]
	This is also called an \emph{outer product}, and used to describe the joint probability distribution of independent variables in probability theory.}
	We identify the tensor product of two basis elements with the outer product by
	\[\phi_n\otimes\eta_m = \overline{\phi_n(y)}\eta_m(x).\]
	Now, the tensor product operator $\otimes$ can be identified with the outer product
	$(f\otimes s) (x,y) = \overline{f(y)}s(x)$.
	Indeed,
	\[\sum_{m,n\in\NN}\overline{\ip{f}{\phi_n}}\ip{s}{\eta_m} \ (\phi_n\otimes \eta_m) (x,y)=
	\sum_{m,n\in\NN}\overline{\ip{f}{\phi_n}\phi_n(y)}\ip{s}{\eta_m}\eta_m(x)=\overline{f(y)}s(x).\]
	This shows that the outer product definition is consistent with our definition of tensor product. %, in case the two basis of $L^2({\cal M})$ and $L^2({\cal M}')$  consist of real valued functions.
	%
	%\textcolor{black}{No need for REAL valued functions!?}
\end{example}  

\begin{remark}
\label{remark_tensor_op}
\textcolor{black}{Another equivalent way to define the tensor product $\cW\otimes\cS$ is  as the space of Hilbert–Schmidt operator $\cW\rightarrow\cS$ (see \cite[Page 9]{Fuhr_wavelet}). Here, simple tensors $f\otimes s$ are identified with the rank one operators 
\[\cW\ni q \mapsto \ip{q}{f}s \in \cS,\]
and $\sum_{m,n\in\NN}b_{mn}\phi_m\otimes\eta_n$ is identified with the Hilbert–Schmidt operator
\[\cW\ni q \mapsto \sum_{m,n}b_{mn}\ip{q}{\phi_m}\eta_n \in \cS.\]}
\end{remark}

\subsection{The window-signal space}
\label{The window-signal space}

First, we define a Hilbert space of windows.
Let ${\cal D}(K)\subset\cS$ be the domain of the Dulfo-Moore operator $K$.
Consider the inner product space ${\cal D}(K)$, with the inner product defined by
\begin{equation}
\ip{f_1}{f_2}_{\cW} = \ip{Kf_1}{Kf_2}_{\cS}.
\label{eq:norm_KK}
\end{equation}
\begin{claim}
The inner product space ${\cal D}(K)$ with the inner product (\ref{eq:norm_KK}) is separable.
\end{claim}
\begin{proof}
\textcolor{black}{This follows from the fact that $K$ is densely defined with densely defined inverse (Proposition \ref{prop_wavelet1}). Consider the domain $\mathcal{D}(K^{-1})$ of $K^{-1}$.  There exists a sequence $\{q_k\}_{k\in\NN}\subset \mathcal{D}(K^{-1})$, dense with respect to the inner product of $\cS$. We show that $\{K^{-1}q_k\}_{k\in\NN}$ is dense in $\mathcal{D}(K)$ with respect to the inner product (\ref{eq:norm_KK}). Let $\e>0$ and $r\in\mathcal{D}(K)$. There exists $q\in \mathcal{D}(K^{-1})$ such that $K^{-1}q=r$, and $q_k$ such that $\norm{q_k-q}_{\cS}<\e$. Thus 
\[\norm{K^{-1}q_k-r}_{\cW}=\norm{K^{-1}q_k-K^{-1}q}_{\cW}=\norm{q_k-q}_{\cS}<\e.\]}
\end{proof}

\begin{definition}
\label{def:WS}
Consider the above setting.
\begin{itemize}
	\item 
	The \emph{window space} $\cW$ is defined to be the completion to a Hilbert space of the inner product space $\mathcal{D}(K)$ with inner product (\ref{eq:norm_KK}).
	\item
	 The \emph{window-signal space} is defined to be $\cW\otimes\cS$.
\end{itemize}
\end{definition}

Note that ${\cal D}(K)=\cW\cap\cS$ is dense both in $\cW$ and $\cS$.

\begin{example}
\label{ex_waveWS}
Consider the 1D wavelet transform, and restrict the analysis to the frequency domain. Then we have $\cS=L^2(\RR; d\w)$ with the usual inner product, and $\cW=L^2(\RR; \frac{1}{\abs{\w}}d\w)$ with the inner product
\[\ip{f_1}{f_2}_{\cW} = \int_{\RR}\hf_1(\w) \overline{\hf_2(\w)}\frac{1}{\abs{\w}}d\w.\]
 %If we base the tensor product on a real-valued basis in the frequency domain, then 
By Example \ref{Ex_tens_L2R}, up to an isometric isomorphism, we have $[f\otimes s](\w_1,\w_2) = \overline{\hf(\w_1)}\hs(\w_2)$. Moreover,  $\cW\otimes\cS$ is isometrically isomorphic to $L^2(\RR^2; \frac{1}{\abs{\w_2}}d\w_1d\w_2)$, and for $F_1,F_2\in \cW\otimes\cS$, we have 
	\[\ip{F_1}{F_2} = \int_{\RR^2} F_1(\w_1,\w_2)\overline{F_2(\w_1,\w_2)} \frac{1}{\abs{\w_2}}d\w_1d\w_2.\]
	Note that in this example the window-signal space is interpreted as a frequency domain function space. 
\end{example} 

\subsection{The $quantity_m$ transforms of window-signal spaces}

In this subsection we generalize Example \ref{ex_waveWS} for simply dilated SPWTs.   
Consider the $quantity_m$ transform, restricted to the domain \textcolor{black}{$(\cW\cap\cS)$, namely,} $\Psi_m:(\cW\cap\cS)\rightarrow L^2(Y_m\times N_m)$. Consider the positive operator $\Psi_m K \Psi_m^*$ on $\Psi_m (\cW\cap\cS)\subset L^2(Y_m\times N_m)$. Similarly to the construction of the window space, we define the inner product $\ip{\Psi_m K \Psi_m^* F}{\Psi_m K \Psi_m^* Q}$ in $\Psi_m(\cW\cap\cS)$, and extend this separable inner product space to a Hilbert space $L^2(Y_m\times N_m; \cW)$. We formally treat $L^2(Y_m\times N_m; \cW)$ as a space of functions, as it contains a dense subspace of functions. We extend $\Psi_m$ to an isometric isomorphism $\Psi'_m:\cW\rightarrow L^2(Y_m\times N_m; \cW)$. For symmetry of the notation, we denote the $quantity_m$ space of the signal space by $L^2(Y_m\times N_m; \cS):=L^2(Y_m\times N_m)$. We denote 
\begin{equation}
L^2\big((Y_m\times N_m)^2\big):= L^2(Y_m\times N_m; \cW) \otimes L^2(Y_m\times N_m; \cS).
\label{eq:quant_m_dom_WS}
\end{equation}
 We define $\Psi'_m\otimes \Psi_m$ on simple tensors by $[\Psi'_m\otimes \Psi_m](f\otimes s)=(\Psi_m' f)\otimes(\Psi_m s)$, and extend to an isometric isomorphism $\Psi'_m\otimes \Psi_m: \cW\otimes\cS \rightarrow L^2\big((Y_m\times N_m)^2\big)$. Consider two bases of %real valued 
 functions $\{\phi_n\}_{n\in\NN}$ and $\{\eta_n\}_{n\in\NN}$ of $L^2(Y_m\times N_m; \cW)$ and $L^2(Y_m\times N_m; \cS)$ respectively. Consider the bases $\{\Psi^{\prime *}_m\phi_n\}_{n\in\NN}$ and $\{\Psi_m^*\eta_n\}_{n\in\NN}$ of $\cW$ and $\cS$ respectively. Similarly to Example \ref{Ex_tens_L2R}, for the corresponding $\otimes$ operator, we define $\otimes$ explicitly by
\begin{equation}
[f\otimes_m s](g_m',y_m',g_m,y_m) = \overline{\Psi_m' f(g_m',y_m')}\Psi_m s(g_m,y_m).
\label{eq:L2WS2}
\end{equation}
The function notation in (\ref{eq:L2WS2}) is formal, and true for $f,s\in \cW\cap \cS$. Note that the explicit definition of $\otimes$ as $\otimes_m$ depend on $m$. We generally denote by $\otimes$ the abstract tensor product, and by $\otimes_m$ the explicit definition corresponding to $quantity_m$.

\subsection{A wavelet Plancherel transform}
\label{A wavelet Plancherel transform}

Given a window-signal space, with a $\otimes$ operator, there is a canonical extension of the wavelet transform to the window-signal space. Similarly to the classical Fourier theory in $L^2$, the extended  wavelet transform is first defined explicitly on a dense subspace, and then extended to the whole space via boundedness.

Given $f\in \cW\cap \cS$ and $s\in \cS$, the mapping $V:\cW\otimes\cS\rightarrow L^2(G)$ is defined for any simple tensor $f\otimes s$ with $f\in \cW\cap\cS$ by
\[V(f\otimes s) = V_f[s].\]
The mapping $V$ is defined on the space $[(\cW\cap\cS)\otimes\cS]_0$ of finite linear combinations of such simple tensors by linearity. Note that $[(\cW\cap\cS)\otimes\cS]_0$ is dense in $\cW\otimes\cS$.
Moreover, by Proposition \ref{prop_wavelet1}, $V$ is an isometry. It can thus be extended to $V:\cW\otimes\cS\rightarrow L^2(G)$, which we call the wavelet-Plancherel transform. 

\begin{theorem}[Wavelet Plancherel Theorem]
\label{Th:WP}
The wavelet-Plancherel transform $V$ is an isometric isomorphism between $\cW\otimes\cS$ and $V(\cW\otimes\cS)$.
\end{theorem}

\textcolor{black}{From an abstract continuous wavelet theory point of view, Theorem \ref{Th:WP} may seem like a triviality. For example, the image space $V(\cW\otimes\cS)$ is treated in \cite[Theorem 2.33]{Fuhr_wavelet}. However, we choose to write this result as a theorem, not for the pure mathematical novelty, but since we believe that it is fundamental and useful. When interpreting $V$ as a ``Plancherel transform,'' we open the door to practical signal processing applications akin to classical Fourier analysis. We believe that the Wavelet Plancherel Theorem deserves a place in the toolkit of the continuous wavelet practitioner.}

It may seem at this stage that the wavelet Plancherel extension of the classical wavelet theory does not give any advantage over the classical theory. Indeed, the wavelet Plancherel theorem only states that $V$ is an isometric isomorphism to its images space, a property that $V_f$ already exhibits. However, as will be shown in Proposition \ref{prop:VWS}, the function space $V(\cW\otimes\cS)$ has a much more accessible characterization than $V_f[\cS]$, which is characterized as a reproducing kernel Hilbert space (see e.g. \cite{Fuhr_wavelet}). Moreover, $V(\cW\otimes\cS)$ is invariant under multiplication by a rich class of functions, whereas $V_f[\cS]$ is not.

\begin{example} 
\label{Example:1Dwave453}
In the 1D wavelet-Plancherel transform, $V(\cW\otimes\cS)= L^2(G)$, by Proposition \ref{prop:VWS} ahead, with center $Z=\{e\}$ of the 1D affine group.
The wavelet-Plancherel transform and its inverse are given by (\ref{eq:1Dwav_slope}) and Subsection \ref{Inversion formula for the wavelet-Plancherel transform}.
%
%The wavelet-Plancherel transform $V(f\otimes s)$ of a simple tensor is given by 
%\[V(f\otimes s)(g_1,g_2)=\int_{\RR} \overline{e^{-i g_1 \w}e^{\frac{1}{2}g_2}\hf(e^{g_2}\w)}s(\w)d\w.\]
%Thus, by linear closure, for any $F\in\cW\otimes\cS$,
%\begin{equation}
%V(F)(g_1,g_2)=\int_{\RR} e^{i g_1 \w}e^{\frac{1}{2}g_2}F(e^{g_2}\w,\w)d\w.
%\label{eq:wav_slope}
%\end{equation}

The pull-back of the classical synthesis $V_f^*:L^2(G)\rightarrow \cS$ to $V_f^*V:\cW\otimes\cS\rightarrow \cS$ is given by
\begin{equation}
V_f^*V F (\w) = \int_{\RR} F(\w',\w){f(\w')}d\w'.
\label{eqg:rr6y}
\end{equation}
Indeed, for simple tensors $[q\otimes s](\w',w)=\overline{q(\w')}s(\w)$,
\[[V_f^*V q\otimes s](\w)= [V_f^*V_q s] (\w) = \ip{f}{q}_{\cW}s(\w) = \int f(\w')\overline{q(\w')}s(\w) d\w'=\int f(\w')[q\otimes s](\w',\w) d\w'.\]
%Identity (\ref{eqg:rr6y}) gives an interpretation for the window signal space. The function $F\in\cW\otimes\cS$ can be interpreted as a rule for how different windows ``see'' $F$ as a signal. Namely, viewed by the window $f$, the function $F$ takes the form of (\ref{eqg:rr6y}).
\end{example}

\begin{remark}
\label{remark:WP_invv}
\textcolor{black}{
The inversion formula of the wavelet Plancherel transform can be treated formally by identifying the tensor product space $\cW\otimes\cS$ with the space of Hilbert–Schmidt operators described in Remark \ref{remark_tensor_op}. 
The adjoint $V^*:L^2(G)\rightarrow\cW\otimes\cS$ can be written as the operator
\textcolor{black}{
\begin{equation}
   [V^*Q]_{\rm op} = \int_G Q(g)\pi(g)d\mu(g)
    \label{eq:inver_WP}
\end{equation}}
for  any $Q\in L^2(G)$.
\textcolor{black}{Note that (\ref{eq:inver_WP}) is an integral of operators, defined in the weak sense. By Remark \ref{remark_tensor_op}, the operator in (\ref{eq:inver_WP}) represents an element in $\cW\otimes\cS$. For signal and window  bases $\{\eta_n\}_n$ and  $\{\phi_m\}_m$ respectively, the right-hand-side of (\ref{eq:inver_WP}) is identified with the window-signal
\begin{equation}
  V^*Q=\sum_{n,m} \ip{\eta_m}{[V^*Q]_{\rm op}\phi_m}\phi_m\otimes\eta_n = \sum_{n,m}\int_G \ip{\eta_m}{Q(g)\pi(g)\phi_n}d\mu(g)\  \phi_m\otimes\eta_n.
    \label{eq:inver_WP01}
\end{equation}
} Equations (\ref{eq:inver_WP}) and (\ref{eq:inver_WP01}) are interpreted as an exhaustive characterization of $Q$, describing how $Q$ is synthesized to different signals using different windows. \textcolor{black}{Namely, given the window $f$, the view $[V^*Q]_{\rm op} f$ of the operator $[V^*Q]_{\rm op}$ is equal to $V_f^* Q$.}
}

\textcolor{black}{
Formula (\ref{eq:inver_WP}) is related to the inverse Plancherel transform for nonunimodular groups (see \cite[Theorem 3.48]{Fuhr_wavelet}), where in (\ref{eq:inver_WP}) we restrict to one discrete series representation and consider operators $\cW\rightarrow\cS$ instead of operators $\cS\rightarrow\cS$. 
}
\end{remark}

\section{Wavelet Plancherel phase space filtering}
\label{Wavelet Plancherel phase space filtering}

In this section we show how to calculate the pull-back of phase space multiplicative operators to the window-signal space. 

As a first step towards the pull-back of phase space multiplicative operators, we consider building blocks of such operators, namely phase space observables. To illustrate the idea, consider the 1D wavelet transform. Consider the two operators $\breve{N}_1$ and $\breve{N}_2$, defined for functions $F$ in a properly defined dense subspace\footnote{\textcolor{black}{The domain of $\breve{N}_1$ is $\left\{F\in L^2(G)\ |\ \big((g_1,g_2)\mapsto g_1 F(g_1,g_2)\big)\in L^2(G)\right\}$, and similarly for $\breve{N}_2$.}} of $L^2(G)$ by 
\[[\breve{N}_1 F](g_1,g_2)= g_1 F(g_1,g_2) \quad  , \quad [\breve{N}_2 F](g_1,g_2)= g_2 F(g_1,g_2).\] 
Note that such a definition is only possible since the 1D wavelet transform is based on a group of tuples of numerical values (it is a SPWT), \textcolor{black}{so multiplying the value $F(g_1,g_2)$ with the coordinate $g_1$ is well defined}. The operators $\breve{N}_1 ,\breve{N}_2$ multiply each point in phase space by the value of its physical quantity. It is customary in quantum mechanics to call such operators the observables of $qunatity_1$ and $quantity_2$ \cite{HAPS,quantum_measure}. In our case,  $\breve{N}_1 ,\breve{N}_2$ are the \textit{time} and the \textit{scale} observables. Now, consider a general time-multiplicative operator $\breve{R}_1$, defined by 
\[[\breve{R}_1 F](g_1,g_2)= R_1(g_1) F(g_1,g_2),\] 
where $R_1:\RR\rightarrow\CC$ is a measurable function.
 This operator can be written as $\breve{R}_1 = R_1(\breve{N}_1)$, where the composition of $R_1$ on $\breve{N}_1$ is understood in the functional calculus sense (see Remark \ref{Band_limit_poly} in Appendix \ref{Functional calculus}). Suppose we are able to pull back $\breve{N}_1$ to $\cW\otimes\cS$, and get the operator $V^*\breve{N}_1V$. Then, by (\ref{eq:pull_func_calc}) we have
\[R_1(V^*\breve{N}_1V)=V^*R_1(\breve{N}_1)V= V^*\breve{R}_1V.\]
This shows that if there is a computational tractable way to calculate functions of $V^*\breve{N}_1V$, we can calculate phase space filters directly in the window-signal space \textcolor{black}{as $R_1(V^*\breve{N}_1V)$}. %
In Subsection \ref{The pull-back of phase space observables}, we start with a formula for pulling back observables from phase space in the case of SPWT.
In Subsection \ref{The pull-back of phase space filters for geometric wavelet transforms}, we show that there is a first order linear PDE underlying each DGWT, whose solutions provide efficient approximations of a large class of phase space filters.

\subsection{\textcolor{black}{Theoretical preparation for the pull-back proposition}}
\label{The pull-back of phase space observables2}

To formulate the pull-back proposition (Proposition \ref{prop:SPWT_pull_obs} below), along with its assumptions, we first recall some general definitions and properties from functional analysis. 

\subsubsection{Restriction and closure of operators}

	Given a linear operator $T$, if the closure of the graph of $T$ is the graph of an operator, this operator is called the closure of $T$, and we say that $T$ is closable \textcolor{black}{\cite{rudin}}. In this case, we denote the closure of $T$ by $\overline{T}$. For any vector $v$ in the domain of $\overline{T}$, there exsists a sequence of vectors $\{v_j\}_{j\in\NN}$ in the domain of $T$, such that $\lim_{j\rightarrow\infty}v_j=v$ and $\lim_{j\rightarrow\infty}Tv_j=\overline{T}v$. %The domain of $A$ is called a core, or essential domain, of $\overline{A}$. the closure  

For a linear operator $T$ in a Hilbert space $\cH$, and a subspace $\mathcal{V}$ of the domain of $T$, we denote by $T\big|_{\mathcal{V}}$ the restriction of $T$ to $\mathcal{V}$.
Recall that not every symmetric operator can be closed, or extended, to a self-adjoint operator. Especially, if we take a self-adjoint operator $T$ and restrict it to an invariant subspace $\mathcal{V}$ of $T$, we are not guaranteed that the symmetric operator $T\big|_{\mathcal{V}}$ is closable to a self-adjoint operator\footnote{
	For example, $i\frac{\partial}{\partial x}$ in $L^2(\RR)$ is self-adjoint for a properly defined subspace, and symmetric for smooth function with support in $\RR_+$. However, this restriction of $i\frac{\partial}{\partial x}$ cannot be extended to a self-adjoint operator.
	}.

\subsubsection{Analysis with tuples}
For the pull-back formula, we rely on calculations of tuples of operators, vectors, and tensor products, and matrices of operators.
Fix $M\in\NN$.
 The tensor product of a vector $f$ with a tuple of vectors ${\bf s}=(s_1,\ldots,s_M)$ is the tuple of vectors $f\otimes{\bf s}=(f\otimes s_1,\ldots,f\otimes s_M)$, and similarly for a tuple ${\bf f}=(f_1,\ldots, f_M)$ and a vector $s$. The tensor product of a tuple of vectors ${\bf f}$ with the tuple ${\bf s}$ is defined to be ${\bf f}\otimes{\bf s}=(f_1\otimes s_1,\ldots,f_M\otimes s_M)$. An operator valued matrix ${\bf A}$ is an $M\times M$ array of operators $A_{j,k}$, and the application of ${\bf A}$ on the tuple of vectors ${\bf F}= (F_1,\ldots, F_M)$ is defined to be the tuple of vectors ${\bf A}{\bf F}$ with entries
\[[{\bf A}{\bf F}]_j = \sum_{k=1}^M A_{j,k}F_k\]
for $j=1,\ldots,M$.

\subsection{\textcolor{black}{Phase space observables, pull-back, and invariance properties}}
Consider a simply dilated SPWT, with the notations of Definition \ref{ass_voice2}. Let $m=1,\ldots,M$, and consider the subgroup $N_m= G_m^1\times\ldots\times G^{K_m}_m$. For each $k=1,\ldots,K^m$, consider the self-adjoint observable $\breve{G}_m^k$ defined by
\begin{equation}
[\breve{G}_m^k F](g) = g_m^k F(g).
\label{eq:obs}
\end{equation}
on the domain of $L^2(G)$ functions such that (\ref{eq:obs}) is in $L^2(G)$.
It is beneficial to model all of the observables $\breve{G}_m^1,\ldots,\breve{G}^{K_m}_m$ ``at once'', as a tuple of operators.

\begin{definition}
\textcolor{black}{
The \emph{multi-observable} ${\bf\breve{N}}_m$, mapping functions in $L^2(G)$ to $L^2(G)^{K_m}$, is defined by 
\begin{equation}
[{\bf\breve{N}}_m F](g) = \big(g_m^1 F(g),\ldots,g_m^{K_m} F(g) \big).
\label{eq:multi_obs}
\end{equation}
The domain ${\cal D}({\bf\breve{N}}_m)$ of ${\bf\breve{N}}_m$ is the space of functions such that each entry of (\ref{eq:multi_obs}) is in $L^2(G)$.}
\end{definition}
 Note that each coordinate of ${\bf\breve{N}}_m$ is either unitary or essentially self-adjoint, as it is uniquely extended to the corresponding self-adjoint observable $\breve{G}_m^k$.
 
 \begin{definition}
 \textcolor{black}{The \emph{pulled-back multi-observable} is defined as
\begin{equation}
{\bf\bT}_m = V^*{\bf\breve{N}}_mV,
\label{eq:multi_pull01}
\end{equation}
on the domain ${\cal D}({\bf\breve{T}}_m)$ of vectors $F\in \cW\otimes\cS$ such that $VF\in {\cal D}({\bf\breve{N}}_m)$.} %${\cal D}({\bf\breve{T}}_m)=V^*{\cal D}({\bf\breve{N}}_m)$.
 \end{definition}
 For each $m=1,\ldots,M$, let $\Psi_m$ be the $quantity_m$ transform of $\pi$, and consider the $quantity_m$-domain tuple of observables ${\bf\bQ}_m$, as in Definition \ref{def:quantity_trans}. The observable ${\bf\bQ}'_m$ is defined similarly as a tuple of multiplicative operators in $L^2(Y_m\times N_m ; \cW)$.

Our goal is to show that the space $V(\cW\otimes\cS)$ is invariant under the application of observables. Let ${\bf\breve{N}}_m\big|_{V(\cW\otimes\cS)}$ be the restriction of  ${\bf\breve{N}}_m$ to  the subdomain ${\cal D}({\bf\breve{N}}_m)\cap V(\cW\otimes\cS)$. The range of ${\bf\breve{N}}_m\big|_{V(\cW\otimes\cS)}$  is $L^2(G)$. 
To prove the invariance property of $V(\cW\otimes\cS)$, we derive an explicit formula of ${\bf\breve{N}}_m\big|_{V(\cW\otimes\cS)}$ of the form
\begin{equation}
{\bf\breve{N}}_m\big|_{V(\cW\otimes\cS)} = V {\bf\bT}'_mV^*
\label{eq:temp5yr}
\end{equation}
where ${\bf\bT}'_m$ is some linear operator in $\cW\otimes\cS$. By (\ref{eq:multi_pull01}), (\ref{eq:temp5yr}), and the fact that $V^*V$ is the identity, we have
\[{\bf\bT}_m = V^*{\bf\breve{N}}_mV = V^*V {\bf\bT}'_mV^*V = {\bf\bT}'_m\]
so
\begin{equation}
{\bf\breve{N}}_m\big|_{V(\cW\otimes\cS)} = V {\bf\bT}_mV^*.
\label{eq:tempd73g60}
\end{equation}
No matter the form of ${\bf\bT}_m$, equation (\ref{eq:tempd73g60}) means that ${\bf\breve{N}}_m\big|_{V(\cW\otimes\cS)}$ maps to $V(\cW\otimes\cS)$, and thus $V(\cW\otimes\cS)$ is invariant under  ${\bf\breve{N}}_m$.

 For clarity, we would like to point the attention to the difference between ${\bf\breve{N}}_m\big|_{V(\cW\otimes\cS)}$ and $V V^*{\bf\breve{N}}V V^*$. 
Since $V$ is an isometric embedding, $VV^*$ is the orthogonal projection upon $V(\cW\otimes\cS)$, and not the identity.
Thus, the invariance property cannot be deduced from the identity 
\[V V^*{\bf\breve{N}}V V^*=V{\bf\bT}_mV^*.\]

 %We may say that $V(\cW\otimes\cS)$ is invariant under ${\bf\breve{N}}_m$ only if ${\bf\breve{N}}_m\big|{V(\cW\otimes\cS)}$ and $V V^*{\bf\breve{N}}V V^*$ are unitarily equivalent. 
 %While $V V^*{\bf\breve{N}}V V^*$ is unitarily equivalent to ${\bf\bT}_m$,

More generally, we want to show that $V(\cW\otimes\cS)$ is invariant under the application of some class of multiplicative operators, \textcolor{black}{that we call \emph{phase space filter},} to be define next.
 Consider the quotient group $G/Z$, where $Z$ is the center of $G$. The quotient group $G/Z$ has the semi-direct product structure
\[\begin{split}
G \ \ \ \ \  &= H_0 \\
H_m \ \ \  &= N_{m+1}\rtimes H_{m+1} \quad , \quad m=0,\ldots, M-2 \\
 H_{M-1} &= N_m
\end{split}\]
with the same $A_m$ automorphisms as of $G$, and the left Haar measure of the group $G/Z$
(see  \cite[Equation (64)]{MyRef}). In the next definition, we
 consider the space $L^{\infty}(G/Z)$ of essentially bounded measurable functions.
 \begin{definition}
 \textcolor{black}{Denote by ${\cal M}^{\infty}(G) \subset L^{\infty}(G)$ the space of functions $R\in  L^{\infty}(G)$ of the form 
\[R({\bf z},{\bf g}_1,\ldots,{\bf g}_M)=Q({\bf g}_1,\ldots,{\bf g}_M)\]
for some $Q\in L^{\infty}(G/Z)$. We call the space $\breve{{\cal M}}^{\infty}(G)$ of multiplicative operators by ${\cal M}^{\infty}(G)$ functions the \emph{phase space filters space}. Namely $\breve{R}\in \breve{{\cal M}}^{\infty}(G)$ if there exists $R\in{\cal M}^{\infty}(G)$ such that for every $F\in L^2(G)$
\[[\breve{R} F](g)=R(g)F(g).\]
We call such $\breve{R}$ a \emph{phase space filter}.}
 \label{def:PS_filt}
 \end{definition}
 
Proposition \ref{prop:SPWT_pull_obs} below shows that $V(\cW\otimes\cS)$ is invariant under phase space filters in $\breve{{\cal M}}^{\infty}(G)$s. Moreover, it gives an explicit formula for ${\bf\bT}_m$.

	%\url{https://en.wikipedia.org/wiki/Unbounded_operator#Closed_linear_operators}
	%
	%Given a linear operator $A$, if the closure of the graph of $A$ is thre graph of an operator, this operator is called the closure of $A$, and we say that $A$ is closable. Denote the closure of $A$ by $\overline{A}$. The domain of $A$ is called a core, or essential domain, of $\overline{A}$.
	%
	%Property: An operator admits a closure if and only if for every pair of sequences $\{x_n\}$ and $\{y_n\}$ in the domain of $A$, $D(a)$, both converging to $x$, such that both $\{A x_n\}$ and $\{A y_n\}$ convergem one has $\lim_n A x_n = \lim_n A y_n$. Don't need this one.
	%
	%
	%\url{https://en.wikipedia.org/wiki/Extensions_of_symmetric_operators}
	%
	% it is a convenient fact that every symmetric operator A is closable. That is, A has a smallest closed extension, called the closure of A. This can be shown by invoking the symmetric assumption and Riesz representation theorem. Since A and its closure have the same closed extensions, it can always be assumed that the symmetric operator of interest is closed.
	
	\subsection{The pull-back of phase space observables}
\label{The pull-back of phase space observables}

Consider a simply dilated SPWT.
For each $m=1,\ldots,M$,
consider the $\otimes=\otimes_m$ operator based on signal and window $quantity_m$ function bases, as in (\ref{eq:L2WS2}). 
Consider the tuples of multiplicative operators ${\bf \breve{Q}}'_m\otimes {\bf I},{\bf I}\otimes{\bf \breve{Q}}_m, {\bf \breve{Q}}'_m\otimes{\bf \breve{Q}}_m$ in $L^2\big((Y_m\times N_m)^2\big)$, defined by
		\begin{equation}
		\begin{split}
		[{\bf \breve{Q}}'_m\otimes {\bf I}\  F ](y_m',g_m',y_m,g_m)= & \overline{{\bf g}_m'}F(y_m',g_m',y_m,g_m)\\
		[{\bf I}\otimes{\bf \breve{Q}}_m  F ](y_m',g_m',y_m,g_m)= & {\bf g}_mF(y_m',g_m',y_m,g_m) \\
		[{\bf \breve{Q}}'_m\otimes{\bf \breve{Q}}_m  F ](y_m',g_m',y_m,g_m)= & \overline{{\bf g}_m'}{\bf g}_mF(y_m',g_m',y_m,g_m) \\
		\end{split}
		\label{eq:QXI}
		\end{equation}
		Denote by $[\cW\otimes\cS]_0$ the space of (finite) linear combinations of simple tensors in $\cW\otimes\cS$. %, under a specific construction of the $\otimes$ operator. 
Denote $V(\cW\otimes\cS)_0=V([\cW\otimes\cS]_0)$. 

		In Proposition \ref{prop:SPWT_pull_obs} below, we derive explicit formulas for the pull-back of phase space observables, and establish the invariance of $V(\cW\otimes\cS)$ under $\breve{{\cal M}}^{\infty}(G)$ filters. \textcolor{black}{The formulas are computed by induction, where the formula for ${\bf\bT}_m$ is derived by those of  ${\bf\bT}_{m'}$ with $m'=m+1,\ldots,M$, using the following definition.}
		
		\begin{definition}
		\textcolor{black}{Let ${\bf\bT}_1,\ldots,{\bf\bT}_M$ be the pulled-back multi-observables of a simply dilated SPWT. We define the \emph{composite  pulled-back multi-observables} ${\bf\bT}'_m$, for each $m=1,\ldots, M-1$, as follows.
		\begin{itemize}
	\item if $G_m^1$ is $\RR$ or $\ZZ$ then
\begin{equation}
		\begin{split}
		{\bf\breve{T}}'_m= &\ \ [\Psi_m\otimes\Psi_m]^*[{\bf I}\otimes{\bf \breve{Q}}_m][\Psi_m\otimes\Psi_m]\\
		& -{\bf A}_m({\bf\bT}_{m+1},\ldots,{\bf \bT}_M) [\Psi_m\otimes\Psi_m]^*[{\bf \breve{Q}}'_m\otimes {\bf I}][\Psi_m\otimes\Psi_m].
		\end{split}
		\label{eq:pull1_self_adjoint2}
		\end{equation}
where the expression ${\bf A}_m({\bf\bT}_{m+1},\ldots,{\bf \bT}_M)$ is an operator valued matrix, understood in the functional calculus sense (see Remark \ref{simul_A_m})
\item
if $G_m^1$ is $e^{i\RR}$ or $e^{2\pi i \ZZ/N}$ then
\begin{equation}
		{\bf\breve{T}}'_m=[\Psi_m\otimes\Psi_m]^*[{\bf \breve{Q}}'_m\otimes{\bf \breve{Q}}_m][\Psi_m\otimes\Psi_m].
		\label{eq:pull1_self_adjoint2a1}
		\end{equation}
	\end{itemize}
		For $m=M$, the composite  pulled-back multi-observable ${\bf\bT}'_M$ is defined by (\ref{eq:pull1_self_adjoint2}) or (\ref{eq:pull1_self_adjoint2a1}) with ${\bf A}_M({\bf\bT}_{M+1},\ldots,{\bf \bT}_M)$ defined formally as the identity $\mathbf{I}$.}
\end{definition}
		
\begin{remark}
\label{simul_A_m}
Note that all of the spectral projections of all of the multiplicative operators  \newline
$\{\breve{G}_m^1,\ldots \breve{G}_m^{K_m}\}_{m=1}^M$ commute. By the fact that $V$ is an isometric isomorphism of $\cW\otimes\cS$ to an invariant space of all of the $\{\breve{G}_m^1,\ldots \breve{G}_m^{K_m}\}_{m=1}^M$ operators, by pull-back, all spectral projections of operators in all tuples $\{{\bf \bT}_m\}_{m=1}^M$ also commute. 
This allows us to write a simultaneous spectral decomposition of $\{{\bf \bT}_m\}_{m=1}^M$, as a ``Fubini'' integral. Namely,
\begin{equation}
{ \{{\bf \bT}_m\}_{m=1}^M} = \iint_{G/Z} (\lambda_m^1,\ldots ,\lambda_m^{K_m})_{m=1}^M \Pi_{m=1}^M\big(dP(\lambda_m^1)\ldots dP(\lambda_m^{K_m})\big),
\label{eq:simul_A_m1}
\end{equation}
where the ``infinitesimal'' spectral projections in (\ref{eq:simul_A_m1}) are compositions of ``infinitesimal'' spectral projections of the form $dP(\lambda_m^k)$, \textcolor{black}{as explained more rigorously in Remark \ref{R3e3}}.
The simultaneous application of ${\bf A}_m$ on ${\bf\bT}_{m+1},\ldots,{\bf \bT}_M$ in (\ref{eq:pull1_self_adjoint}) and (\ref{eq:pull1_self_adjoint2}) is thus defined as the matrix of commuting operators
\begin{equation}
\iint_{H_m} {\bf A}_{m}(\lambda_{m'}^1,\ldots \lambda_{m'}^{K_{m'}})_{m'=m+1}^M \Pi_{{m'}=m+1}^M\big(dP(\lambda_{m'}^1)\ldots dP(\lambda_{m'}^{K_{m'}})\big).
\label{eq:simul_A_m2}
\end{equation}
Equation (\ref{eq:simul_A_m2}) is taken on a properly defined domain, as defined in Remark \ref{Band_limit_poly}.
\end{remark}

\textcolor{black}{Next, we derive a formula for ${\bf\bT}'_m$ applied on simple tensors $f\otimes s$.}
\begin{claim}
\textcolor{black}{Let $m=1,\ldots,M$, and $f,s \in {\Psi}^{\prime *}_m{\cal D}({\bf\bQ}'_m)\cap {\Psi}_m^*{\cal D}({\bf\bQ}_m)\subset \cW\cap\cS$. Then for any construction of $\otimes$, $f\otimes s\in {\cal D}({\bf\breve{T}}'_m)$, and the following holds.
\begin{itemize}
	\item If $G_m^1$ is $\RR$ or $\ZZ$ then
	\begin{equation}
\begin{split}
 & {\bf\bT}'_m (f\otimes s) =\\
 & f\otimes[\Psi_m ^*{\bf\breve{Q}}_m\Psi_m s] - {\bf A}_m({\bf\bT}_{m+1},\ldots,{\bf \bT}_M)\Big( (\Psi_m ^*{\bf\breve{Q}}_m\Psi_m f)\otimes s  \Big).
\end{split}
\label{eq:pull1_self_adjoint}
\end{equation}
\item
If $G_m^1$ is $e^{i\RR}$ or $e^{2\pi i \ZZ/N}$ then
\begin{equation}
{\bf\bT}'_m (f\otimes s) = (\Psi_m^*{\bf\breve{Q}}_m \Psi_m f)\otimes(\Psi_m^*{\bf\breve{Q}}_m \Psi_m s).
\label{eq:pull1_unitary}
\end{equation}
\end{itemize}}
\end{claim}

\textcolor{black}{In Proposition \ref{prop:SPWT_pull_obs} next, we  also  prove that ${\bf\bT}_m={\bf\bT}'_m$ for every $m=1,\ldots,M$, under some conditions.}		
\begin{proposition}
\label{prop:SPWT_pull_obs}
Consider a simply dilated SPWT, and denote
%\textcolor{black}{${\bf\bT}_m\big|_{[\cW\otimes\cS]_0}=V^*{\bf\breve{N}}_m\big|_{V(\cW\otimes\cS)_0}V\big|_{[\cW\otimes\cS]_0}$.}
\textcolor{black}{\[{\bf\bT}_m\big|_{[\cW\otimes\cS]_0}=V^*{\bf\breve{N}}_mV\big|_{[\cW\otimes\cS]_0} =V^*{\bf\breve{N}}_m\big|_{V(\cW\otimes\cS)_0}V\big|_{[\cW\otimes\cS]_0}.\]}
\textcolor{black}{Let ${\bf\bT}_1,\ldots,{\bf\bT}_M$ be the pulled-back multi-observables, and ${\bf\bT}'_1,\ldots,{\bf\bT}'_M$ the composite pulled-back multi-observables.}
Then the following holds.
\begin{enumerate}
	\item 
	\textcolor{black}{We have ${\bf\bT}_M={\bf\bT}'_M$,} and ${\bf\breve{N}}_M\big|_{V(\cW\otimes\cS)}$ is the self-adjoint or unitary \textcolor{black}{operator given by $V{\bf\bT}_MV^*$.}
	\item
	\textcolor{black}{Let $m=1,\ldots,M-1$. If for every $m'=m+1,\ldots,M$ the closure of}
	\textcolor{black}{${\bf\bT}_{m'}\big|_{[\cW\otimes\cS]_0}$ is the  self-adjoint or unitary operator ${\bf\bT}'_{m'}$, then ${\bf\bT}_m\big|_{[\cW\otimes\cS]_0}={\bf\bT}'_m\big|_{[\cW\otimes\cS]_0}$, and ${\bf\breve{N}}_{m}\big|_{V(\cW\otimes\cS)_0}$ is the symmetric or isometric operator in $V(\cW\otimes\cS)$ given by 
	\[{\bf\breve{N}}_{m}\big|_{V(\cW\otimes\cS)_0} = V{\bf\bT}_m\big|_{[\cW\otimes\cS]_0}V^*.\].}
	%Moreover,
	%for $f,s \in {\Psi}^{\prime *}_m{\cal D}({\bf\bQ}'_m)\cap {\Psi}_m^*{\cal D}({\bf\bQ}_m)\subset \cW\cap\cS$, \textcolor{black}{${\bf\bT}_m (f\otimes s) = {\bf\bT}'_m (f\otimes s)$, with formulae (\ref{eq:pull1_self_adjoint}) or (\ref{eq:pull1_unitary})}. Then for any construction of $\otimes$, $f\otimes s\in {\cal D}({\bf\breve{T}}_m)$, and 
%\begin{itemize}
%	\item if $G_m^1$ is $\RR$ or $\ZZ$ then
%	\begin{equation}
%\begin{split}
% & {\bf\bT}_m (f\otimes s) =\\
% & f\otimes[\Psi_m ^*{\bf\breve{Q}}_m\Psi_m s] - {\bf A}_m({\bf\bT}_{m+1},\ldots,{\bf \bT}_M)\Big( (\Psi_m ^*{\bf\breve{Q}}_m\Psi_m f)\otimes s  \Big)
%\end{split}
%\label{eq:pull1_self_adjoint}
%\end{equation}
%\item
%if $G_m^1$ is $e^{i\RR}$ or $e^{2\pi i \ZZ/N}$ then
%\begin{equation}
%{\bf\bT}_m (f\otimes s) = (\Psi_m^*{\bf\breve{Q}}_m \Psi_m f)\otimes(\Psi_m^*{\bf\breve{Q}}_m \Psi_m s).
%\label{eq:pull1_unitary}
%\end{equation}
%\end{itemize}
\end{enumerate}
\textcolor{black}{Suppose that the closure of ${\bf\bT}_m\big|_{[\cW\otimes\cS]_0}$ is the  self-adjoint or unitary operator ${\bf\bT}'_m$ for every $m=1,\ldots,M-1$. Then the following holds.}
\begin{enumerate}
	\item 
	The subspace $V(\cW\otimes\cS)\subset L^2(G)$ is invariant under $\breve{{\cal M}}^{\infty}(G)$ operators.
	\item
	\textcolor{black}{For any $m=1,\ldots,M$ and  
 any construction of $\otimes$, ${\bf\bT}_m={\bf\bT}'_m$, and ${\bf\breve{N}}_m\big|_{V(\cW\otimes\cS)}$ is the self-adjoint or unitary operator given by $V{\bf\breve{T}}_mV^*$.}
%\begin{itemize}
%	\item if $G_m^1$ is $\RR$ or $\ZZ$ then
%\begin{equation}
%		\begin{split}
%		{\bf\breve{T}}_m= &\ \ [\Psi_m\otimes\Psi_m]^*[{\bf I}\otimes{\bf \breve{Q}}_m][\Psi_m\otimes\Psi_m]\\
%		& -{\bf A}_m({\bf\bT}_{m+1},\ldots,{\bf \bT}_M) [\Psi_m\otimes\Psi_m]^*[{\bf \breve{Q}}'_m\otimes {\bf I}][\Psi_m\otimes\Psi_m].
%		\end{split}
%		\label{eq:pull1_self_adjoint2}
%		\end{equation}
%where the expression ${\bf A}_m({\bf\bT}_{m+1},\ldots,{\bf \bT}_M)$ is an operator valued matrix, understood in the functional calculus sense (see Remark \ref{simul_A_m})
%\item
%if $G_m^1$ is $e^{i\RR}$ or $e^{2\pi i \ZZ/N}$ then
%\begin{equation}
%		{\bf\breve{T}}_m=[\Psi_m\otimes\Psi_m]^*[{\bf \breve{Q}}'_m\otimes{\bf \breve{Q}}_m][\Psi_m\otimes\Psi_m].
%		\label{eq:pull1_self_adjoint2a1}
%		\end{equation}
%	\end{itemize}
%	In both cases
%	\[ {\bf\breve{N}}_m\big|_{V(\cW\otimes\cS)}=V{\bf\breve{T}}_mV^*.\]
\end{enumerate}
\end{proposition}

The proof of this proposition is in Appendix B.
Formulas (\ref{eq:pull1_self_adjoint}) and (\ref{eq:pull1_self_adjoint2}) can be calculated by induction, starting at $m=M$, where there is no $A_M$ automorphism, and advancing $m$ backwards.
\textcolor{black}{
	Note that closability of an isometric operator to a unitary operator is always guaranteed. Also note that every symmetric operator is closable, so the only condition to check in Proposition \ref{prop:SPWT_pull_obs} is that the closure of the symmetric operators ${\bf\bT}_m\big|_{[\cW\otimes\cS]_0}$ is self-adjoint.}
	\textcolor{black}{The condition that the closure of ${\bf\breve{T}}_m\big|_{[\cW\otimes\cS]_0}$ is self-adjoint in $\cW\otimes\cS$ may seem hard to check in an abstract setting. However, our main goal in this paper is to use the pull-back formulas of the phase space observables in a computational machinery. For the resulting formulation to be computationally tractable, explicit formulas of these pull-backs must be obtained. Then, checking that the closure of ${\bf\breve{T}}_m\big|_{[\cW\otimes\cS]_0}$ is self-adjoint is a simple matter of checking if the resulting explicit formulas are self-adjoint operators.}
\textcolor{black}{	We note that for DGWT the closability to a self-adjoint operator condition is always fulfilled, as we prove in Proposition \ref{pullDGWT}.}

%\textcolor{blue}{Perhaps formulate a claim for this type of functional calculus computation.}

\begin{remark}
All of the pulled-back observables $\{{\bf \breve{T}}_m\}_{m=1}^M$ commute. Thus, we do not have a ``Heisenberg uncertainty'' in $\cW\otimes\cS$, and simultaneous localization in all $quantity_m$ observables is possible.
\end{remark}

\subsection{Characterization of the range of the wavelet-Plancherel transform}

We can now give an explicit characterization of the space $V(\cW\otimes\cS)$. A character of a group $B$ is a homomorphism from $B$ to the group $\{e^{i\RR},\cdot\}$.

\begin{proposition}
\label{prop:VWS}
Consider a simply dilated SPWT based on the group $G$.  Suppose that the closure of ${\bf\bT}_m\big|_{[\cW\otimes\cS]_0}$ is the  self-adjoint or unitary operator given by (\ref{eq:pull1_self_adjoint2}) or (\ref{eq:pull1_self_adjoint2a1}) for every $m=1,\ldots,M-1$.
\begin{enumerate}
	\item 
	The center $Z$ of $G$ is compact, and there exists a character $\chi:Z\rightarrow e^{i\RR}$ such that $\pi_z({\bf z})=\chi({\bf z}){\bf I}$, where ${\bf I}$ is the identity operator in $\cS$.
	\item
	The space $V(\cW\otimes\cS)$ is the space of functions $F\in L^2(G)$ for which there exists a function $Q\in L^2(G/Z)$ such that
	\begin{equation}
	F({\bf z},{\bf g}_1,\ldots,{\bf g}_M) = \overline{\chi({\bf z})}Q({\bf g}_1,\ldots,{\bf g}_M).
	\label{eq:prop:VWS}
	\end{equation}
\end{enumerate}
\end{proposition}

The proof of this proposition is in Appendix B.

\begin{remark}
\label{no_closure_needed}

In case $V(\cW\otimes\cS)=L^2(G)$, $V(\cW\otimes\cS)$ is trivially invariant under ${\bf\breve{N}}_m$, and more generally under  $\breve{{\cal M}}^{\infty}$ operators, so the problem of checking that the closure of ${\bf\bT}_m\big|_{[\cW\otimes\cS]_0}$ is self-adjoint is avoided. Indeed, $V$ is an isometric isomorphism, and ${\bf\bT}_m = V^*{\bf\breve{N}}_m V$ is unitarily equivalent to the self-adjoint or unitary operator $\breve{N}$, and is thus unitary or self-adjoint.

 It is thus worthwhile to classify the case $V(\cW\otimes\cS)=L^2(G)$. It can be shown that for groups with trivial center, $V(\cW\otimes\cS)$ is invariant under $\breve{{\cal M}}^{\infty}$ operators if and only if $V(\cW\otimes\cS)=L^2(G)$. The sufficiency direction ``$\Leftarrow$'' is trivial. The necessity direction ``$\Rightarrow$'' can be proved similarly to Proposition \ref{prop:VWS}. 
\end{remark}

In Proposition \ref{pullDGWT} we show that for DGWTs $V(\cW\otimes\cS)=L^2(G)$, and thus no condition on the closure of ${\bf\bT}_m\big|_{[\cW\otimes\cS]_0}$ is needed.
An important example where the center in not trivial is the STFT, studied in Appendix \ref{Examples222r}. There, it can be checked directly that the closure of the operators ${\bf\bT}_m\big|_{[\cW\otimes\cS]_0}$ are self-adjoint or unitary.

\subsection{Phase space filters via trigonometric polynomials of observables}
\label{Phase space filters via trigonometric polynomials of observables}

Consider a simply dilated SPWT.
For each $m=1,\ldots,M$ and $k=1,\ldots,K_m$, consider the following class of functions
\begin{definition}
A measurable function $R_m^k\in L^{\infty}(G)$ is called $quantity_m^k$ periodic, if the following two conditions are met.
\begin{enumerate}
	\item 
	The function $R_m^k$ depends only on the variable $g_m^k$.
	\item
	The restriction $R_m^k|_{G_m^k}$ of $R_m^k$ to $G_m^k$ is periodic.
\end{enumerate} 
\end{definition}
Note that if $G_m^k$ is $e^{i\RR}$ or $e^{2\pi i \ZZ/N}$, then any $L^{\infty}(G)$ function that depends only on $g_m^k$ is periodic.

\textcolor{black}{In Subsections \ref{Coefficient search in continuous wavelet systems} and \ref{Coefficient search via a wavelet-Plancherel theory} we presented a search algorithm for large 1D CWT atoms. The algorithm is based on pulling back to $\cW\otimes\cS$ the computation of time-pass and scale-pass filters in $L^2(G)$. Here, time-pass and scale-pass filters are multiplicative operators by characteristic functions of scale and time intervals in the time-scale phase space $G\cong\RR^2$. In Section \ref{Greedy sparse continuous wavelet transform} we  develop this algorithm in more detail. In particular, the time-pass filters that we use are periodic. Next, we explain how to discretize general periodic phase space filters.} %This is adequate, since for our discrete space of signals, we consider signals with periodicity in phase space corresponding to the periodicity of the filter, multiplied by an $L^2$ envelope.  

%Our computational machinery is based on periodic band-pass filters, since they can be discretized effectively, as explained next.

Consider a piecewise smooth $quantity_m^k$ periodic function $R_m^k$, and the corresponding phase space filter
\[\breve{R}_m^k F({\bf g})=R_m^k(g_m^k) F({\bf g}).\]
Consider the classical Fourier series on $e^{i\RR}$ for continuous physical quantities, and on $e^{2\pi i \ZZ/N}$ for discrete physical quantities. By Parseval's identity, we can approximate $R_m^k$ on one period using a trigonometric polynomial
\[R_m^k\approx \sum_{l=-L}^L c_l e_l=Q_L.\]
Here, $e_l$ are the complex exponential Fourier basis functions, and $c_l\in\CC$ are coefficients. Note that for real physical quantities we have $e_l:= \frac{1}{\sqrt{2\pi}}e^{i l (\cdot)}$, and for physical quantities in the unit complex sphere we have $e_l:= (\cdot)^l$ up to normalization.
Next, we show that the periodic phase space filter $\breve{R}_m^k=R_m^k(\breve{G}_m^k)$  can be approximated in some sense by
\begin{equation}
Q_L({\breve{G}}_m^k)= \sum_{l=-L}^L c_l e_l({\breve{G}}_m^k).
\label{eq:trigo_poly_filt}
\end{equation}

We are interested in the approximation of $\breve{R}_m^k F$ by $Q_L({\breve{G}}_m^k) F$ for bounded functions $F\in V(\cW\otimes\cS)$. We focus on bounded functions, since for $f\in\cW\cap\cS$, and $s\in\cS$, by Cauchy Schwarz inequality $\norm{V_f[s]}_{\infty}\leq \norm{f}_{\cS}\norm{s}_{\cS}$. In addition, when performing signal processing in phase space, like in the WP4 method, one typically repeatedly apply piece-wise smooth periodic phase space filters on $V(f\otimes s)$, which give bounded functions in $V(\cW\otimes\cS)$.

To show the approximation property, let $\e>0$, and consider a compact rectangle $D\subset G$ such that
$\int_{G\setminus D}\abs{F(g)}^2dg < \e$.
By assumption, $F$ is bounded in $D$. Moreover, $\breve{R}_m^k$ is piecewise smooth, so by Gibbs phenomenon \cite[Page 93]{Gibbs} \textcolor{black}{$\norm{Q_L}_{\infty}$ is bounded by $(1.01 +o_{L}(1))\norm{R_m^k}_{\infty}$, where  $o_{L}(1)$ is a function that decays in $L$.} Thus
\[\begin{split}
E_L=&\norm{\breve{R}_m^k F - Q_L({\breve{G}}_m^k)F}^2\\
=&\int_D\abs{R(g_m^k)F(g) -Q_L(g_m^k)F(g)}^2 dg \ +  \int_{G\setminus D}\abs{R(g_m^k)F(g) -Q_L(g_m^k)F(g)}^2 dg\\
\leq &C_0\int_{D\cap G_m^k}\abs{R(g_m^k) -Q_L(g_m^k)}^2 dg_m^k \ + \  C_1\int_{G\setminus D}\abs{F(g)}^2 dg =o_{L}(1)+ O(\e).
\end{split}
\]
\textcolor{black}{
where 
\[C_0 = \norm{F}_{\infty}\mu(D|_{Y_{m}^k})\]
with $\mu(D|_{Y_{m}^k})$ the volume of the restriction of $D$ to $\prod_{(m',k')\neq (m,k)}G_{m'}^{k'}$, and
\[C_1 = \norm{Q_L}_{\infty}+ \norm{R_m^k}_{\infty}.\]}
%$C_1$ depends on the bound of $R_m^k$.
 It is thus evident that $\lim_{L\rightarrow\infty}E_L=0$.

We remark that for the \textcolor{black}{WP4} search algorithm \textcolor{black}{of Subsections \ref{Coefficient search in continuous wavelet systems} and \ref{Coefficient search via a wavelet-Plancherel theory}}, band-pass filters can be reasonably approximated by very low order trigonometric polynomial, e.g., nine nonzero coefficients.
This discussion shows that the entities we want to pull-back to $\cW\otimes\cS$ are the unitary operators $e_l({\breve{G}}_m^k)$.
In the next subsection we show that there exists a linear first order PDE, whose solutions are precisely $V^*e_l({\breve{G}}_m^k)V$. In Appendix \ref{Examples222r} we give closed form formulas for $V^*e_l({\breve{G}}_m^k)V$ for prominent wavelet transforms.

\subsection{The pull-back of phase space filters for geometric wavelet transforms}
\label{The pull-back of phase space filters for geometric wavelet transforms}

In this section we restrict ourselves to geometric wavelet transforms. Our goal is to derive a computationally tractable formula for the pull-back of multiplicative operators in phase space. When $\breve{R}^k_m$ is a multiplicative operator by a characteristic function of an interval in $G_m^k$, we call it a $quantity_m$-pass transforms.
 As explained in Subsection \ref{Phase space filters via trigonometric polynomials of observables}, as building blocks for approximating multiplicative operators in $position$ we show how to pull back complex exponentials of the $position$ observable.
Next, we show how to calculate the $quantity_m$-pass transforms on $\cW\otimes\cS$ and $position$ complex exponentials. 
\begin{proposition}
\label{pullDGWT}
\textcolor{black}{Consider a DGWT. Then the following holds.}
\begin{enumerate}
	\item 
	We have $V(\cW\otimes\cS)=L^2(G)$.
	\item
	For $m=2,\ldots,M$, ${\bf\bT}_m$  are self-adjoint or unitary multiplicative operators in the frequency domain, and $quantity_m$-pass filters are characteristic functions of sets in $U^2$. Moreover,
\begin{equation}
[{\bf\bT}_1 F](\bw',\bw) = i \frac{\partial}{\partial \bw}F(\bw',\bw)
 + {\bf A}_m({\bf\bT}_2,\ldots,{\bf\bT}_N)i \frac{\partial}{\partial \bw'}F(\bw',\bw).
\label{eq:time_obs11}
\end{equation}
is a self-adjoint differential operator,
where ${\bf A}_m({\bf\bT}_2,\ldots,{\bf\bT}_N)$ is a multiplicative operator-valued matrix in $L^2(U^2)$.
The unitary operator $\exp({it {\bf\bT}_1})$ is defined for $F\in\cW\otimes\cS$ via the differential equation
\begin{equation}
\begin{split}
\frac{\partial}{\partial t}\big[\exp({it {\bf\bT}_1})F\big]\ \ \ \ = \ \ &i {\bf\bT}_1 \big[\exp({it {\bf\bT}_1})F\big] = -\Big(\frac{\partial}{\partial \bw}
 + {\bf A}_m({\bf\bT}_2,\ldots,{\bf\bT}_N) \frac{\partial}{\partial \bw'}\Big)\big[\exp({it {\bf\bT}_1})F\big]\\
\big[\exp({it {\bf\bT}_1})F\big]_{t=0}= \ \ & F.
\end{split}
\label{eq:flow_pos1}
\end{equation}
Namely, $\exp({it {\bf\bT}_1})F$ is calculated by taking the flow of $F$ along the integral lines of (\ref{eq:flow_pos1}), 
and position-pass filters are convolutions along these integral lines.
\end{enumerate}

\end{proposition}

Note that the first order linear PDEs in the system (\ref{eq:flow_pos1}) are independent.
It is thus evident that $\exp({it {\bf\bT}_1})F$ is calculated by taking the flow of $F$ along the integral lines of (\ref{eq:flow_pos1}), for a distance $t$. A periodic filter of the form (\ref{eq:trigo_poly_filt}), applied on $F$, is thus a finite linear combination of transformed versions of $F$ along the integral lines, and a general $position$ filter is convolution along the integral lines. 
 In Appendix C we give closed form formulas for the solution of (\ref{eq:flow_pos1}), and its integral lines, for prominent examples of wavelet transforms.

\section{Greedy sparse continuous wavelet transform}
\label{Greedy sparse continuous wavelet transform}

\textcolor{black}{In this section we develop in detail the sparse approximation methods that was proposed in Subsection \ref{Sparse approximations with continuous wavelet systems}.} Namely, we formulate a greedy algorithm for extracting sparse approximations to signals using atoms from the 1D wavelet system. We show that, by  using the wavelet-Plancherel theory, we can search for large wavelet coefficients in a phase space grid consisting of $O(N^2)$ points in $O(N\log(N))$ operations. Here, $N$ is the resolution of the discrete signal. This search approach is utilized in an efficient matching pursuit algorithm. 

\subsection{A search algorithm via the wavelet-Plancherel theory}
\label{A search algorithm via the wavelet-Plancherel theory}

\textcolor{black}{Using the wavelet-Plancherel theory, we can describe the matching pursuit method of Subsections \ref{Coefficient search in continuous wavelet systems} and  \ref{Coefficient search via a wavelet-Plancherel theory} more accurately.} Given a signal $s$ and a window $f$, our goal is to find large wavelet coefficients of $V_f[s]$. We do this via a bisection search in $V(\cW\otimes\cS)$. We start with a compact domain $G_0$ of $G$, on which $V_f[s]$ takes most of its energy. We take $G_0$ to be a rectangle $T_0\times S_0$ in the coordinate representation $N_1\rtimes N_2$. \textcolor{black}{We partition $G_0$ to four rectangular sub domains $G_0^1=T^1_0\times S^1_0,\ldots,G_0^4=T^2_0\times S^2_0$.} These four rectangular domains are assumed to cover identical areas in a sense that will be defined precisely in Proposition \ref{partition_prop}. For each $G_0^r$, we consider an approximation $\tilde{P}_{G_0^r}$ to the projection upon $G_0^r$. The approximate projection is a multiplicative operator of the form $\tilde{P}_{G_0^r}=\tilde{P}_{T_0^r} \tilde{P}_{S_0^r}$. Here, $\tilde{P}_{T_0^r}$ is an approximate periodic time-pass, namely a trigonometric polynomial of $\breve{N}_1$ with period $\abs{T_0}$, and $\tilde{P}_{S_0^r}$ is a perfect scale-pass, namely the characteristic function of $S_0^r$ applied on $\breve{N}_2$. We then define $G_1$ as the rectangle having the greatest $\norm{\tilde{P}_{G_0^r}V[f\otimes s]}$. Using the wavelet-Plancherel theorem, this calculation is carried out in the $\cW\otimes\cS$ space using the pull-back of $\tilde{P}_{G_0^r}$, and calculating the norm in $\cW\otimes\cS$. 

Next, we show how to continue this process. At step $J$, we consider the rectangle $G_{J-1}$, partitioned into $G_{J-1}^1,\ldots,G_{J-1}^4$. We would like to define $G_{J}$ as the rectangle $G_{J-1}^r$ having the greatest $\norm{\tilde{P}_{G^r_{J-1}}V [f\otimes s]}$. However, the narrower the support of $G_{J}$ along the $\breve{N}_1$ axis is, the more Fourier coefficients are required for the periodic approximate $time$-pass filter. \textcolor{black}{To keep the number of Fourier coefficient bounded, and control the asymptotic complexity of the method, we use the following filtering scheme.} Consider the trigonometric polynomial
\begin{equation}
R_1^0(x)=\sum_{l=-L}^L c_l e_l(x)
\label{eq:T_filter}
\end{equation}
that approximates the characteristic function of $[-\pi,0]$ in $L^2[-\pi,\pi]$. Here, $e_l$ are the complex exponential Fourier basis functions, and $c_l$ are scalar coefficients (see Subsection \ref{The pull-back of phase space filters for geometric wavelet transforms}). The function $R_1^1=1-R_1^0$ approximates the characteristic function of $[0,\pi]$.  Note that $R_1^0,R_1^1$ are $2\pi$ periodic. The dilated filter
\begin{equation}
R^0_j(x)=\sum_{l=-L}^L c_j e_{2^jl}(x)
\label{eq:T_filter2}
\end{equation}
is a $2^{-j}2\pi$ periodic function, approximating the characteristic function of $[-2^{-j}\pi,0]$. The filter $R^1_j=1-R^0_j$  approximates the characteristic function of $[0,2^{-j}\pi]$. To approximate the characteristic function of some interval $2^{-J} 2\pi r+ [-2^{-J}\pi,0]$ , for some integer $r$, via a $2\pi$ periodic filter, we simply take the product $\prod_{j=1}^J{R_j^{b_j}}$, where $\{b_j\}_{j=0}^J$ is the binary representation of the number $r\in\NN$. 

We thus continue the search process, and define $G_{J+1}$ as the rectangle $G_J^r$ having the greatest 
\begin{equation}
\Big\|\tilde{P}_{G^r_J}\prod_{j=0}^J\tilde{P}_{G_j}V\ (f\otimes s)\Big\|.
\label{eq:gnf7}
\end{equation}
Note that equation (\ref{eq:gnf7}) requires only the calculations of the projections $\tilde{P}_{G^r_J}$, $r=1,2,3,4$,  which are based on $2L+1$ exponentials of $\breve{N}_1$ and perfect $scale$-passes. Indeed, the product part in (\ref{eq:gnf7}) was calculated in the previous step of the search.
 When we stop the process at a big $J=\overline{J}$, we are left with an approximation of the projection of $V_f[s]$ upon a small rectangle about the point $g^1=g^1_{\overline{J}}$ in $G$. The resulting wavelet atom $\pi(g^1)f$ has a large wavelet coefficient. This coefficient search algorithm does not necessarily give the global maximum. On the other hand, this search method is also not a local method, since all of the wavelet coefficients of $V_f[s]$ are accounted for in the search. %We argue that philosophically this search method is ``somewhere between a global and a local optimization method''.

\subsection{Continuous wavelet matching pursuit}

\textcolor{black}{Assume that $f\in\cS\cap\cW$ is normalized in $\cS$.  
%Given a wavelet atom $\pi(g^1)f$ with a large coefficient, 
 We consider in this paper two ways to compute a sparse approximation of the signal $s\in\cS$.
In Matching Pursuit (MP) \cite{Matching_pursuit}, we take the first approximation to $s$ as $s'_1=\ip{s}{\pi(g^1)f}\pi(g^1)f$, where $g^1$ is the large wavelet atom found by the wavelet-Plancherel search algorithm. We compute the residual $s_1=s-s'_1$, and repeat the process by induction for the signal $s_k$, to obtain sparse approximations $s'_k$, for $k=1,\ldots, K$. %At step $K$, we obtain a sparse approximation $s'_K$ to $s$, consisting of $K$ wavelet coefficients.
In Orthogonal Matching Pursuit (OMP) \cite{OMP}, for each $k$ we consider all of the points $g^1,\ldots,g^{k-1}$ found at steps $1,\ldots,k-1$, and define the sparse approximation $s'_k$ as the orthogonal projection of $s$ upon ${\rm span}\{\pi(g^1)f,\ldots,\pi(g^{k-1})f\}$. %We then define the residual $r_{k-1}=s-s_{k-1}$, find a wavelet atom $\pi(g_L^k)f$ with large coefficient of $f\otimes r_{k-1}$, and continue. At step $K$, we obtain a sparse approximation $s_K$ to $s$, consisting of $K$ wavelet coefficients. 
We call either of these two methods \emph{wavelet Plancherel phase-space projection pursuit}, or \emph{WP$^4$}, which we write also as \emph{WP4}.}

\begin{algorithm}[\textcolor{black}{Wavelet Plancherel phase-space projection pursuit (WP4)}]
\label{Wavelet Plancherel phase-space projection pursuit (WP4)}
$ $ \newline
\emph{Input}: signal $s\in \cS$. \newline
\emph{Output}: sparse signal $s'_K\in\cS$ approximating $s$ after $K$ iterations. \newline
\emph{Notations}: 

$\quad$ $s_k=$ signal residual

$\quad$ $s'_k=$ sparse approximation at step $k$

$\quad$ ${\rm Proj}(s; \mathcal{V})=$ orthogonal projection of the vector $s\in\cS$ to the subspace $\mathcal{V}\subset\cS$

$\quad$ $G_0=$ the rectangular search domain in phase space

\begin{itemize}
    \item 
    Initialize the sparse approximation at step $0$ as $s'_0=0$
    \item
    Initialize the residual signal at step $0$ as $s_0=0$
    \item
    Repeat for $k=1,2,\ldots$ until a stopping criterion\footnote{The stopping criterion can be error below some tolerance $\norm{s-s_k}<\epsilon$, or $k$ above some maximal number of iterations $K'$.} is satisfied, at step $K$
\begin{itemize}
    \item Find a large wavelet coefficients $g'_k$ from $V_f[s_{k-1}](G_0)$ using the wavelet-Plancherel search algorithm.
    \begin{multicols}{2}
    Version 1: Matching Pursuit
\begin{itemize}
 \item $s'_k=s'_{k-1} + V_f[s_{k-1}](g_k')\pi(g_k')f$ 

 \item $s_k=s_{k-1} - V_f[s_{k-1}](g_k')\pi(g_k')f$
\end{itemize}

    Version 2: Orthogonal Matching Pursuit 
\begin{itemize}
 \item $s'_k={\rm Proj}\big(\ s\ ; \ {\rm span}\{\pi(g'_j)f\}_{j=1}^k\big)$  
 \item $s_k=s-s_k'$
\end{itemize}
\end{multicols}
\end{itemize}
\item
Return $s'_K$
\end{itemize}

\end{algorithm}

In this section we show that by choosing a suitable discretization of $\cW\otimes\cS$, with a suitable data structure and scale bisection, we can derive a coefficient search algorithm with lower computational complexity than naive continuous wavelet methods, based on a 2D discretization of phase space. Moreover, the computational complexity of our method is the same as that of a discrete method, while squaring the sampling resolution in phase space.
 To understand the improvement in resolution of the WP4 method over standard methods, we explain in Subsection \ref{1D continuous wavelet as a time-frequency transform} how the continuous 1D wavelet transform can be interpreted as a time-frequency transform. We last present a setting in which the WP4 method is beneficial, namely wavelet phase vocoder.

We note that the issue of stability of the projections is not addressed in this paper. However, we observe that in practice, the pointwise deviations of the function underlying $\prod_{j=1}^J\tilde{P}_{G_j}$, from the characteristic function of $G_J$ do not all conspire to be in the same direction, and $\prod_{j=1}^JV^*\tilde{P}_{G_j}V$ approximate the perfect band-pass filter $V^*P_{G_j}V$ reasonably. The approximation is reasonable even for low order trigonometric polynomials, e.g., having nine nonzero coefficients. Moreover, note that a crude approximation is appropriate in our situation, since we look for big wavelet coefficients, but do not need the exact value of the projected $V_f[s]$. Indeed, after the point $g^k$ is found in phase space, the coefficient is calculated directly by $\ip{s}{\pi(g^k)f}\pi(g^k)f$ in MP, or by projection of $s$ upon ${\rm span}\{\pi(g^1)f,\ldots,\pi(g^{k-1})f\}$ in OMP.

\subsection{Standard discretization of the 1D wavelet transform}
\label{Standard discretization of the 1D wavelet transform}

In this subsection we review two standard methods, that do not rely on the wavelet-Plancherel theory, namely wavelet frame methods.
Since the construction is important for understanding the advantage of the WP4 method, we offer instructive calculations in this section, without going into technical details. For reference, see \cite{fast_wave,Painless} and \cite[Sections 4.3 and 5.2]{wavelet_tour}.

We restrict the analysis to signals $\hs$ and a windows $\hf$ supported in a compact interval of the positive frequency line $(0,\infty)$, namely, in $L^2(0,\infty)$.
We consider the wavelet transform based on the \emph{reduced affine transform}, namely, the group $A$ of (\ref{eq:Affine_group1}) restricted to $g_3=1$. It can be shown that the wavelet transform $\pi(g_1,g_2)$, defined by (\ref{eq:wave_rep}) with $g_3=1$, is a square integrable representation.
Denote by $\hat{\pi}(g_1,g_2) = \cF\pi(g_1,g_2)\cF^*$ the wavelet representation in the frequency domain.

The wavelet transform can be formulated as follows. The wavelet transform $V_f[s]$ at the fixed scale $g_2$, and variable time $g_1$, is the inverse Fourier transform of $\hs \overline{\hat{\pi}_2(g_2)\hf}$ at $g_1$, where $\hat{\pi}_2(g_2) = \cF\pi_2(g_2)\cF^*$.
Note that $[\hat{\pi}_2(g_2)\hf](\w) = e^{\frac{1}{2}g_2}\hf(e^{g_2}\w)$ is frequency dilation.
 The reconstruction formula of the wavelet transform, restricted to the fixed scale $g_2$, gives
\begin{equation}
\begin{split}
\int V_f[s](g_1,g_2)\hat{\pi}(g_1,g_2)\hf \ dg_1 & =\int V_f[s](g_1,g_2)e^{-i g_1\w}\hat{\pi}_2(g_2)\hf \ dg_1 \\
 & = \hat{\pi}_2(g_2)\hf \int V_f[s](g_1,g_2)e^{-i g_1\w}\ \ dg_1    =\hs\abs{\hat{\pi}_2(g_2)\hf}^2.
\end{split}
\label{eq:reco_g1}
\end{equation}
Integrating (\ref{eq:reco_g1}) over $g_2$, the reconstruction formula can then be written as
\begin{equation}
\hs=\hs\int\abs{\hat{\pi}_2(g_2)\hf}^2 e^{-g_2}dg_2,
\label{eq:reco_g10}
\end{equation}
which shows that $\int\abs{\hat{\pi}_2(g_2)\hf}^2 e^{-g_2}dg_2=1$. In (\ref{eq:reco_g10}), the term $e^{-g_2}$ is due to the Haar measure of the reduced affine group.
The following discrete methods are derived from these observations, using FFT as the main computational tool.

Assume that the signal $\hs$ is discretized in the frequency domain, with the equi-spaced frequency samples $0<\w_0,\w_1,\ldots,\w_N$, with $\w_{n+1}-\w_n=r$, $\w_N=O(rN)$. Assume that $\hf$ is compactly supported away from $0$, and supported about the frequency $\w_0$. Precisely, let $(\w_0-\Delta,\w_0+\Delta)$ be the support of $\hf$, and
assume that $\Delta<\w_0$, and $\Delta>r$. 
 Consider a decreasing grid in scale, 
\begin{equation}
\{-g_2^m\}_{m=0}^M \quad , \quad 0=-g_2^0>-g_2^1>\ldots>-g_2^M,
\label{eq:scale_grid1}
\end{equation}  
such that $\exp(g_2^M)\w_0\geq \w_N$. %Assume that $M=O(\log(N))$, 
The dilated window $\hat{\pi}_2(-g_2^m)\hf$ is centered about the frequency $\exp(g^m_2)\w_0$, with support 
\[(\exp(g^m_2)\w_0-\exp(g^m_2)\Delta\ ,\ \exp(g^m_2)\w_0+\exp(g^m_2)\Delta)\] 
for every 
$m=0,\ldots,M$.
Assume that $\exp(g^m_2)(\w_0+\Delta)$ is sufficiently larger than $\exp(g^{m+1}_2)\w_0$ for every 
 $m=0,\ldots,M-1$. Namely, assume there is enough overlap between the support of 
$\hat{\pi}_2(-g_2^m)\hf, \hat{\pi}_2(-g_2^{m+1})\hf$ for every $m=0,\ldots,{M-1}$, to guarantee that
\begin{equation}
\sum_{m=0}^M W_m\abs{\hat{\pi}_2(-g_2^m)\hf(\w)}^2 \approx 1
\label{eq:reco_g10D}
\end{equation}
for every $\w\in [\w_0,\w_M]$. Here, $W_m$ are weights that depend on the grid, which normalize the sum so that it approximates the integral of (\ref{eq:reco_g10}). Note that (\ref{eq:reco_g10D}) guarantees approximate reconstruction of the discrete method, in view of  (\ref{eq:reco_g10}). \textcolor{black}{Suppose we are interested in the content of the signal in the time interval $[X_{\rm l}, X_{\rm r}]$  with $X_{\rm r}-X_{\rm l} = r^{-1}$.}
For each fixed scale $-g_2^m$, we consider a regular time grid $\{g^k_1\}_k$ with number of \textcolor{black}{samples $K_m$ proportional to the support size $2\exp(g_2^m)\Delta$ of $\hat{\pi}_2(-g_2^m)\hf$, namely $K_m=C \lceil2\exp(g_2^m)\Delta\rceil$ for some $C>1$.} The approximate value of the wavelet transform at $(g_1^k,-g_2^m)$ is given by
\begin{equation}
V_f[s](g_1^k,-g_2^m)\approx \Big[\cF^{-1}e^{-\frac{1}{2}g_2^m}\overline{\hf(e^{-g_2^m}\cdot)}\hs(\cdot)\Big](g_1^k),
\label{eq:rerwe121212wd56}
\end{equation}
where $\cF^{-1}$ is the discrete inverse Fourier transform.  In (\ref{eq:rerwe121212wd56}), $\hf$ is sampled at different time points for different scales $-g_2^m$.

Next, we derive the computational complexity of this discrete wavelet method. Consider the correspondence between scale and frequency,
$\k=\w_0\exp(g_1)$ (see Subsection \ref{1D continuous wavelet as a time-frequency transform}). Let 
$\{\k_m=\exp(g_2^m)\w_0\}_{m=0}^M$ be a frequency grid, constructed by $\k_m=\kappa(m)$, where $\kappa$ is a smooth function, and $\kappa_N\geq \w_N$.
The support of $\hat{\pi}_2(-g_2^m)\hf$ is 
$(\k_m-\frac{\Delta}{\w_0}\k_m\ ,\ \k_m+\frac{\Delta}{\w_0}\k_m)$. \textcolor{black}{Denote by $n^m_{\rm l}$ the greatest $n$ such that $\w_n\leq C(\k_m-\frac{\Delta}{\w_0}\k_m)$, and by $n^m_{\rm r}$ the smallest $n$ such that $\w_n\geq C( \k_m+\frac{\Delta}{\w_0}\k_m)$. Then, by the relations $g_1^k = k\frac{X_{\rm r} - X_{\rm l}}{n^m_{\rm r}-n^m_{\rm l}} + X_{\rm l}$ and $\w_n = rn + \w_{n^m_{\rm l}}$, we can write (\ref{eq:rerwe121212wd56}) as
\[
\begin{split}
  V_f[s](g_1^k,-g_2^m)\approx & r\exp\Big(2\pi i\frac{X_{\rm r} - X_{\rm l}}{n^m_{\rm r}-n^m_{\rm l}}k\w_{n^m_{\rm l}} + 
  2\pi i X_{\rm l}\w_{n_{\rm l}^m}
  \Big) \cdot\\
& \sum_{n=0}^{n^m_{\rm r}-n^m_{\rm l}} \exp\Big(  2\pi i X_{\rm l}rn\Big)
\exp\Big(2\pi i\frac{X_{\rm r} - X_{\rm l}}{n^m_{\rm r}-n^m_{\rm l}}rnk\Big)
e^{-\frac{1}{2}g_2^m}\overline{\hf(e^{-g_2^m}\w_n)}\hs(\w_n). 
\end{split}
\]
Hence, we can compute $V_f[s](g_1^k,-g_2^m)$ using FFT by
\[
\begin{split}
  \tilde{L}(X_{\rm l})V_f[s](g_1^k,-g_2^m) \approx & r\exp\bigg(2\pi i\Big(\frac{X_{\rm r} - X_{\rm l}}{n^m_{\rm r}-n^m_{\rm l}}k - X_{\rm l}\bigg)\w_{n^m_{\rm l}} + 
 2\pi i  X_{\rm l}\w_{n_{\rm l}^m}
  \Big) \cdot\\
& \sum_{n=0}^{n^m_{\rm r}-n^m_{\rm l}}  
\exp\Big(2\pi i\frac{1}{n^m_{\rm r}-n^m_{\rm l}}nk\Big)
e^{-\frac{1}{2}g_2^m}\overline{\hf(e^{-g_2^m}\w_n)}\hs(\w_n) ,
  \end{split}
\]
where $\tilde{L}$ is the periodic translation operator.
}
Since there are $O(\frac{1}{r}\frac{\Delta}{\w_0}\k_m)$ samples in the support of $\hat{\pi}_1(-g_2^m)\hf$, the computational complexity of one inverse FFT is
$O\Big(\frac{1}{r}\frac{\Delta}{\w_0}\k_m\log(\frac{1}{r}\frac{\Delta}{\w_0}\k_m)\Big)$.
If $\k$ is smooth enough in $m\in\RR$, by the complexity of FFT, we can estimate the computational complexity by the order of
\[\frac{1}{r}\int_{m_0}^M \frac{\Delta}{\w_0}\k(m)\log(\frac{\Delta}{\w_0}\k(m))dm.\]

We now give two extremes  of this construction. The first example is the standard exponential scale grid \textcolor{black}{\cite[Section 5.3]{wavelet_tour5}}. Here, $\kappa(m)=\exp(hm)$, and the complexity is of order of
\begin{multline}
\frac{1}{r}\int_{\log(\w_0)/h}^{\log(\w_N)/h} \frac{\Delta}{\w_0}\exp(h m) (\log(\frac{\Delta}{\w_0})+h m)dm   \\
= \frac{1}{r}\frac{\Delta}{\w_0}\Big(\frac{1}{h}\exp(h m)\log(\frac{\Delta}{\w_0}) + \frac{1}{h}\exp(h m)( \a m-1) \Big) \Big|_{\log(\w_0)/h}^{\log(\w_N)/h} \\
=O(\frac{1}{h}\frac{1}{r}\frac{\Delta}{\w_0}\w_N\log(\w_N))=O(\frac{\Delta}{\w_0h}N\log(N)).
\label{Disc_complex1}
\end{multline}
Note that the overall number of samples in phase space in this method is $O(N)$.

Another example is when $\k(m)=\w_0+rm$, and we get computational complexity of order 
\[\frac{1}{r}\int_{0}^{N} \frac{\Delta}{\w_0}(\w_0+rm)\log\big(\frac{\Delta}{\w_0}(\w_0+rm)\big)dm= O(\frac{\Delta}{\w_0 }N^2\log(N))\]

Last, if we are interested in an $N\times N$ regular grid resolution in time-frequency, taking $N$ time samples for each scale, the computational complexity is $O(N^2\log(N))$. 

To conclude, in the discrete wavelet methods, if we want a computational complexity of $O(N\log(N))$, we are limited to a resolution of $O(N)$ samples in phase space, and if we want $O(N^2)$ samples in phase space, we are limited to a computational complexity of $O(N^2\log(N))$.

\subsection{Discretization of the wavelet Plancherel method}
\label{Discretization of the wavelet Plancherel method}

In this subsection we consider two discretizations of $\cW\otimes\cS$ that accommodate computations of phase space band-pass filters. We show that the first discretization leads to a naive algorithm, while the second leads to a novel algorithm, with lower computational complexity. To avoid dealing with the reflection subgroup, we assume that both the signal $\hs$ and the window $\hf$ are positively supported in the frequency domain.

Recall that for the 1D wavelet transform, $\cW\otimes\cS\cong L^2(\RR^2;\frac{1}{\abs{\w'}}d\w'd\w)$. The integral lines of the exponential of $\bT_1$, $\exp(it\bT_1)$, is given in (\ref{eq:scale_x_flow11}) and (\ref{eq:scale_x_flow22}). In view of (\ref{eq:scale_x_flow11}) and (\ref{eq:scale_x_flow22}), these integral lines are \emph{slope rays}, \textcolor{black}{as defined in (\ref{eq:slope_ray}).} 
We consider the slope variables in $\cW\otimes\cS$, $z=\x'/\x$ and $\x$, \textcolor{black}{as defined in (\ref{eq:slope3}).} In these variables, $\exp(it\bT_1)$ is a translation along $\x$ with  constant slope $z$. As a result, time-pass filters are convolutions along slopes. The band-pass filters in $\bT_2$ are multiplicative operators, depending only on the slope $z$. 

\subsubsection{Naive discretization of the window-signal space}
Let us introduce our first discretization of the coefficient search in $\cW\otimes\cS$.
Recall that the $time$ exponential $\exp(it\bT_1)$ depends only on the variable $\x$, and the $scale$-pass filters depend only on the slope variable of $z$. Thus, a natural choice for a discretization is to model functions $F\in \cW\otimes\cS$ as samples of $F$ on a uniform $\{(\x_n,z_k)\}_{n=1,k=1}^{N\ \ ,\ K}$ grid. In the initialization of the search, with $\hf\in\cW$ and $\hs\in\cS$, we define
\begin{equation}
F(\x_n,z_k)=[\hf\otimes \hs](\x_n,z_k\x_n)=\overline{\hf(z_k\x_n)}\hs(\x_n).
\label{eq:dfdfdfnmnmakfsdfk}
\end{equation}
The $scale$-pass filters simply restrict the grid to a subset of the slope samples $\{z_k\}_{k=1}^K$. The square norm of a $slope$-passed $F$ is given by 
$\sum_{k=K_{0}}^{K_{1}}\sum_{n=1}^N \abs{F(\x_n,z_k)}^2$.
A $time$-pass of $F$, based on $L$ nonzero Fourier coefficients, is a convolution of $F$ along the $\x$ direction, independently of $z$. Namely it is a linear combination of $L$ translated versions of $F$ along the slope-rays. For $K=O(N)$, one coefficient search takes $O(N^2L)$ operations.

In view of (\ref{eq:dfdfdfnmnmakfsdfk}) and (\ref{eq:rerwe121212wd56}), the above method could have been constructed based on the standard wavelet theory, without the need of the wavelet-Plancherel heavy machinery. Moreover, this method does not give a significant, if any, improvement over the $O(N^2\log(N))$ complexity of the naive method.

\subsubsection{Efficient discretization of the window-signal space}
\label{Efficient discretization of the window-signal space}

Next, we construct our second discretization of $\cW\otimes\cS$. 
This discretization is based on the tensor product of a discretized $\cW$ space, $\tilde{\cW}$, and a discretized $\cS$ space, $\tilde{\cS}$. In choosing $\tilde{\cW}$ and $\tilde{\cS}$, we follow two guidelines. First, the model should be invariant under the operations of the search algorithm. Since in the search algorithm, functions of $\tilde{\cW}\otimes\tilde{\cS}$ are translated along slope-rays, and restricted to slope-bands, we would like the space $\tilde{\cW}\otimes\tilde{\cS}$ to be invariant under these operations. Moreover, since $time$-pass filters are based on additions, we would like $\tilde{\cW}\otimes\tilde{\cS}$ to also be invariant under addition. 
The second guideline is that $\cS$ should be discretized ``finely'', to accommodate a rich class of signals. Thus we discretize $\cS$ using a uniform grid $\{\w_n\}_{n=0}^N$. However, the window can be defined ``coarsely'', as a low dimensional parametric model. We consider windows that are B-splines in the frequency domain, and specifically in this paper we restrict ourselves to piecewise linear continuous windows. We note that even a piecewise linear continuous window with three knots can model mother wavelets with a range of desirable properties. Using a window, with its first knot near $0$, we can design ``impulse-like'' mother wavelets, for detecting discontinuities. For another example, using a window  with its first knot away from $0$, we can design ``oscillatory'' mother wavelets, for detecting local frequencies.  We note that spline wavelet systems were extensively studied in the past, see, e.g., \cite{Unser:1993:PSW:160964.174020}.
We call $\tilde{\cW}\otimes\tilde{\cS}$ the space of spline sequences.

Next, we show how to interpret $\tilde{\cW}\otimes\tilde{\cS}$ as a (nonlinear) subspace of $\cW\otimes\cS$. Let $\{\hf_n\}_{n=0}^N\in \tilde{\cW}\otimes\tilde{\cS}$, where each $\hf_n$ is a linear spline. Let $d:\RR\rightarrow\RR$ be an $L^2(\RR)$ function, with support $[-\e,\e]$, such that $\e<\w_0$. We call $d$ a bump function. To embed the discrete model in $\cW\otimes\cS$, the spline sequence $\{\hf_n\}_{n=0}^N$ is mapped to the function
\begin{equation}
\{\hf_n\}_{n=0}^N(\w',\w)=\sum_{n=0}^N \hf_n(\w') d(\w-\w_n)
\label{eq:spline1}
\end{equation}
in $\cW\otimes\cS$. The tensor product of a spline window $\hf$ and a discrete signal $\{\hs(\w_n)\}_{n=0}^N$, gives the spline sequence 
\[\hf\otimes\{\hs(\w_n)\}_{n=0}^N = \{\hs(\w_n)\overline{\hf}\}_{n=0}^N\in \tilde{\cW}\otimes\tilde{\cS}.\]
In the following we do not distinguish between spline sequences and their embedding to the window-signal space.

\subsubsection{Invariance of the discrete window-signal space to search operations}
\label{Invariance of the discrete window-signal space to search operations}

Next, we show that the space $\tilde{\cW}\otimes\tilde{\cS}$ is invariant under approximate operations of the search algorithm. Let $\{\hf_n\}_{n=0}^N\in \tilde{\cW}\otimes\tilde{\cS}$. Each $\hf_n$, can be represented as the two sequences,  $\{\w'_k\}$ and $\{f_n(\w'_k)\}$, of the knots of the spline and the values of $\hf_n$ at the knots respectively. The spline $\hf_n$ may be discontinuous only at its first and its last knots, in case the values of the spline are nonzero there. In view of the embedding of spline sequences to $\cW\otimes\cS$, a $scale$-pass  filter $P^2_{[a,b]}$, upon the $scale$ interval $[a,b]$, approximately corresponds to restricting each $\hf_n$ to an interval that depends on $n$. This is true since the bump $d$ in (\ref{eq:spline1}) is narrow. Thus, an approximate $scale$-pass keeps $\{\hf_n\}_{n=0}^N$ a spline sequence, and does not increase the number of knots.

Second, $\tilde{\cW}\otimes\tilde{\cS}$ is also invariant under approximate translations along slope-rays. Assume that $F(z\w,\w)=\{\hf_n\}_{n=0}^N(z\w,\w)$ is zero for any slope $z>B$ for some $B$. Translations along slope-rays keep the $\cW$ cross-sections $F\ (\cdot\ ,{\rm const})$ of $F$ as linear splines. Moreover, if we choose $\e$ (of the bump $d$) sufficiently smaller than $\frac{1}{B}$ and $\w_0$,  
then the shear in $\exp(it\bT_1)f_n(\w') d(\w-\w_n)$ is negligible. Indeed, by the support of $d$ we consider $(\w-\w_n)=O(\e)$, so
\[\exp(it\bT_1)[f_n(\w') d(\w-\w_n)]= \hf_n(\frac{\w'}{\w}(\w-t)) d(\w-\w_n-t) \]
\[ =\hf_n(\frac{\w'}{\w_n+O(\e)}(\w_n + O(\e)-t)) d(\w-\w_n-t) \]
\[=\hf_n(\frac{\w'}{\w_n}(1+O(\frac{\e}{\w_n}))(\w_n + O(\e)-t)) d(\w-\w_n-t)\]
\[=\hf_n\Big(\frac{\w'}{\w_n}\big(\w_n+O(\e)-t +O(\e)+O(\frac{\e^2}{\w_n})+O(\frac{\e t}{\w_n})\big)\Big) d(\w-\w_n-t)\]
and by the fact that $\frac{\w'}{\w_n}<B$ and $\e\ll\w_0$,
\[=\hf_n\Big(\w'(1-\frac{t}{\w_n})+O(\e)+O(\frac{\e }{\w_n})\big)\Big) d(\w-\w_n-t)\]
so
\[ \exp(it\bT_1)[f_n(\w') d(\w-\w_n)] \approx \hf_n\big(\w'(1-\frac{t}{\w_n})\big) d(\w-\w_n-t).\]
In conclusion, up to an approximation, $\exp(it\bT_1)\{\hf_n\}_{n=0}^N\in \tilde{\cW}\otimes\tilde{\cS}$.

We can give two justifications to the assumption that the support of $F$ is bounded from above in the slope variable. First, note that if $\hs$ is supported in $[\w_0,\infty)$, and $\hf$ is supported in $(0,\w'_K)$, then $F(z\w,\w)=\hf\otimes\hs(z\w,\w)$ is zero for any slope $z>B=\frac{\w'_K}{\w_0}$. In practice, a small band of low frequencies of $\hs$ can be filtered out, and saved separately, prior to the implementation of the method. Alternatively, the support of $\hf\otimes\hs$ on high slopes can be filtered out using a $scale$-pass filter, and saved separately as a ``father wavelet'' component. %This can be seen by formula (\ref{eq:wav_slope}) of the wavelet-Plancherel transform.

Last, the addition of two spline sequences is also a spline sequence, with at most twice the number of knots. This shows that the model is invariant under the search operations. Note as well that the computational complexity of each of the above operations is linear in the overall number of knots.

\textcolor{black}{
In Appendix \ref{Implementation of the Wavelet-Plancherel search algorithm} we give implementation details of the wavelet-Plancherel search algorithm, and summarize the method in Algorithm \ref{Wavelet-Plancherel Search Algorithm}.
}

\subsection{Complexity of the WP4 algorithm}
\label{Complexity of the WP4 algorithm}

Assume that we perform our search on $\hf\otimes\hs$, where $\hf$ consists of $K$ knots. The spline sequence $\hf\otimes\hs$ consists of overall $NK$ knots. In general, the linear combination of $\hf\otimes\hs$ with $\exp(i t \bT_1)\hf\otimes\hs$ is a spline sequence with overall $2NK$ knots.
At this stage, it may seem like the data size of our model can increase exponentially as the search algorithm goes deeper, since we keep on summing spline sequences. However, we next show that the overall number of values of our model is bounded by $NKL$, at any depth of the search.

\begin{proposition}
\label{partition_prop}
Let $\hs$ consist of $N$ equidistant samples $0<\x_0,\ldots,\x_N$, and $\hf$ consist of $K$ knots $0<\x'_1,\ldots,\x'_K$.
Consider the coefficient search algorithm of WP4, with $time$-pass based on $L$ nonzero coefficients. 
Consider the $scale$ binary search, such that for slope support $[a,b]$, bisects the band in
\begin{equation}
c= \frac{2}{\frac{1}{a}+\frac{1}{b}}.
\label{eq:77ghjghjghj}
\end{equation}
Suppose that the initial slopes satisfy  $a\geq\frac{\w'_K}{\w_N}$ and $b\leq \frac{\w'_1}{\w_0}$.
Then the overall number of knots of the spline sequence model is approximately bounded by $NKL$ at every search step. Namely, for every $\d>0$ there is $N_{\d}>0$ such that for every $N>N_{\d}$,  the overall number of knots of the spline sequence model is bounded by $(1+\d)NKL$.
\end{proposition}

\begin{proof}
Each application of $\exp(it\bT_1)$ on a spline sequence creates new knots, by translating the original knots along slopes. Note that the exponentials $\exp(it\bT_1)$ in the Fourier expansion (\ref{eq:T_filter2}) translate in distances along $\cS$ that are multiples of the grid spacing of $\cS$.   For a knot $(\x_n,\x'_k)$, the set of all translated versions of this knot along a slope intersecting it, is a regular grid on the slope. There are $L$ new knots due to the first $time$-pass filter. More generally, there are $2^jL$ new knots due to all of the $j$ first $time$-pass filters.
Now, the partition (\ref{eq:77ghjghjghj}) divides evenly the knots of $\hf\otimes\hs$ between the two bands of slopes $[a,c]$ and $[c,b]$. To see this, denote $\l_k=\{(\w'_k,\w_n)\ |\ n=0,\ldots,N\}$.
For each $k=1,\ldots,K$, the number of knots of $\l_k$ in the slope band $[a,c]$ is estimated by
\[\frac{N}{\w_N-\w_0}\Big(\frac{\w'_k}{a}-\frac{\w'_k}{c}\Big),\]
with error at most $1$. Similarly, the number of knots of $\l_k$ in the slope band $[c,b]$ is estimated by
\[\frac{N}{\w_N-\w_0}\Big(\frac{\w'_k}{c}-\frac{\w'_k}{b}\Big).\]
Thus the overall number of knots of $\{(\w'_k,\w_n)\ |\ k=1,\ldots,K \ ,\ n=0,\ldots,N\}$, intersecting $[a,c]$ and $[c,b]$ respectively is
\[\frac{N}{\w_N-\w_0}\sum_{k=1}^K \Big(\frac{\w'_k}{a}-\frac{\w'_k}{c}\Big) \quad , \quad \frac{N}{\w_N-\w_0}\sum_{k=1}^K \Big(\frac{\w'_k}{c}-\frac{\w'_k}{b}\Big),\]
with error at most $K$. To guarantee equality between these estimates, we choose $c= \frac{2}{\frac{1}{a}+\frac{1}{b}}$. Each of the resulting slope bands $[a,c]$ and $[c,b]$ have at most $\frac{1}{2}NK+K=(\frac{1}{2}+\frac{1}{N})NK$ knots.

Thus, the partition (\ref{eq:77ghjghjghj})  also divides evenly the set of knots due to all of the $j$ first $time$-pass filters, counting repeated knots. This shows that at step $j$ of the search, the model comprises at most $(\frac{1}{2}+\frac{1}{N})^{j}2^jLKM$ knots. Observe that 
\[(\frac{1}{2}+\frac{1}{N})^{j}2^j = (1+\frac{2}{N})^j \]
and for $j=\log(N)$, the expression $(1+\frac{2}{N})^{\log(N)}$ converges to $1$ as $N\rightarrow\infty$.
\end{proof}

The initialization of the slopes in Proposition \ref{partition_prop} is satisfactory if the window $\hf$ has a small support near $\w_0$. The signal $\hs$ can be padded by zeros in such a situation, to guarantee that the initial slope interval covers all of the data. If we want to use the initial slopes $a=\frac{\w'_1}{\w_N}$ and $b=\frac{\w'_K}{\w_0}$, we need to demand that the $c$ slope partitions the knots $\{(\w'_k,\w_n)\ |\ k=1,\ldots,K \ ,\ n=0,\ldots,N\}$ intersecting the slope band $[a,b]$ into two equal sized subsets.

Since the model comprises at most of $NKL$ values, $time$ exponentials and a $scale$-pass filters take $O(NKL)$ operations \textcolor{black}{(see Appendix \ref{Implementation of the Wavelet-Plancherel search algorithm} for the implementation details)}. Moreover, one $time$-pass consists of $L$ $time$ exponentials, so we have the following result.

\begin{proposition}
The computational complexity of one search is $O(N\log(N)KL^2)$.
\end{proposition}

We remark that in practice we can choose for example $K=3$ and $L=9$, and both are $O(1)$ with respect to $N$. Therefore, our algorithm takes $O(N\log(N))$ operations for each coefficient search. This is comparable to a standard Matching Pursuit on the wavelet coefficients of a discrete wavelet frame \cite{Painless}, as described in Subsection \ref{Standard discretization of the 1D wavelet transform}.

%\begin{remark}
%Let us present a point of view that explains the reduction in computational complexity of the WP4 in comparison to the naive method. Both spaces $\cW\otimes\cS$ and $L^2(G)$ are function spaces over a 2D domain.  The window and signal data are entangled in $L^2(G)$, and there is no reasonable discretization of $G$ that separates the signal data from the window data. On the other hand, the window and signal data in $\cW\otimes\cS$ are in a sense separated, and can be discretized separabely. This allows us to encode the window direction with only a few parameters, since windows do not contain information, and spend most of the resources encoding the signal direction, which contains all of the information.
%\end{remark}

\subsection{1D continuous wavelet as a time-frequency transform}
\label{1D continuous wavelet as a time-frequency transform}

The WP4 method is good at pinpointing time-scale components of a signal accurately. We show in the next subsection the WP4 squares the phase space resolution in comparison to
the phase space resolution of discrete methods, with the same computational complexity. To understand the underlying resolution of the WP4 method, we explain in this subsection how to interpret the 1D continuous wavelet transform as a time-frequency transform. The time-frequency interpretation of the continuous wavelet transform is standard, see e.g. \cite{Practical}. For instructive purposes we offer a simplified analysis, skipping technical details.

Consider the space of signals $\hs$ with supports in 
 $[0,\infty)$.
Assume that the window $\hf$ is compactly supported, concentrated about the frequency $\w_0$. Let $(\w_0-\Delta,\w_0+\Delta)$ be the support of $\hf$, and
suppose that $0<\Delta<\w_0$. 

Next, we consider the standard formulation of the 1D continuous wavelet transform, with non-exponential scale $\b$ and time $\a$ (see Subsection \ref{The 1D wavelet transform} and \cite{Ten_lectures}).
The standard 1D continuous wavelet transform is written in frequency as
\[W_{f}[s](\a,\b) = \int_0^{\infty} \hs(\w)\overline{\sqrt{\b}e^{2\pi i \w \a}\hf\Big(\b\w\Big)}d\w\] 
with reconstruction
\begin{equation}
\hs(\w)= \int_{-\infty}^\infty\int_{0}^\infty W_f[s](\a,\b) \sqrt{\b}e^{2\pi i \w \a}\hf\Big(\b\w\Big)  d\a \frac{1}{\b^2}d\b.
\label{eq:reco_standard}
\end{equation}
Note that the dilated window $\sqrt{\b}\hf\Big(\b\w\Big)$ is centered at the frequency $\frac{\w_0}{\b}$, with support 
$(\frac{\w_0}{\b}-\frac{\Delta}{\b},\frac{\w_0}{\b}+\frac{\Delta}{\b})$. Thus, the scale $\b$ corresponds to the frequency $\k=\frac{\w_0}{\b}$.
Therefore, the change of variable $\k=\frac{\w_0}{\b}$ in the reconstruction formula (\ref{eq:reco_standard}), transforms phase space into a time-frequency space, and we have
\begin{equation}
\hs(\w)= \int_{-\infty}^\infty\int_{0}^\infty W_{f}[s](\a,\frac{\w_0}{\k}) \sqrt{\frac{\w_0}{\k}}e^{2\pi i \w \a}\hf\Big(\frac{\w_0}{\k}\w\Big)  d\a d\k,
    \label{eq:TF_Haar}
\end{equation}
where $\a$ is the time variable in phase space, and $\k$ is the frequency variable. 

The phase space kernels $W_{\hf}[\hf](\a,\frac{\w_0}{\k})$ (the ambiguity functions)
become elongated along the frequency direction when the frequency $\k$ gets larger, and become elongated along the time direction when the frequency $\k$ gets smaller (see Proposition \ref{Amb_is_kernel} and the discussion around it for ambiguity functions). From this point of view, we can think of the wavelet transform as a version of the STFT, with support size of the window proportional to the frequency. In other words, each transformed version of the window has the same number of oscillations, no matter  what the frequency value is.

This formulation is beneficial when analyzing signals that are best understood in terms of time-frequency, e.g. audio signal, using the wavelet transform. We use this point of view in Subsection \ref{Sampling resolution in phase space} to design a wavelet-based phase vocoder.

\subsection{Sampling resolution of WP4 in the time-frequency plane}
\label{Sampling resolution in phase space}

In view of (\ref{eq:77ghjghjghj})  and the change of variable from scale to frequency, $\k=\frac{\w_0}{\b}$, it is evident that the WP4 method partitions uniformly in frequency.
This shows that the underlying resolution of the method is a regular $N\times N$ grid in the time-frequency plane. In this subsection we compare WP4 to standard discrete wavelet methods, and investigate the merits of WP4 from a signal processing point of view. To make the ideas concrete, we take as a toy test case the phase vocoder method. 

\subsubsection{Resolution and complexity of standard discrete wavelet methods}
The discrete frame method with computational complexity equivalent to the WP4 method is based on the exponential frequency grid
$\{\k_m=\exp(hm)\w_0\}_{m=0}^{O(\log(N))}$ (see Subsection \ref{Standard discretization of the 1D wavelet transform}).
This discrete method consists of overall $O(N)$ samples in phase space. The samples in the discrete method are evenly distributed in the time-frequency plane, in the sense that any big rectangle of the same area contains approximately the same number of samples as $N\rightarrow\infty$, regardless of its position in the time-frequency plane. However, the resolution in the frequency direction decreases exponentially as the frequency increases. Indeed, the space between consecutive frequencies is proportional to the frequency. This exponential frequency grid can be constructed sequentially by $\k_{m+1}=\exp(h)\k_m$. 
If we want the resolution of this discrete method to be comparable to the WP4 resolution in the high frequencies, the spacing $h$ has to satisfy 
\[\k_{M-1}+r \approx \k_{M}=\exp(h)\k_{M-1} \] %\asymp
where $r$ is the grid spacing of the signal domain in frequency. Therefore
\[1+h\approx \exp(h)\approx\frac{\w_N+r}{\w_N}=1+\frac{r}{\w_N}\]
or $h=O(\frac{r}{\w_N})=O(\frac{1}{N})$. In view of (\ref{Disc_complex1}), such a choice of $h$ results in an $O(N^2\log(N))$ computational complexity, in comparison to the $O(N\log(N))$ complexity of WP4.
To conclude, the WP4 method improves the resolution in phase space of the discrete method, from $O(N)$ to $O(N^2)$, while keeping the complexity $O(N\log(N))$.

\subsubsection{Feature extraction using a densely discretized phase space}

Next, we explain the importance of the above observation in feature extraction. Recall that the spread of the kernels in phase space, along the frequency direction, is proportional to the frequency of the kernel (Subsection \ref{1D continuous wavelet as a time-frequency transform}). From this point of view, an exponential grid in the frequency direction is appropriate. Indeed, this discrete dictionary is a frame, and the discrete wavelet transform is invertible with stable inversion. However, if we treat the time-frequency plane as a feature space, and the wavelet coefficients as local frequency features, then the resolution of the discrete wavelet transform in this feature space is coarse, while the resolution in the WP4 method is fine. 

A fine resolution in the time-frequency feature space is important in some signal processing tasks. In the following we review in short such a task, namely time stretching phase vocoder \cite{vocoder_book}. The task in the phase vocoder method is to dilate the time of an audio signal, without changing its pitch, or more accurately without dilating its frequency content.
The standard phase vocoder is based on the STFT. First, the signal is transformed to the time-frequency plane via STFT. Then, the phase space representation of the signal is dilated along the time direction, while the phase of each coefficient is adjusted in a manner to be described later. Last, the resulting phase space function is synthesized back to the signal space. 

Since the wavelet transform can be interpreted as a time-frequency transform, in principle, the method should also work when replacing STFT with the wavelet transform. Some work have been done in this direction (see e.g. \cite{vocoder_ex,MyRefStoch}).
An important advantage that can be gained by using wavelets instead of STFT has to do with localization. In the STFT, the support $\Delta$ of the window is fixed. Therefore, there is a fixed time distance that govern the blurriness in the time direction. Indeed, every feature in the signal with a local frequency at time $t_0$, interacts with any window with the same frequency, in the time interval $(t_0-\Delta,t_0+\Delta)$. This phenomenon results in the smearing of signal features with sharp attacks (e.g., drums), and in periodic artifacts when dilating feature of sharp attack separated in time by less than $\Delta$. On the other hand, the support of a wavelet atom is inversely proportional to its frequency. This means that features with sharp attacks, which have slowly decaying frequency contents, are smeared less when using wavelets.
Note that to avoid wavelet atoms with overly large time supports, the low frequency content of the signal can be stretched using a standard STFT phase vocoder.

The signal model underlying the phase vocoder method is a sum of slowly varying pure waves. An audio signal is modeled as
\begin{equation}
\sum_{m=0}^M A_m(t) \exp(i\w_m(t))
\label{vocoder_model}
\end{equation}
where the instantaneous frequency of the $m$th component, $\w'_m(t)$, and the amplitude $A_m(t)$ are slowly varying. The phase vocoder method is justified for this model, if the frequency resolution is fine enough to approximate the instantaneous frequencies  \cite{vocoder_imp}. Note that many real-life audio signals are polyphonic, with the components of (\ref{vocoder_model}) well spread in the time-frequency plane. For this reason, a standard discrete wavelet transform produces low quality results for polyphonic audio signals, as it has a coarse time-frequency resolution \cite{MyRefStoch}. On the other hand, since the resolution of the WP4 method is fine in time-frequency, it is more suitable for time stretching polyphonic audio signals.

In Figure \ref{fig:vocoder} we compare a standard discrete wavelet method to WP4 in a toy example.
In the standard method, a sparse approximation to the signal $s$ is extracted using OMP on the discrete wavelet dictionary, while in the second method the sparse approximation is obtained using WP4.
In both methods, the signal is approximated by
\[\sum_k \exp(i\theta_k)c_k \ \pi(\a_k,\b_k)f_2,\]
where the scalar coefficients consist of $\theta_k\in[0,2\pi]$ and $c_k\in\RR_+$, and the window is $f_2$.
We then stretch time by an integer factor $T$, to get the output signal
\begin{equation}
\sum_k \exp(iT\theta_k)c_k \pi(T\a_k,\b_k)\ f_3,
\label{eq:Wvocoder}
\end{equation}
where $f_3$ is a different window. 
For the reason behind the $\exp(i\theta_k)c_k\mapsto \exp(iT\theta_k)c_k$ transformation, see \cite{vocoder_book}.
In both cases we follow the next procedure.
We consider three windows $\hf_1,\hf_2,\hf_3$, all centered about $\w_0$, with frequency supports $\Delta,\frac{1}{2}\Delta$ and $\frac{1}{2T}\Delta$ respectively.
We extract the coefficients $(\a_k,\b_k)$ using the pursuit method on the window $f_1$. To increase the overlap between the atoms, we then consider the atoms 
$\pi(\a_k,\b_k)f_2$, and calculate the coefficients by
\[\exp(i\theta_k)c_k = \ip{s}{\pi(\a_k,\b_k)f_2}.\]
Last, to reconstruct the stretched signal, we stretch the window to $f_3$, while keeping its frequency the same as $f_1,f_2$, and synthesize using (\ref{eq:Wvocoder}).

\begin{figure}[!ht]
\centering 
\includegraphics[width=1\linewidth]{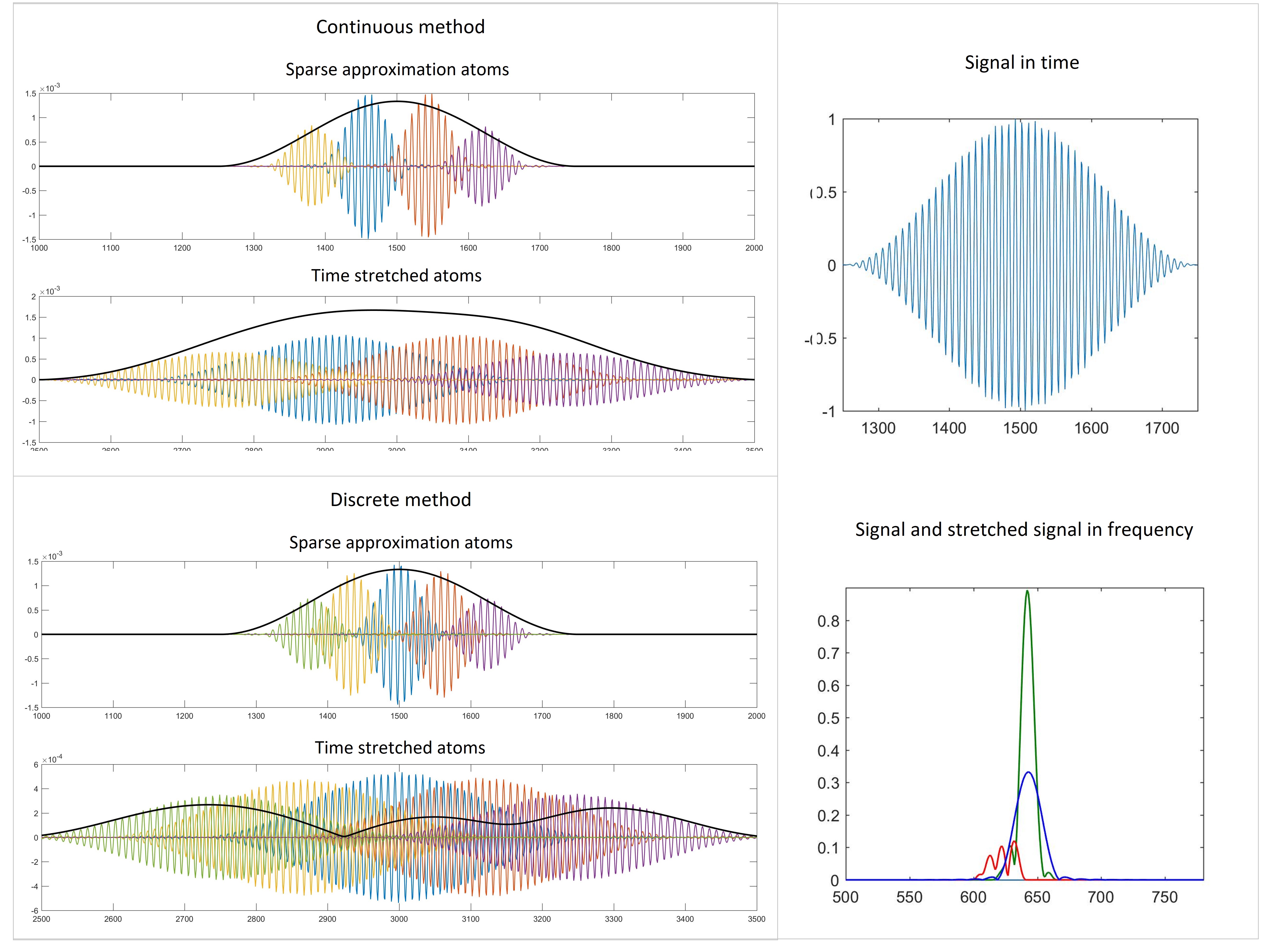}
\caption{ A signal $s$ with a constant rate of oscillation, localized in time, the sparse wavelet phase vocoder time stretching of the signal using the discrete frame method $s_D$, and the WP4 method $s_W$.
Top right: the signal $s$ in time.
 Bottom right: in blue $\abs{\hs}$, in green $\abs{\hs_W}$, in red $\abs{\hs_D}$. The WP4 preserves the instantaneous frequency of the signal, while dilating it.
Left: Top two: the envelope $\abs{s}$ in black (first graph), and $\abs{s_W}$ in black (second graph), along with the 4 coefficients of the sparse WP4 method in color (original coefficients in the first graph, and stretched coefficients in the second).
Bottom two: the same as the top two, but by the discrete frame method.}
\label{fig:vocoder} 
\end{figure}

\section*{Acknowledgment}

This research was supported by Grant No 1154/10 from the Israel Science Foundation (ISF).

\bibliographystyle{plain}	    
\bibliography{Ref_uncertainty3}

\newpage

\appendix

\section{Functional calculus}
\label{Functional calculus}

\textcolor{black}{We consider the ``Riemann-Stieltjes'' formulation of the spectral theorem (see e.g.\cite{Spectral_Hilbert}).}

\begin{definition}
\label{defi_PVM}
Let $\cH$ be a separable (complex) Hilbert space, and $\SS$ be $\RR$ or $e^{i\RR}$. Let ${\cal B}$ be the standard Borel $\sigma$-algebra of $\SS$, and let ${\cal P}$ be the set of orthogonal projections in $\cH$. A mapping $P:{\cal B}\rightarrow {\cal P}$ is called a \emph{projection valued Borel measure (PVM)} if
\begin{enumerate}
	\item $P(\SS)=I$ and $P(\emptyset)=0$.
	\item
	If $\{B_n\}_{n\in\NN}$ is a sequence of pairwise disjoint Borel sets, then for every $k\neq j$, $P(B_k)$ and $P(B_j)$ are projections to two orthogonal subspaces, and
	\[P(\bigcup_{n\in\NN}{B_n})=\sum_{n\in\NN}P(B_n).\] 
\end{enumerate}
\end{definition}

We present a ``Riemann-Stieltjes'' formulation of the spectral theorem (see e.g.\cite{Spectral_Hilbert}).

\begin{theorem}
\label{spectral theorem}
Let $T$ be a self-adjoint or unitary operator in the separable Hilbert space $\cH$. Let $\SS=\RR$ in case $T$ is self-adjoint, and $\SS=e^{i\RR}$ in case $T$ is unitary.
Then,
there is a PVM, $P:\SS\rightarrow {\cal P}$
such that
\begin{enumerate}
	\item 
	\begin{equation}
T=\int_{\SS}\lambda \ dP(\lambda)
\label{eq:spectral_theorem}
\end{equation}
where the integral in (\ref{eq:spectral_theorem}) is defined as follows.
Let ${\bf x}=\{x_0,\ldots,x_n\}$ denote a finite Riemann partition of $\SS$. Let ${\rm diam}({\bf x})$ denote the maximal diameter of intervals in ${\bf x}$. Denote by $[x_k,x_{k+1}]$ a general interval in the partition, and by $[x_0,x_n]$ the union of the intervals of ${\bf x}$. We have
\[T= \lim_{ \scriptsize
\begin{array}{c}
	{\rm diam}({\bf x})\rightarrow 0 \cr
	[x_0,x_n]\rightarrow \SS
\end{array}
} \sum_{k=0}^{n-1} x_k P\big([x_k,x_{k+1})\big)  \]
where the limit is in the strong topology in case $T$ is unbounded (and thus self-adjoint), and in the operator norm topology otherwise.
	\item
	For $\SS=\RR$,
	$f$ is in the domain of $T$ if and only if
	\begin{equation}
	\int_{\RR}\lambda^2 \ d\norm{P\big((-\infty,\lambda]\big)f}^2<\infty
	\label{eq:spectral_domain1}
	\end{equation}
	where the integral in (\ref{eq:spectral_domain1}) is the Riemann-Stieltjes integral \textcolor{black}{\cite[Chapter XII]{RS_int}} with respect to the weight function $\norm{P\big((-\infty,\ \cdot\ ]\big)f}^2:\RR\rightarrow\RR$.
\end{enumerate}

\end{theorem}

\begin{remark}(Functional calculus)
\label{Band_limit_poly}
A measurable function $\phi:\SS\rightarrow\CC$ of a self-adjoint or unitary operator $T$ is defined to be the normal operator
\begin{equation}
\phi(T)=\int_{\SS}\phi(\lambda) dP(\lambda).
\label{eq:functional_calculus1}
\end{equation}
defined on the domain of vectors $f\in\cH$ satisfying 
\[\int_{\RR}\abs{\phi(\lambda)}^2 \ d\norm{P\big((-\infty,\lambda]\big)f}^2<\infty.\]
This definition is consistent with polynomials of $T$ in the following sense. If $f\in \cH$ is band-limited, namely there exists some compact subset $B\subset\SS$ such that $f=P(B)f$, and if  $\{q_n\}_{n\in\NN}$ is a sequence of polynomials satisfying 
\[\lim_{n\rightarrow\infty}\norm{p_n-\phi}_{L^{\infty}(B)}=0\]
then
\[\lim_{n\rightarrow\infty}\norm{p_n(T)f - \phi(T)f}_{\cH}=0\]
where $p_n(T)$ is in the sense of compositions, additions, and multiplication by scalars of $T$, and $\phi(T)$ is in the sense of (\ref{eq:functional_calculus1}).
\end{remark}

\section{ Proofs}

\subsection{Proof of Proposition \ref{T_m_multiplicative_in_freq}}

\begin{lemma}
\label{lemma_Gamma_trans}
Consider a DGWT, let $\Gamma_m$ be the $quantity_m$ diffeomorphism for $m=2,\ldots,M$, and $L_m$ be the left translation in $L^2(Y_m\times N_m)$, along $N_m$. Then
\[[\Gamma_m \hf\big((\cdot)\bD_m(g_m')\big)](y_m,g_m) =\Gamma_m \hf (y_m,g_m^{\prime-1}g_m)= [L_m(g_m')\Gamma_m \hf](y_m,g_m)\]
\end{lemma}

\begin{proof}
\[\begin{split}
[\Gamma_m \hf\big( (\cdot) \bD_m(g_m')\big)](y_m,g_m)  = &  \hf\big(\bw_m\bD^m(y_m)\bD_m(g_m^{-1})\bD_m(g_m')\big)\\
 = & \hf\big(\bw_m\bD^m(y_m)\bD_m( g_m^{-1}\bullet g_m')\big) \\
 = &  \hf\big(\bw_m\bD^m(y_m)\bD_m\big( (g_m'^{-1}\bullet g_m)^{-1}\big)\big) \\
= &   [L_m(g_m^{\prime})\Gamma_m\hf](y_m,g_m). \\
\end{split}\]
\end{proof}

\begin{lemma}
\label{lemma_Gamma_trans_norm}
For every $(y_m,g_m)\in Y_m\times N_m$, and $g_m'\in N_m$, we have
\begin{equation}
\abs{{\rm det}J(y_m,g_m^{\prime -1}g_m)}=\abs{{\rm det} \bD(g_m')}\abs{{\rm det}J(y_m,g_m)}.
\label{eq:gyjghj66sjk}
\end{equation}
\end{lemma}

\begin{proof}
The application of $L_m(g_m')$ on $\Psi_m \hf$ gives the function
\[L_m(g_m')\Psi_m \hf:(y_m,g_m)\mapsto \sqrt{\abs{{\rm det}J(y_m,g_m^{\prime-1}g_m)}}[\Gamma_m \hf](y_m,g_m^{\prime-1}g_m).\]
By the fact that $L_m(g_m')$ is a unitary representation, and by the fact that $\Psi_m$ is an isometric isomorphism, we have
\begin{equation}
 \norm{
\begin{array}{l}
L_m(g_m')\Psi_m \hf :\\
(y_m,g_m)\mapsto \sqrt{\abs{{\rm det}J(y_m,g_m^{\prime-1}g_m)}}[\Gamma_m \hf](y_m,g_m^{\prime-1}g_m)
\end{array}
}_{L^2(Y_m\times N_m)}  
 =\norm{\hf}_{L^2_U(\RR^N)}.
\end{equation}
On the other hand, by (\ref{eq:geo_wav_freq}), and by Lemma \ref{lemma_Gamma_trans}, the application of $\Psi_m$ on $\pi_m(g_m')\hf$ gives the function
\[\Psi_m\pi_m(g_m')\hf:(y_m,g_m)\mapsto \sqrt{\abs{{\rm det} \bD(g_m')}\abs{J(y_m,g_m)}}[\Gamma_m \hf](y_m,g_m^{\prime-1}g_m).\]
Again, since $\pi_m(g_m')$ is unitary and $\Psi_m$ is an isometric isomorphism, we have
\[\begin{split}
\norm{
\begin{array}{l}
\Psi_m\pi_m(g_m')\hf:\\ (y_m,g_m)  \mapsto \sqrt{\abs{{\rm det} \bD(g_m')}\abs{{\rm det}J(y_m,g_m)}}[\Gamma_m \hf](y_m,g_m^{\prime-1}g_m)
\end{array}
}_{L^2(Y_m\times N_m)}& \\
=\norm{\hf}_{L^2_U(\RR^N)}& .
\end{split}\]
Thus, for every $\hf\in L^2(U)$,
\begin{equation}
\begin{split}
 &\norm{L_m(g_m')\Psi_m \hf:(y_m,g_m)\mapsto \sqrt{\abs{{\rm det}J(y_m,g_m^{\prime-1}g_m)}}[\Gamma_m \hf](y_m,g_m^{\prime-1}g_m)}_{L^2(Y_m\times N_m)}  \\
&= \norm{
\begin{array}{l}
\Psi_m\pi_m(g_m')\hf:\\
(y_m,g_m)\mapsto \sqrt{\abs{{\rm det} \bD(g_m')}\abs{{\rm det}J(y_m,g_m)}}[\Gamma_m \hf](y_m,g_m^{\prime-1}g_m)
\end{array}
}_{L^2(Y_m\times N_m)}.  
\end{split}
\label{eq:popogp8fdh}
\end{equation}

We now show that (\ref{eq:gyjghj66sjk}) follows from (\ref{eq:popogp8fdh}).
For a fixed $g_m'$,  consider functions of the form
\begin{equation}
h(y_m,g_m)=\abs{[\Gamma_m \hf](y_m,g_m^{\prime-1}g_m)}^2,
\label{eq:f_2_F_via_Gamma}
\end{equation}  Note that (\ref{eq:f_2_F_via_Gamma}) appears in both sides of (\ref{eq:popogp8fdh}). 
By the fact that $C_m$ is a diffeomorphism, when applied on sets it is a bijective mapping between the family of compact domains in $Y_m\times N_m$ and the family of compact domains in $U$. We can thus exhaust the space of characteristic functions of compact domains with functions of the form (\ref{eq:f_2_F_via_Gamma}). In this sense, functions of the form (\ref{eq:f_2_F_via_Gamma}) are generic.

Generally, the value of any continuous function $Q:Y_m\times N_m\rightarrow\RR$ at some point $(y_m,g_m)$ can be calculated as the limit of the integrals of $Q$ against a sequence of characteristic functions of concentric compact rectangles, centered at $(y_m,g_m)$, with intersection of all of the rectangles equal to $\{(y_m,g_m)\}$.
By this, and by the fact that $\abs{{\rm det}J(y_m,g_m^{\prime-1}g_m)}$ and $\abs{{\rm det} {\bf D}(h_1)}\abs{{\rm det}J(y_m,g_m)}$ are continuous functions,
 we must have
\[\forall (y_m,g_m)\in Y_m\times N_m\ , \quad\abs{{\rm det}J(y_m,g_m^{\prime-1}g_m)}=\abs{{\rm det} {\bf D}(g_m')}\abs{{\rm det}J(y_m,g_m)}.\]
\end{proof}

\begin{proof}[Proof of Proposition \ref{T_m_multiplicative_in_freq}]

First, it is trivial to see that the inverse Fourier transform is a $position$ transform with
\[[\bbQ_1 F](\mathbf{x}) = \big(x_1 F(\mathbf{x}),\ldots,x_NF(\mathbf{x})\big)\]
for every $F\in L^2(\RR^N)$, with the notation $\mathbf{x}=(x_1,\ldots,x_N)$.

In the following, we treat the signal space as the frequency domain $L^2(U)$, and treat $\pi$ as the representation (\ref{eq:geo_wav_freq}).  
Let $2 \leq m \leq M$.
First we show property (\ref{eq:quant_trans1}) of $quantity_m$ transforms.
By Lemmas \ref{lemma_Gamma_trans} and \ref{lemma_Gamma_trans_norm},
$\Psi_m$ transforms $\pi_m$ to $L_m$. Indeed
\[\begin{split}
[\Psi_m \pi_m(g_m')\hf](y_m,g_m) =& \abs{{\rm det} \bD(g_m')}\abs{{\rm det} J(y_m,g_m)}[\Gamma_m \hf\big( (\cdot) \bD_m(g_m')\big)](y_m,g_m) \\
= &\abs{{\rm det} J(y_m,g_m^{\prime-1}g_m)}[L_m(g_m^{\prime})\Gamma_m \hf](y_m,g_m)\\
= & [L_m(g_m')\Psi_m \hf](y_m,g_m).
\end{split}\]

Next, we show property (\ref{eq:quant_trans2}) of $quantity_m$ transforms. 
Let $\bbQ_m$ be the multiplicative operator in $L^2(Y_m\times N_m)$, mapping to to $L^2(Y_m\times N_m)^{K_m}$, defined by
\[[\bbQ_m F](y_m,g_m) = \bg_m F(y_m,g_m).\]
First, for $g'=g'_1\in N_1$, $\pi(g'_1)$ is modulations in $L^2(U)$. Since $\Psi_m$ is based on a change of variable, we conclude that $\rho^m(g'_m)$ is a multiplicative operator in $L^2(Y_m\times N_m)$. It thus commutes with ${\bf\bQ}_m$, which is also a  multiplicative operator, and (\ref{eq:quant_trans2}) is guaranteed for $\rho^m$ restricted to $N_1$.

Consider the restriction of $\pi$ to $H_1$, $\pi_1(h_1')$, for $h_1'\in H_1$.
Let
\[\rho^m(h_1') = \Psi_m\pi_1(h_1')\Psi_m^*.\]
We now show
\begin{equation}
\rho^m(h_1')^*\bbQ_m\rho^m(h_1') = \bg_m\bullet A_m(h_m') \bbQ_m,
\label{eq:asasasl9y}
\end{equation}
which shows property (\ref{eq:quant_trans2}) of Definition \ref{def:quantity_trans}, for any $m=2,\ldots,M$. As a result, (\ref{eq:quant_trans2}) is satisfied for any $g'\in G$ since for $g'=g'_1 h'_1$
\[\begin{split}
\rho^m(g')^*\bbQ_m\rho^m(g') = &\rho^m(h'_1)^*\rho^m(g'_1)^*\bbQ_m \rho^m(g'_1) \rho^m(h'_1)\\
 = & \rho^m(h'_1)^*\bbQ_m  \rho^m(h'_1) = \bg_m\bullet A_m(h_m') \bbQ_m.
\end{split}\]

Equation (\ref{eq:asasasl9y}) is equivalent to
\begin{equation}
\pi_1(h_1')^*\Psi_m^*\bbQ_m\Psi_m\pi_1(h_1') = \bg_m\bullet A_m(h_m')\Psi_m^* \bbQ_m\Psi_m.
\label{eq:to_show1t}
\end{equation}
In this proof we map $\pi_1(h_1')\hf$ to $L^2(Y_m\times N_m)$ using $\Psi_m$, then operate on the result multiplicatively using $\bbQ_m$, and map back using $\Psi_m^*$. The transform $\Psi_m$ is based on the change of variable $C_m:Y_m\times N_m\rightarrow U$, and $\Psi_m^*$ is based on the inverse change of variables $C_m^{-1}$. Since we apply $C_m^{-1}$ right after $C_m$, an explicit inversion formula is not needed. Instead, we formulate our proof using the less explicit ``$\mapsto$'' sign. Note that $C_m^{-1}:C_m(y_m,g_m)\mapsto (y_m,g_m)$ and $\Psi_m^*$ can be written implicitly by
\begin{equation}
\begin{split}
\Psi_m^* : & \bigg(F:(y_m,g_m)\mapsto F(y_m,g_m)\bigg) \\
& \mapsto \bigg(\Psi_m^* F: C_m(y_m,g_m) \mapsto \frac{1}{\sqrt{\abs{{\rm det}J(y_m,g_m)}}}F(y_m,g_m)\bigg).
\end{split}
\label{eq:4kjhgfy7}
\end{equation}
Consider the mapping $q_m:Y_m\times N_m \rightarrow N_m $ defined by
\[q_m(y_m,g_m) = \bg_m.\]
Note that $q_m$ is the function that $\bbQ_m$ multiplies by.
Pulling $q_m$ to a function defined on $U\subset\RR^N$ via $C_m^{-1}$, we get the mapping 
\[q_m\circ C_m^{-1}:C_m(y_m,g_m) \mapsto \bg_m,\] 
or equivalently
\begin{equation}
q_m\circ C_m^{-1}:\bw_m\bD^m(y_m)\bD_m(g_m^{-1}) \mapsto \bg_m.
\label{eq:gj88i8ie5xx}
\end{equation}
Thus, by (\ref{eq:quantity_trans}), (\ref{eq:4kjhgfy7}), and (\ref{eq:gj88i8ie5xx}), $\Psi_m^* \bbQ_m\Psi_m$ is implicitly given by
\[[\Psi_m^* \bbQ_m\Psi_m\hf]\big(\bw_m\bD^m(y_m)\bD_m(g_m^{-1})\big)  =  \bg_m \hf\big(\bw_m\bD^m(y_m)\bD_m(g_m^{-1})\big).\]

To calculate the left-hand side of (\ref{eq:to_show1t}), we have
\[\begin{split}
& [\pi_1(h_1')^*\Psi_m^* \bbQ_m\Psi_m\pi_1(h_1')\hf]\big(\bw_m\bD^m(y_m)\bD_m(g_m^{-1})\big) \\
& =[\pi_1(h_1')^*\Psi_m^* \bbQ_m\Psi_m \sqrt{{\rm det}\bD(h_1')}\hf\big(\ \cdot\ \bD(h_1')\big)]\big(\bw_m\bD^m(y_m)\bD_m(g_m^{-1})\big),
\end{split}\]
which is the application of $\pi_1(h_1')^*$ on the function
\[\begin{split}
 & \Psi_m^* \bbQ_m\Psi_m \sqrt{{\rm det}\bD(h_1')}\hf\big(\ \cdot\ \bD(h_1')\big): \\  
 & \quad \quad\quad \quad\quad \quad\quad \quad \bw_m\bD^m(y_m)\bD_m(g_m^{-1})    \\
 & \quad \quad\quad \quad\quad \quad \quad \quad \longmapsto \bg_m \sqrt{{\rm det}\bD(h_1')}\hf\big(\bw_m\bD^m(y_m)\bD_m(g_m^{-1})\bD(h_1')\big). 
\end{split}\]
 Since in $\pi_1(h_1')^*$ we multiply by $\sqrt{{\rm det}\bD(h_1')}^{\ -1}$, we can write
\[\begin{split}
& [\pi^1(h_1')^*\Psi_m^* \bbQ_m\Psi_m\pi^1(h_1')\hf]\big(\bw_m\bD^m(y_m)\bD_m(g_m^{-1})\big) \\
& = R(h_1';\bw_m\bD^m(y_m)\bD_m(g_m^{-1}))\hf\big(\bw_m\bD^m(y_m)\bD_m(g_m^{-1})\big).
\end{split}\]
where $R(h_1';\bw_m\bD^m(y_m)\bD_m(g_m^{-1}))$ is the application of $(\cdot)\bD(h_1^{\prime-1})^{-1}$ on the input variable of
\[q_m\circ C^{-1}:\ \bw_m\bD^m(y_m)\bD_m(g_m^{-1}) \mapsto \bg_m,\]
namely, it is
\begin{equation}
R(h_1'; \ \cdot\ ):\ \bw_m\bD^m(y_m)\bD_m(g_m^{-1})\bD(h_1') \mapsto \bg_m.
\label{eq:fgdi9h1}
\end{equation}
Indeed, for a general function $E:U\rightarrow\CC^{K_m}$, the non-normalized dilation
 $E(\ \cdot\ \bD(h_1^{\prime-1}))$ %:\bw\mapsto {\bf R}(\bw\bD(h_1^{\prime-1}))$
can be written by
\[E(\ \cdot\ \bD(h_1^{\prime-1})):\bw \mapsto E\big(\bw \bD(h_1^{\prime-1})\big),\]
or by
\begin{equation}
E(\ \cdot\ \bD(h_1^{\prime-1})):\bw \bD(h_1^{\prime-1})^{-1}\mapsto  E(\bw).
\label{eq:4kjhgfy7mmm}
\end{equation}

Next, we show that $R(h_1'; \ \cdot\ )$ can be written by
\[R(h_1'; \ \cdot\ ):\ \bw_m\bD^m({y}_m)\bD_m\big( g_m^{-1}\big)\mapsto \bg_m'\bullet{\bf A}_m({\bf h}_m') \bg_m, \]
which completes the proof.

The mapping (\ref{eq:fgdi9h1}) is equivalent to
\[\begin{split}
 R(h_1'; \ \cdot\ ):\ &  \bw_m\bD^m(y_m)\bD\big( g_m^{-1} h_1'\big)\mapsto \bg_m  \\
 R(h_1'; \ \cdot\ ):\ &  \bw_m\bD^m(y_m)\bD\big( g_m^{-1} g_2'\ldots g_m' h_m'\big)\mapsto \bg_m  
\end{split}\]
By the commutation relation due to the semi-direct product structure, there are $g_2''\ldots  g_{m-1}''$ such that
\[\begin{split}
R(h_1'; \ \cdot\ ):\ &  \bw_m\bD^m(y_m)\bD\big( g_2''\ldots  g_{m-1}'' g_m^{-1} g_m' h_m'\big)\mapsto \bg_m \\
R(h_1'; \ \cdot\ ):\ &   \bw_m\bD^m(y_m)\bD\big( g_2''\ldots   g_{m-1}'' h_m' \ h_m^{\prime-1} g_m^{-1}   g_m' h_m'\big)\mapsto \bg_m \\
R(h_1'; \ \cdot\ ):\ &   \bw_m\bD^m(y_m)\bD\big( g_2''\ldots   g_{m-1}'' h_m' \ A_m(h_m^{\prime-1}) (g_m^{-1}   g_m') \big)\mapsto \bg_m\\
R(h_1'; \ \cdot\ ):\ &   \bw_m\bD(g_1\ldots g_{m-1}g_{m+1}\ldots g_M g_2''\ldots   g_{m-1}'' h_m')\bD\big(  A_m(h_m^{\prime-1}) (g_m^{-1}   g_m') \big)\mapsto \bg_m.
\end{split}\]
Again, by the commutation relation due to the semi-direct product structure we can sort the variables in $g_1\ldots g_{m-1}g_{m+1}\ldots g_M g_2''\ldots   g_{m-1}'' h_m'$ by their index and obtain some $\tilde{y}_m\in Y_m$ such that
\begin{equation}
R(h_1'; \ \cdot\ ):\ \bw_m\bD^m(\tilde{y}_m)\bD_m\big(  A_m(h_m^{\prime-1}) (g_m^{-1}   g_m') \big)\mapsto \bg_m. 
\label{eq:fdfd5faotg}
\end{equation}
Note that the variable $\tilde{y}_m$ is in one-to-one correspondence with $y_m$, so the implicit formula (\ref{eq:fdfd5faotg}) determines $R(h_1'; \ \cdot\ )$ for any point in its domain.
Let us change the variable
\[\begin{split}
 & A_m(h_m^{\prime-1}) (g_m^{-1}   g_m') = \eta_m^{-1}\\
 &  g_m^{-1}   g_m' = A_m(h_m') \eta_m^{-1} \\
 &  g_m^{-1}   =  g_m^{\prime-1}A_m(h_m') \eta_m^{-1}. 
\end{split}\]
By the fact that $A_m(h_m')$ is a homomorphism, and $N_m$ commutative, we have
\[ g_m   =  g_m'A_m(h_m') \eta_m.\]
Therefore, changing the notation $\eta_m\mapsto g_m$ and $\tilde{y}_m\mapsto y_m$ (by bijectivity of $\tilde{y}_m,y_m$), and plugging the change of variable in (\ref{eq:fdfd5faotg}), we get
\[R(h_1'; \ \cdot\ ):\ \bw_m\bD^m({y}_m)\bD_m\big( g_m^{-1}\big)\mapsto \bg_m'\bullet{\bf A}_m({\bf h}_m') \bg_m \sim g_m'A_m(h_m') g_m. \]

\end{proof}

\subsection{Proof of Proposition \ref{prop:SPWT_pull_obs}}

Before we present the proof of Proposition \ref{prop:SPWT_pull_obs}, we begin with theoretical preparation.

\subsubsection{Technical assumptions for the pull-back analysis}

To understand the pull-back  in Proposition \ref{prop:SPWT_pull_obs}, along with its assumptions, we need the following discussion. 
%First, recall some general definitions from functional analysis. 	Given a linear operator $T$, if the closure of the graph of $T$ is the graph of an operator, this operator is called the closure of $T$, and we say that $T$ is closable. In this case, we denote the closure of $T$ by $\overline{T}$. For any vector $v$ in the domain of $\overline{T}$, there exsists a sequence of vectors $\{v_j\}_{j\in\NN}$ in the domain of $T$, such that $\lim_{j\rightarrow\infty}v_j=v$ and $\lim_{j\rightarrow\infty}Tv_j=\overline{T}v$.  
Consider a simply dilated SPWT. Denote by $[\cW\otimes\cS]_0$ the space of (finite) linear combinations of simple tensors in $\cW\otimes\cS$, under a specific construction of the $\otimes$ operator. Denote 
$V(\cW\otimes\cS)_0=V([\cW\otimes\cS]_0)$. 
		
As will be shown in the proof of Proposition \ref{prop:SPWT_pull_obs}, the space $V(\cS\otimes\cW)$ %the multi-observable ${\bf\breve{N}}_M$ 
 is invariant under the multi-observable ${\bf\breve{N}}_M$, 
%$\breve{{\cal M}}^{\infty}(G)$ operators,
 and specifically the restriction ${\bf\breve{N}}_M\big|_{V(\cW\otimes\cS)}$ is a self-adjoint or unitary operator. 	
It is then possible to derive an explicit formula for ${\bf\bT}_{M-1}\big|_{[\cW\otimes\cS]_0}$ (see (\ref{eq:pull1_self_adjoint}) and (\ref{eq:pull1_unitary})). Moreover,
	we show in the proof the following inversion formula of (\ref{eq:multi_pull01}), restricted to finite linear combinations of simple tensors,
	\begin{equation}
	{\bf\breve{N}}_{M-1}\big|_{V(\cW\otimes\cS)_0} = V{\bf\bT}_{M-1}\big|_{[\cW\otimes\cS]_0}V^*.
	\label{eq:N_MV2}
	\end{equation}
	
	At this stage we need an assumption 
	in order to proceed. Equation (\ref{eq:N_MV2}) says in particular that the space $V(\cW\otimes\cS)_0\cap {\cal D}({\bf\breve{N}}_{M-1})$ maps to $V(\cW\otimes\cS)_0$ under ${\bf\breve{N}}_{M-1}$. Note that the fact that ${\bf\breve{N}}_{M-1}$ is self-adjoint, only implies that ${\bf\breve{N}}_{M-1}\big|_{V(\cW\otimes\cS)_0}$ is symmetric, and thus by
	\begin{equation}
	{\bf\bT}_{M-1}\big|_{[\cW\otimes\cS]_0}=V^*{\bf\breve{N}}_{M-1}\big|_{V(\cW\otimes\cS)_0}V,
	\end{equation}
	${\bf\bT}_{M-1}\big|_{[\cW\otimes\cS]_0}$ is symmetric. Recall that not every symmetric operator can be extended to a self-adjoint operator.
	The subtle assumption we take is that the closure of ${\bf\bT}_{M-1}\big|_{[\cW\otimes\cS]_0}$ is a self-adjoint or unitary operator ${\bf\bT}_{M-1}$ in $V(\cW\otimes\cS)$. By this, the closure of ${\bf\breve{N}}_{M-1}\big|_{V(\cW\otimes\cS)_0}$ is a self-adjoint or unitary operator ${\bf\breve{N}}_{M-1}'$ in $V(\cW\otimes\cS)$. We further show in the proof that ${\bf\breve{N}}_{M-1}'={\bf\breve{N}}_{M-1}\big|_{V(\cW\otimes\cS)}$. We hence calculate all of the pulled-back multi-observables ${\bf\bT}_{m}$ by induction on $m$, assuming at each stage the closability to a self-adjoint operator condition, and show that  ${\bf\breve{N}}_m\big|_{V(\cW\otimes\cS)_0}$ is extended to a self-adjoint or unitary operator ${\bf\breve{N}}'_m={\bf\breve{N}}_m|_{V(\cW\otimes\cS)}$ in $V(\cW\otimes\cS)$. 
	
	Note that closability of an isometric operator to a unitary operator is always guaranteed. Also note that every symmetric operator is closable, so the only condition to check is that the closure is self-adjoint.

	%\url{https://en.wikipedia.org/wiki/Unbounded_operator#Closed_linear_operators}
	%
	%Given a linear operator $A$, if the closure of the graph of $A$ is thre graph of an operator, this operator is called the closure of $A$, and we say that $A$ is closable. Denote the closure of $A$ by $\overline{A}$. The domain of $A$ is called a core, or essential domain, of $\overline{A}$.
	%
	%Property: An operator admits a closure if and only if for every pair of sequences $\{x_n\}$ and $\{y_n\}$ in the domain of $A$, $D(a)$, both converging to $x$, such that both $\{A x_n\}$ and $\{A y_n\}$ convergem one has $\lim_n A x_n = \lim_n A y_n$. Don't need this one.
	%
	%
	%\url{https://en.wikipedia.org/wiki/Extensions_of_symmetric_operators}
	%
	% it is a convenient fact that every symmetric operator A is closable. That is, A has a smallest closed extension, called the closure of A. This can be shown by invoking the symmetric assumption and Riesz representation theorem. Since A and its closure have the same closed extensions, it can always be assumed that the symmetric operator of interest is closed.

\subsubsection{Proof of Proposition \ref{prop:SPWT_pull_obs}}

\begin{proof}%[Proof of Proposition \ref{prop:SPWT_pull_obs}]
In this proof, for a tuple of windows ${\bf f}=(f_1,\ldots,f_K)$, we denote 
\[V_{{\bf f}}[s]=\big(V_{f_1}[s],\ldots,V_{f_K}[s]\big),\]
and similarly for a tuple of signals ${\bf s}=(s_1,\ldots,s_K)$. When both the window and the signal are tuples, we denote
\[V_{{\bf f}}[{\bf s}]=\big(V_{f_1}[s_1],\ldots,V_{f_K}[s_K]\big).\]

We prove by induction on $m$, starting with $m=M$, where by convention we denote ${\bf A}_M({\bf h}_M) = {\bf I}$. In the following derivation we prove the base of the induction together with the induction step.
For each $m$, we assume that  $\otimes=\otimes_m$ is based on a signal and window $quantity_m$ bases. 
All of the identities in this proof, developed in a different basis, are the canonical transform (\ref{eq:canonical_transform}) of the same identities developed in signal and window $quantity_m$ bases. Thus, if we show that ${\bf \breve{N}}_m\big|_{V(\cW\otimes\cS)}$ is self-adjoint or unitary in $V(\cW\otimes\cS)$, where $\otimes$ is based on signal and window $quantity_m$ bases, then ${\bf \breve{N}}_m\big|_{V(\cW\otimes\cS)}$ is self-adjoint or unitary in $V(\cW\otimes\cS)$ for any construction of $\otimes$.

For any $f,s\in {\Psi}_m^{\prime*}{\cal D}({\bf\bQ}'_m) \cap {\Psi}_m^*{\cal D}({\bf\bQ}_m)\subset \cW\cap\cS$ note that the application of ${\bf\bQ}'_m$ and ${\bf\bQ}_m$ are equivalent, and the application of $\Psi'_m$ and $\Psi_m$ are equivalent. Calculate
\[V_{\Psi_m^*{\bf\breve{Q}}_m \Psi_m f}[s](g)=\iint [\Psi_m s]({\bf x}_m,y) \overline{[\rho^m(g) {\bf\breve{Q}}_m \Psi_m f]({\bf x}_m,y)}  d{\bf x}_m dy\]
\[=\iint [\Psi_m s]({\bf x}_m,y) \overline{[\rho^m(g) {\bf\breve{Q}}_m \rho^m(g)^* \rho^m(g) \Psi_m f]({\bf x}_m,y)}  d{\bf x}_m dy\]
By the inversion formula (\ref{eq:inverse_coo1}), the ${\bf g}_m$ coordinate of ${\bf g}^{-1}$ is ${\bf A}_m({\bf h}_m^{-1}){\bf g}_m^{-1}={\bf A}_m({\bf h}_m)^{-1}{\bf{g}}_m^{-1}$, so by the canonical commutation relation
(\ref{eq:quant_trans2}) with $g'=g^{-1}$
\begin{equation}
\begin{split}
 & V_{\Psi_m^*{\bf\breve{Q}}_m \Psi_m f}[s](g) \\
 &= \iint [\Psi_m s]({\bf x}_m,y) \overline{[{\bf A}_m({\bf h}_m)^{-1}{\bf{g}}_m^{-1}\bullet {\bf A}({\bf h}_m)^{-1}{\bf\breve{Q}}_m \rho^m(g) \Psi_m f]({\bf x}_m,y)}  d{\bf x}_m dy\\
 &=\iint [\overline{{\bf A}_m({\bf h}_m)^{-1}{\bf{g}}_m^{-1}\bullet {\bf A}({\bf h}_m)^{-1}{\bf\breve{Q}}_m} \Psi_m s]({\bf g}_m,y) \overline{[ \rho^m(g) \Psi_m f]({\bf x}_m,y) } d{\bf x}_m dy
\end{split}
\label{eq:temp4hh08080}
\end{equation}
where the last equality is due to the fact that the application of ${\bf A}_m({\bf h}_m)^{-1}{\bf{g}}_m^{-1}\bullet {\bf A}({\bf h}_m)^{-1}{\bf\breve{Q}}_m$ on the functions $\Psi_m f$ or $\Psi_ms$ is a tuple of multiplications by numerical values.

Let us treat the two cases of the physical quantity $G_m^1$.
In case $G_m^1$ is $\RR$ or $\ZZ$, we have $\bullet = +$, and ${\bf A}_m({\bf h}_m)^{-1}{\bf{g}}_m^{-1}\bullet {\bf A}({\bf h}_m)^{-1}{\bf\breve{Q}}_m$ are multiplications by real numbers.
\[\begin{split}
V_{\Psi_m^*{\bf\breve{Q}}_m \Psi_m f}[s](g)= &\iint [{\bf A}({\bf h}_m)^{-1}{\bf\breve{Q}}_m \Psi_m s]({\bf x}_m,y) \overline{[ \rho^m(g) \Psi_m f]({\bf x}_m,y)}  d{\bf x}_m dy\\
& -\iint [{\bf A}({\bf h}_m)^{-1}{\bf g}_m\Psi_m s]({\bf x}_m,y) \overline{[ \rho^m(g) \Psi_m f]({\bf x}_m,y)}  d{\bf x}_m dy.
\end{split}\]
By the fact that the coordinates on ${\bf A}({\bf h}_m)^{-1}$ are multiplication by scalars, we can write
\[\begin{split}
V_{\Psi_m^*{\bf\breve{Q}}_m \Psi_m f}[s](g)= &{\bf A}({\bf h}_m)^{-1}\iint [{\bf\breve{Q}}_m \Psi_m s]({\bf x}_m,y) \overline{[ \rho^m(g) \Psi_m f]({\bf x}_m,y)}  d{\bf x}_m dy\\
& -{\bf A}({\bf h}_m)^{-1}{\bf g}_m\iint [\Psi_m s]({\bf x}_m,y) \overline{[ \rho^m(g) \Psi_m f]({\bf x}_m,y)}  d{\bf x}_m dy \\
= &{\bf A}({\bf h}_m)^{-1}V_{f}[\Psi_m ^*{\bf\breve{Q}}_m\Psi_m s](g) -{\bf A}({\bf h}_m)^{-1}{\bf \breve{N}}_mV_f[s](g).
\end{split}\]
so
\[{\bf \breve{N}}_mV_f[s]=  V_{f}[\Psi_m ^*{\bf\breve{Q}}_m\Psi_m s] -{\bf A}_m({\bf h}_m)V_{\Psi_m ^*{\bf\breve{Q}}_m\Psi_m f}[s].\]
By the definition of the wavelet-Plancherel transform $V$, and by functional calculus, we can write
\begin{equation}
\begin{split}
{\bf \breve{N}}_mV(f\otimes s) &={\bf \breve{N}}_mV_f[s] \\
&=  V_{f}[\Psi_m ^*{\bf\breve{Q}}_m\Psi_m s] -{\bf A}_m({\bf\breve{N}}_{m+1},\ldots,{\bf \breve{N}}_M)V_{\Psi_m ^*{\bf\breve{Q}}_m\Psi_m f}[s]\\
&=  V\Big(f\otimes(\Psi_m ^*{\bf\breve{Q}}_m\Psi_m s)\Big) -{\bf A}_m({\bf\breve{N}}_{m+1},\ldots,{\bf \breve{N}}_M)V\Big((\Psi_m ^*{\bf\breve{Q}}_m\Psi_m f)\otimes s \Big).
\end{split}
\label{eq:ihgfd4ud}
\end{equation}
Equation (\ref{eq:ihgfd4ud}) can be written in a properly defined subspace of $[\cW\otimes\cS]_0$ by
\begin{equation}
\begin{split}
{\bf \breve{N}}_m VF= & 
 \ \ V[\Psi'_m\otimes \Psi_m]^*[{\bf I}\otimes{\bf\breve{Q}}_m][\Psi'_m\otimes \Psi_m]F \\
 & -{\bf A}_m({\bf\breve{N}}_{m+1},\ldots,{\bf \breve{N}}_M) V[\Psi'_m\otimes \Psi_m]^*[{\bf\breve{Q}}'_m\otimes {\bf I}][\Psi'_m\otimes \Psi_m]F.
\end{split}
\label{eq:ihgfd4ud200}
\end{equation}
By the wavelet-Plancherel theorem we have $F=V^*VF=V^* H$ for a generic $H\in V(\cW\otimes\cS)_0$, and we can write
\begin{equation}
\begin{split}
{\bf \breve{N}}_m\big|_{V(\cW\otimes\cS)_0}= & \ \ 
 V[\Psi'_m\otimes \Psi_m]^*[{\bf I}\otimes{\bf\breve{Q}}_m][\Psi'_m\otimes \Psi_m]V^* \\
&  -{\bf A}_m({\bf\breve{N}}_{m+1},\ldots,{\bf \breve{N}}_M) V[\Psi'_m\otimes \Psi_m]^*[{\bf\breve{Q}}'_m\otimes {\bf I}][\Psi'_m\otimes \Psi_m]V^*.
\end{split}
\label{eq:before_indA1}
\end{equation}
In the base of the induction, $m=M$, and 
\begin{equation}
\begin{split}
{\bf \breve{N}}_M\big|_{V(\cW\otimes\cS)_0}= & \ \ 
 V\left([\Psi'_M\otimes \Psi_M]^*[{\bf I}\otimes{\bf\breve{Q}}_M][\Psi'_M\otimes \Psi_M]\right. \\
&  \left. - [\Psi'_M\otimes \Psi_M]^*[{\bf\breve{Q}}'_M\otimes {\bf I}][\Psi'_M\otimes \Psi_M]\right)V^*.
\end{split}
\label{eq:before_indA1_base}
\end{equation}
Now, 
\begin{equation}
\begin{split}
{\bf \breve{T}}_M\big|_{[\cW\otimes\cS]_0}= & \ \ 
 V^*{\bf \breve{N}}_M\big|_{V(\cW\otimes\cS)_0}V \\
=&  [\Psi'_M\otimes \Psi_M]^*[{\bf I}\otimes{\bf\breve{Q}}_M][\Psi'_M\otimes \Psi_M] \\
&
   - [\Psi'_M\otimes \Psi_M]^*[{\bf\breve{Q}}'_M\otimes {\bf I}][\Psi'_M\otimes \Psi_M].
\end{split}
\label{eq:before_indA1_baseT}
\end{equation}
 Moreover, since ${\bf \breve{N}}_M\big|_{V(\cW\otimes\cS)_0}$ is symmetric, so is ${\bf \breve{T}}_M\big|_{[\cW\otimes\cS]_0}$. We now show that the closure of ${\bf \breve{T}}_M\big|_{[\cW\otimes\cS]_0}$ is the self-adjoint operator given by 
\begin{equation}
\begin{split}
{\bf \breve{T}}_M = & \ \ 
  [\Psi'_M\otimes \Psi_M]^*[{\bf I}\otimes{\bf\breve{Q}}_M][\Psi'_M\otimes \Psi_M] \\
&
   - [\Psi'_M\otimes \Psi_M]^*[{\bf\breve{Q}}'_M\otimes {\bf I}][\Psi'_M\otimes \Psi_M].
\end{split}
\label{eq:before_indA1_baseT3}
\end{equation}

Denote by $L^2\big((Y_M\times N_m)^2\big)_0$ the set of finite linear combinations of simple tensors in the $quantity_M$ domain $L^2\big((Y_M\times N_m)^2\big)=L^2(Y_M\times N_m;\cW)\otimes L^2(Y_M\times N_m;\cS)$ (see (\ref{eq:quant_m_dom_WS})). Denote
\begin{equation}
\begin{split}
{\bf\breve{W}}\big|_{L^2\big((Y_M\times N_m)^2\big)_0}= & [\Psi'_M\otimes \Psi_M]{\bf \breve{T}}_M\big|_{[\cW\otimes\cS]_0}[\Psi'_M\otimes \Psi_M]^*\\
= & [{\bf I}\otimes{\bf\breve{Q}}_M-{\bf\breve{Q}}'_M\otimes {\bf I}]\big|_{L^2\big((Y_M\times N_m)^2\big)_0}.
\label{eq:tempy5yy5q}
\end{split}
\end{equation}
Note that ${\bf\breve{W}}\big|_{L^2\big((Y_M\times N_m)^2\big)_0}$ is a multiplicative operator in the $quantity_M$ domain $L^2\big((Y_M\times N_M)^2\big)$. By (\ref{eq:tempy5yy5q}) and by the fact that $[\Psi'_M\otimes \Psi_M]$ is an isometric isomorphism, to show that the closure of ${\bf \breve{N}}_M\big|_{V(\cW\otimes\cS)_0}$ is ${\bf \breve{N}}_M$ it is enough to show that the closure of ${\bf\breve{W}}\big|_{L^2\big((Y_M\times N_m)^2\big)_0}$ is the self-adjoint multiplicative operator ${\bf\breve{W}}:={\bf I}\otimes{\bf\breve{Q}}_M-{\bf\breve{Q}}'_M\otimes {\bf I}$.

Let ${\bf 1}^m_L$ be the projection that restricts functions in $L^2\big((Y_M\times N_m)^2\big)$ to the domain 
\[D_m^L=\{(y'_m,g'_m,y_m,g_m)\ |\ -L\leq g^{\prime k}_m,g^k_m \leq L \ {\rm for\ every\ } k=1,\ldots,K_m\}.\]
 Let $F$ be in the domain of ${\bf \breve{W}}_M$.  For every $j\in\NN$, let  $L_j>1$ such that %$L_0>1$ such that for every $L>L_0$, 
\[\norm{F-{\bf 1}^M_{L_j} F}\leq \frac{1}{2j} \quad , \quad  \norm{{\bf \breve{W}}_M F-{\bf \breve{W}}_M {\bf 1}^M_{L_j} F}\leq \frac{1}{2j} .\] 
Indeed, $\breve{W}_M {\bf 1}^M_{L_j} F= {\bf 1}^M_{L_j}\breve{W}_M  F$, and $L_2$ norms can be approximated on large enough compact domains.
Any ${\bf 1}^{M}_{L_j} F$ can be approximated by a vector $F_j\in L^2\big((Y_M\times N_m)^2\big)_0$ up to error less than $\frac{1}{2jL_j}$, where the simple tensor components of $F_j$ are all supported in $D_m^{L_j}$, and are thus in the domain of ${\bf \breve{W}}_M\big|_{[\cW\otimes\cS]_0}$. Namely,
\[\norm{{\bf 1}^{M}_{L_j} F- F_j} \leq \frac{1}{2jL_j} \leq \frac{1}{2j}\]
and by the compact support, we also have
\[\norm{{\bf \breve{W}}_M{\bf 1}^{M}_{L_j} F- {\bf \breve{W}}_M\big|_{L^2\big((Y_M\times N_m)^2\big)_0}F_j} \leq L_j\frac{1}{2jL_j}= \frac{1}{2j}.\]
This shows that there is $F_j$ in the domain of ${\bf \breve{W}}_M\big|_{L^2\big((Y_M\times N_m)^2\big)_0}$ such that 
\[\lim_{j\rightarrow\infty}\norm{F- F_j} =0 \quad , \quad  
\lim_{j\rightarrow\infty}\norm{{\bf \breve{W}}_MF- {\bf \breve{W}}_M\big|_{L^2\big((Y_M\times N_m)^2\big)_0}F_j} =0 \]
which shows that the self-adjoint operator ${\bf\breve{W}}_M$ is the closure of ${\bf \breve{W}}_M\big|_{L^2\big((Y_M\times N_m)^2\big)_0}$.

Define the self-adjoint operator
\begin{equation}
{\bf \breve{N}}_M' =V{\bf \breve{T}}_MV^*.
\label{eq:NMprime}
\end{equation}
in the space $V(\cW\otimes\cS)$.
Next, we show that ${\bf \breve{N}}_M' = {\bf \breve{N}}_M\big|_{V(\cW\otimes\cS)}$. Similarly to before, we define by abuse of notation ${\bf 1}^m_L$ to be the projection that restricts functions from $L^2(G)$ to the domain
\[D_m^L=\{g\in G\ |\ -L\leq g^k_m \leq L \ {\rm for\ every\ } k=1,\ldots,K_m\}.\]
Let $F$ be in the domain of ${\bf \breve{N}}_M'$ .
Choose an  approximating sequence  $\{F_j\}_{j\in\NN}\subset[\cW\otimes\cS]_0$, such that% $\bT_m F_j$ approximates $\bT_m F$, and  
 \[\lim_{j\rightarrow\infty}V(F_j)= V(F) \quad , \quad \lim_{j\rightarrow\infty}{\bf\breve{N}}_M' V(F_j)={\bf \breve{N}}_M' V(F).\] 
Indeed, this is possible since every self-adjoint operator is closed.

%By the fact that $L_2$ convergence implies pointwise convergence a.e up to a subsequence, we have for every $L>0$ and a.e $g\in G$, up to a subsequence of $j$,
%\[ \begin{split}
%[{\bf\breve{N}}_M' V(F)](g)  %& ={\bf\breve{N}}_M' \lim_{j\rightarrow\infty}V(F_j)\\
% &= [\lim_{j\rightarrow\infty} {\bf\breve{N}}_M' V(F_j)](g)\\
% &= [\lim_{j\rightarrow\infty}{\bf\breve{N}}_M\big|_{V(\cW\otimes\cS)_0} V (F_j)](g)\\
% &= [{\bf g}_M\lim_{j\rightarrow\infty}  V (F_j)](g)\\
% &= [{\bf 1}^M_L\lim_{j\rightarrow\infty}    V (F_j)](g) + [({\bf I}-{\bf 1}^M_L)\lim_{j\rightarrow\infty}  {\bf g}_M V (F_j)](g)\\
% &= \lim_{j\rightarrow\infty} {\bf 1}^M_L  {\bf g}_M V (F_j)(g) +  [({\bf I}-{\bf 1}^M_L)\lim_{j\rightarrow\infty}  {\bf\breve{N}}_M' V ( F_j)](g)\\
% &= \lim_{j\rightarrow\infty} {\bf 1}^M_L  {\bf g}_M V (F_j)(g) +  [({\bf I}-{\bf 1}^M_L) {\bf\breve{N}}_M' \lim_{j\rightarrow\infty} V ( F_j)](g)\\
% &= \lim_{j\rightarrow\infty} {\bf 1}^M_L  {\bf g}_M V (F_j)(g) +  [({\bf I}-{\bf 1}^M_L) {\bf\breve{N}}_M' V (F)](g).
%\end{split}\]

We have for every $L>0$

\[ \begin{split}
{\bf\breve{N}}_M' V(F)  %& ={\bf\breve{N}}_M' \lim_{j\rightarrow\infty}V(F_j)\\
 &= {\bf 1}^M_L{\bf\breve{N}}_M' V(F) + ({\bf I}-{\bf 1}^M_L){\bf\breve{N}}_M' V(F)\\
 &= {\bf 1}^M_L\lim_{j\rightarrow\infty} {\bf\breve{N}}_M' V(F_j) + ({\bf I}-{\bf 1}^M_L){\bf\breve{N}}_M' V(F)\\
 &={\bf 1}^M_L\lim_{j\rightarrow\infty} {\bf\breve{N}}_M V(F_j) + ({\bf I}-{\bf 1}^M_L){\bf\breve{N}}_M' V(F)
\end{split}\]
By continuity of the projection ${\bf 1}^M_L$
\[
{\bf\breve{N}}_M' V(F) = \lim_{j\rightarrow\infty} {\bf 1}^M_L{\bf\breve{N}}_M V(F_j) + ({\bf I}-{\bf 1}^M_L){\bf\breve{N}}_M' V(F).\]
Note that ${\bf 1}^M_L{\bf\breve{N}}_M$ is an operator that multiplies by a bonded function, and is thus continuous. Also note that $({\bf I}-{\bf 1}^M_L)$ converges to zero in the strong topology as $L\rightarrow\infty$.  Therefore
\[
\begin{split}
{\bf\breve{N}}_M' V(F) &= {\bf 1}^M_L{\bf\breve{N}}_M  \lim_{j\rightarrow\infty}   V (F_j) \ \ \ +  o(1;L)\\
 & = {\bf 1}^M_L{\bf\breve{N}}_M  V (F) \ \ \ +  o(1;L)
\end{split}
\]
where $o(1;L)$ is a vector that converges to zero as $L$ goes to infinity.
This is true for any $L$, so
\[ {\bf\breve{N}}_M' V(F) ={\bf\breve{N}}_M V(F),\]
or\footnote{Inclusion of operators $T_1\subset T_2$ means that the domain ${\cal D}_1$ of $T_1$ is contained in the domain  ${\cal D}_2$ of $T_2$, and $T_1 v=T_2 v$ for every $v\in {\cal D}_1$,}
\[ {\bf\breve{N}}_M' \subset {\bf\breve{N}}_M\big|_{V(\cW\otimes\cS)}.\]

In this situation we must have ${\bf\breve{N}}_M' = {\bf\breve{N}}_M\big|_{V(\cW\otimes\cS)}$, as shown next.  Denote the domain of ${\bf\breve{N}}_M'$ by $V(\cW\otimes\cS)_1$, and note that $V(\cW\otimes\cS)_1$ is contained both in $V(\cW\otimes\cS)$ and in the domain ${\cal D}({\bf\breve{N}}_m)$ of ${\bf\breve{N}}_m$. Moreover, for every $F\in V(\cW\otimes\cS)_1$, ${\bf\breve{N}}_M'F={\bf\breve{N}}_MF$. Now, recall that the domain of the adjoint of ${\bf\breve{N}}_M'$ is defined to be the set of all $Q\in V(\cW\otimes\cS)$ for which there exists $Q^*$ satisfying
\[\forall F\in V(\cW\otimes\cS)_1 \ . \quad \ip{{\bf\breve{N}}_M' F}{Q} = \ip{F}{Q^*}.\]
Now, for every $F\in V(\cW\otimes\cS)_1$ and $Q\in {\cal D}({\bf\breve{N}}_m)\cap V(\cW\otimes\cS)$
\[\ip{{\bf\breve{N}}_M' F}{Q} = \ip{{\bf\breve{N}}_M F}{Q} = \ip{F}{{\bf\breve{N}}_M Q}.\]
This means that any $Q\in {\cal D}({\bf\breve{N}}_m)\cap V(\cW\otimes\cS)$ is in the domain of ${\bf\breve{N}}_M^{\prime *}$. Thus, the domain of ${\bf\breve{N}}_M^{\prime *}$  which is equal to the domain of ${\bf\breve{N}}_M'$, contains ${\cal D}({\bf\breve{N}}_m)\cap V(\cW\otimes\cS)$. To conclude, 
\[{\bf\breve{N}}_M' = {\bf\breve{N}}_M\big|_{V(\cW\otimes\cS)}.\]

Let us now treat the case $m<M$ and $G_m$ is $\RR$ or $\ZZ$. By the induction assumption, for all $m'>m$, ${\bf \breve{N}}_{m'}\big|_{V(\cW\otimes\cS)}$ is self-adjoint or unitary in $V(\cW\otimes\cS)$, 
so
\[V^*{\bf A}_m({\bf\breve{N}}_{m+1},\ldots,{\bf \breve{N}}_M)V= V^*{\bf A}_m({\bf\breve{N}}_{m+1}\big|_{V(\cW\otimes\cS)},\ldots,{\bf \breve{N}}_M\big|_{V(\cW\otimes\cS)})V\]
and by
(\ref{eq:pull_func_calc})
\[V^*{\bf A}_m({\bf\breve{N}}_{m+1},\ldots,{\bf \breve{N}}_M)V= {\bf A}_m({\bf\breve{T}}_{m+1},\ldots,{\bf \breve{T}}_M).\]
Therefore, by (\ref{eq:before_indA1})
\begin{equation}
\begin{split}
{\bf \breve{N}}_m\big|_{V(\cW\otimes\cS)_0}= & \ \ 
 V\left([\Psi'_m\otimes \Psi_m]^*[{\bf I}\otimes{\bf\breve{Q}}_m][\Psi'_m\otimes \Psi_m]\right. \\
&  \left.-{\bf A}_m({\bf\breve{T}}_{m+1},\ldots,{\bf \breve{T}}_M) [\Psi'_m\otimes \Psi_m]^*[{\bf\breve{Q}}'_m\otimes {\bf I}][\Psi'_m\otimes \Psi_m]\right)V^*.
\end{split}
\label{eq:before_indA122}
\end{equation}
We define
\begin{equation}
\begin{split}
{\bf \breve{T}}_m\big|_{[\cW\otimes\cS]_0}= & V^*{\bf \breve{N}}_m\big|_{V(\cW\otimes\cS)_0}V\\
 = & [\Psi'_m\otimes \Psi_m]^*[{\bf I}\otimes{\bf\breve{Q}}_m][\Psi'_m\otimes \Psi_m] \\
&  -{\bf A}_m({\bf\breve{T}}_{m+1},\ldots,{\bf \breve{T}}_M) [\Psi'_m\otimes \Psi_m]^*[{\bf\breve{Q}}'_m\otimes {\bf I}][\Psi'_m\otimes \Psi_m].
\end{split}
\label{eq:before_indA1222}
\end{equation}
Now, by the fact that ${\bf \breve{N}}_m\big|_{V(\cW\otimes\cS)_0}$ is symmetric or isometric, so is ${\bf \breve{T}}_m\big|_{[\cW\otimes\cS]_0}$. We now use the assumption that the closure of ${\bf \breve{T}}_m\big|_{[\cW\otimes\cS]_0}$ is the self-adjoint operator in $\cW\otimes\cS$ given by
\begin{equation}
\begin{split}
{\bf \breve{T}}_m %= & V^*{\bf \breve{N}}_m\big|_{V(\cW\otimes\cS)}V\\
 = & [\Psi'_m\otimes \Psi_m]^*[{\bf I}\otimes{\bf\breve{Q}}_m][\Psi'_m\otimes \Psi_m] \\
&  -{\bf A}_m({\bf\breve{T}}_{m+1},\ldots,{\bf \breve{T}}_M) [\Psi'_m\otimes \Psi_m]^*[{\bf\breve{Q}}'_m\otimes {\bf I}][\Psi'_m\otimes \Psi_m].
\end{split}
\label{eq:before_indA12226}
\end{equation}
We then define 
\[{\bf \breve{N}}_m'=V{\bf \breve{T}}_mV^*,\]
and show as before that
\[{\bf \breve{N}}_m'={\bf \breve{N}}_m\big|_{V(\cW\otimes\cS)}\]
is self-adjoint in $V(\cW\otimes\cS)$.

We now continue from (\ref{eq:temp4hh08080}) 
in case $G_m$ is $e^{i\RR}$ or $e^{2\pi i \ZZ/N}$.  For a tuple of signals ${\bf s}$, by the fact that the SPWT is simply dilated, we have
\[V_{\Psi_m^*{\bf\breve{Q}}_m \Psi_m f}[{\bf s}](g)=\iint [\Psi_m {\bf s}]({\bf x}_m,y) \overline{[{\bf{g}}_m^{-1}\bullet {\bf\breve{Q}}_m \rho^m(g) \Psi_m f]({\bf x}_m,y)}  d{\bf x}_m dy.\]
Here, $\bullet$ is term by term multiplication, so
\[\begin{split}
V_{\Psi_m^*{\bf\breve{Q}}_m \Psi_m f}[{\bf s}](g)&=\overline{{\bf{g}}_m^{-1}}\iint [\Psi_m {\bf s}]({\bf x}_m,y) \overline{[ {\bf\breve{Q}}_m \rho^m(g) \Psi_m f]({\bf x}_m,y)}  d{\bf x}_m dy\\
&={\bf{g}}_m\bullet\iint {\bf\breve{Q}}_m^* [\Psi_m {\bf s}]({\bf x}_m,y) \overline{[ \rho^m(g) \Psi_m f]({\bf x}_m,y)}  d{\bf x}_m dy\\
&={\bf{g}}_m\bullet V_{f}[\Psi_m^*{\bf\breve{Q}}_m^* \Psi_m {\bf s}](g)
\end{split}\]
where the multiplication of $\Psi_m^*{\bf\breve{Q}}_m^*$ with  $\Psi_m {\bf s}$ is term by term.
So, by 
 plugging in ${\bf s}= \Psi_m^*{\bf\breve{Q}}_m \Psi_m s'$, and denoting ${\bf s}'=(s')_{k=1}^{K_m}$ we get
\[V_{\Psi_m^*{\bf\breve{Q}}_m \Psi_m f}[\Psi_m^*{\bf\breve{Q}}_m \Psi_m s'] =
{\bf{g}}_m\bullet V_{f}[{\bf s}'](g)\]
Thus, by changing the notation $s'\mapsto s$, and since $\bullet$ is term by term multiplication,
\[V_{\Psi_m^*{\bf\breve{Q}}_m \Psi_m f}[\Psi_m^*{\bf\breve{Q}}_m \Psi_ms] ={\bf{g}}_mV_{f}[s](g).\]
As a result,
\begin{equation}
\begin{split}
{\bf \breve{N}}_m V(f\otimes s)=& {\bf \breve{N}}_mV_f[s]= V_{\Psi_m^*{\bf\breve{Q}}_m \Psi_m f}[\Psi_m^*{\bf\breve{Q}}_m \Psi_ms]\\
= & V\Big((\Psi_m^*{\bf\breve{Q}}_m \Psi_m f)\otimes(\Psi_m^*{\bf\breve{Q}}_m \Psi_ms)\Big).
\end{split}
\label{eq:ihgfd4ud2}
\end{equation}
Similarly to (\ref{eq:before_indA1}),  
 (\ref{eq:ihgfd4ud2}) extends to non-simple tensors by
\begin{equation}
{\bf \breve{N}}_m\big|_{V(\cW\otimes\cS)}
= V[\Psi'_m\otimes \Psi_m]^*[{\bf\breve{Q}}_m'\otimes{\bf\breve{Q}}_m][\Psi'_m\otimes \Psi_m]V^*.
\label{eq:before_indA2}
\end{equation}

Next, we sketch the proof of 1, namely the invariance of $V(\cW\otimes\cS)$ under $\breve{\mathcal{M}}^{\infty}(G)$ operators. We start with an intuitive explanation, and then extend to a proper argument.
Let us denote by ${\cal M}^{1,\infty}(G)$ the space of bounded measurable functions on $G$, independent of the variable along the center $Z$, and integrable along the variables of $G/Z$. 
By the above result, we can show that any multiplicative operator by an ${\cal M}^{1,\infty}(G)$ function maps $V(\cW\otimes\cS)$ to itself. Indeed, ${\cal M}^{1,\infty}(G)$ functions can be approximated by linear combinations of characteristic functions of rectangular domains in $G/Z$. The operators that multiply by characteristic functions of rectangular domains in $G/Z$, are exactly the spectral projections of $({\bf \breve{N}}_1,\ldots,{\bf \breve{N}}_N)$ as defined in Remark \ref{simul_A_m}. Thus, multiplicative operators that multiply by ${\cal M}^{1,\infty}(G)$ functions commute with the projection $P_{V(\cW\otimes\cS)}$ upon $V(\cW\otimes\cS)$. More generally, the minimal von Neumann algebra\footnote{A subalgebra of the bounded operators in $L^2(G)$, which is closed under taking adjoints, contains the identity operator, and is closed with respect to the strong operator topology.} containing the spectral projections of $({\bf \breve{N}}_1,\ldots,{\bf \breve{N}}_N)$ are the multiplicative operators by ${\cal M}^{\infty}(G)$ functions \cite{VN_algebra}.

 \end{proof}

\subsection{Proof of Proposition \ref{prop:VWS}}

\begin{proof}
Part 1 are general properties. The center of any topological group that has a square integrable representation is compact (see e.g. \cite{Fuhr_wavelet}). In addition, the restriction of any representation to the center of its group is a scalar operator.

Next, we sketch the proof of 2. The idea is to generate an approximation to any function satisfying (\ref{eq:prop:VWS}). We do this using multiplicative operators applied on the wavelet-Plancherel transforms of simple tensors, as building blocks. For any $g\in G/Z$, there are $f\in\cW\cap\cS$ and $s\in\cS$ such that $V(f\otimes s)(g)\neq 0$. Indeed, we can take $f\in\cW\cap\cS$ and $s=\pi(g)f$. By continuity, $V(f\otimes s)$ is nonzero in a neighborhood of $g$. We can now construct an $\breve{{\cal M}}^{\infty}(G)$ multiplicative operator $\breve{R}$, such that $\breve{R}V(f\otimes s)$ is a characteristic function of a small neighborhood of $g$. Note that these characteristic functions on $G/Z$, extend via $\overline{\chi({\bf z})}$ in the $Z$ direction.
Indeed, there exists $Q\in L^{\infty}(G/Z)$ such that
\begin{equation}
\begin{split}
\breve{R}V(f\otimes s)= & Q({\bf g}_1,\ldots,{\bf g}_M)\ip{s}{\pi({\bf z},{\bf g}_1,\ldots,{\bf g}_M)(g)f}\\
= & \overline{\chi({\bf z})}Q({\bf g}_1,\ldots,{\bf g}_M)\ip{s}{\pi({\bf e},{\bf g}_1,\ldots,{\bf g}_M)f}.
\end{split}
\label{eq:5rts8}
\end{equation}
 We can thus approximate any continuous function of the form (\ref{eq:prop:VWS}), and by a density argument, prove that the space of functions of the form (\ref{eq:prop:VWS}) is contained in $V(\cW\otimes\cS)$. % 2 of Proposition \ref{prop:VWS}. 
The inclusion of $V(\cW\otimes\cS)$ in the space of functions satisfying (\ref{eq:prop:VWS}) is shown by (\ref{eq:5rts8}), with $Q=1$.

\end{proof}

\subsection{Proof of Proposition \ref{pullDGWT}}

\begin{proof}
We start with showing 1.
This is a result of Kleppner and Lipsman's theorem \cite{Kleppner}, on $\RR^N \rtimes H_1$. We give a short explanation, without an exposition of the Plancherel theory for nonunimodular type I locally compact topological groups. For more details we refer the reader to \cite[Section 3]{Fuhr_wavelet}, and to \cite{geometric} and \cite{Qregula1} for the dual orbit and stabilizer of general dilation group wavelets based on quasi-regular representations. We outline the proof as follows. By the definition of DGWT, there is only one dual open orbit $U$, and for every $\gamma\in U\subset\RR^d= \hat{\RR}^d$, the stabilizer $\hat{H}_{\gamma}$ is $\{e\}$. Therefore, by Kleppner and Lipsman's theorem, the Plancherel measure of the space of unitary irreducible representations $\hat{G}$ is concentrated at one discrete series representation, which must be the DGWT representation $\pi$.  We note that this argument can be found, e.g., in Example 3.12 of \cite{Kep_ex1}, and \cite{Kep_ex0}.

In this situation, the Plancherel theorem \cite[Section 3]{Fuhr_wavelet} gives the following results.
The subspace $L^2(G)'$, defined as the inverse Plancherel transform of vector fields supported on the conull set $\{\pi\}\subset \hat{G}$, is a dense subspace of $L^2(G)$. In addition, $L^2(G)'$ is invariant under the left translation $L_G$, defined for $F\in L^2(G)$ by $L_G(g)F(g')=F(g^{-1}g')$. Moreover, $L_G\big|_{L^2(G)'}$ is unitarily equivalent to $\pi^{\kappa}$, acting on $\cS^{\kappa}$ by 
\[\pi^{\kappa}(g)\{s_j\}_{j=1}^{\kappa} = \{\pi(g)s_j\}_{j=1}^{\kappa}. \] 
Here, $\kappa\in\NN\cup\{\infty\}$ is the multiplicity of $\pi$ in $L_G\big|_{L^2(G)'}$, and $\cS^{\infty}$ is defined in the $l^2$ sense (as sequences with square summable norms $\{\norm{s_j}\}_{j=1}^{\infty}$). Thus,
\begin{equation}
L^2(G)= \overline{span\{\cH\subset L^2(G)\ :\ \cH {\rm\ is\ an\ invariant\ subspace\ of\ } L_G,\  L_G\big|_{\cH}\simeq\pi\}}
\label{eq:RepPlanch}
\end{equation}
In Theorem 2.33 of \cite{Fuhr_wavelet}, the right-hand side of (\ref{eq:RepPlanch}) is denoted by $L^2_{\pi}(G)$, and by part (a) of that theorem,
\begin{equation}
L^2_{\pi}(G) = \overline{{\rm span}\{V_f[s]\ |\ f\in \cS\cap\cW \ ,\ s\in\cS\}}.
\label{eq:L2pi}
\end{equation}
 Since the right-hand side of (\ref{eq:L2pi}) is precisely $V(\cW\otimes\cS)$, together with (\ref{eq:RepPlanch}), we conclude $L^2(G) =V(\cW\otimes\cS)$.

We continue with proving 2. First, to justify the use of Proposition \ref{prop:SPWT_pull_obs}, we use the fact that $V(\cW\otimes\cS)=L^2(G)$ and Remark \ref{no_closure_needed}.
Consider  a DGWT and its diffeomorphism $quantity_m$ transforms.
We first derive an explicit characterization of the pull-back of $R({ \breve{G}}_m^k)$ multiplicative operators for $m=2,\ldots,M$, and $R\in  L^{\infty}(G_m^k)$. Observe that all of the ${\bf\bT}_m$ operators, with $2\leq m\leq M$, are multiplicative operators in the tensor product frequency domain $L^2(U^2)= L^2(U;K^2(\w)d\w)\otimes L^2(U;d\w)$. We can show this by induction on (\ref{eq:pull1_self_adjoint2}) and(\ref{eq:pull1_self_adjoint2a1}) of Proposition \ref{prop:SPWT_pull_obs}, starting at $m=M$ and decreasing $m$. We rely on the fact that the diffeomorphism $quantity_m$ transforms are based on change of variables in the frequency domain, and thus all of the ${\bf\bQ}_m,{\bf\bQ}_m'$ operators are multiplicative. 
For $m=M$, ${\bf\bT}_M$ is a multiplicative operator, since ${\bf\bQ}_M,{\bf\bQ}_M'$ are. In the induction step we assume that ${\bf\bT}_{m+1},\ldots,{\bf\bT}_{M}$ are multiplicative operators, and thus so is ${\bf A}_m({\bf\bT}_{m+1},\ldots,{\bf\bT}_{M})$. Since ${\bf\bQ}_m,{\bf\bQ}_m'$ are also multiplicative operators, so is ${\bf\bT}_m$. 

Now, we study the pull-back of multiplicative operators along the $position$ variables, namely $R({\bf\bT}_1)$ operators.
%Recall that we restrict ourselves to trigonometric polynomials of the form (\ref{eq:trigo_poly_filt}). 
We are interested in the calculation of $\exp({i t {\bf\bT}_1})$. 
We have
\[\Psi_1^*{\bQ}_m^k\Psi_1  = i \frac{\partial}{\partial \w_k},\]
and denote $i \frac{\partial}{\partial \bw}=\Big( i \frac{\partial}{\partial \w_k} \Big)_{k=1}^{K_1}$.
By (\ref{eq:pull1_self_adjoint}) of Proposition \ref{prop:SPWT_pull_obs},
\[{\bf\bT}_1 (\hf \otimes \hs) = \hf \otimes [i \frac{\partial}{\partial \bw} \hs]
 - {\bf A}_m({\bf\bT}_2,\ldots,{\bf\bT}_N)\Big( (i \frac{\partial}{\partial \bw}\hf)\otimes \hs \Big).\]
where ${\bf A}_m({\bf\bT}_2,\ldots,{\bf\bT}_N)$ is a multiplicative operator valued matrix in $L^2(U^2)$. 
By Remark \ref{no_closure_needed}, the closure of this symmetric operator, defined for (finite) linear combinations of simple functions, is self-adjoint. Since an operator is essentially self-adjoint if its closure is self-adjoint, there is a unique way to extend the domain of ${\bf\bT}_1$ to a self-adjoint operator. Hence, the closure of ${\bf\bT}_1$ must be the self-adjoint differential operator
\begin{equation}
[{\bf\bT}_1 F](\bw',\bw) = i \frac{\partial}{\partial \bw}F(\bw',\bw)
 + {\bf A}_m({\bf\bT}_2,\ldots,{\bf\bT}_N)i \frac{\partial}{\partial \bw'}F(\bw',\bw).
\label{eq:time_obs11000}
\end{equation}
The operator $\exp({it {\bf\bT}_1})$ is defined for $F\in\cW\otimes\cS$ via the differential equation
\begin{equation}
\begin{split}
\frac{\partial}{\partial t}\big[\exp({it {\bf\bT}_1})F\big]\ \ \ \ = \ \ &i {\bf\bT}_1 \big[\exp({it {\bf\bT}_1})F\big] = -\Big(\frac{\partial}{\partial \bw}
 + {\bf A}_m({\bf\bT}_2,\ldots,{\bf\bT}_N) \frac{\partial}{\partial \bw'}\Big)\big[\exp({it {\bf\bT}_1})F\big]\\
\big[\exp({it {\bf\bT}_1})F\big]_{t=0}= \ \ & F.
\end{split}
\label{eq:flow_pos1000}
\end{equation}

\end{proof}

\section{Appendix C: pull-back formulas for useful GDWTs}
\label{Examples222r}

The $qunatity_m$ transforms of the wavelet and the Shearlet transforms from Subsection \ref{Examples_SPWT} are diffeomorphism $quantity_m$ transforms. In the following we solve (\ref{eq:flow_pos1}) for our examples of DGWT.

%$ $
%
%\noindent
%\emph{The 1D wavelet transform}
\subsection{The 1D wavelet transform}
\label{A:The 1D wavelet transform}

By (\ref{eq:pull1_self_adjoint}), the operator $\bT_{2}$ is given by
\[{\bT}_2(\hf\otimes \hs) =  \hf\otimes(\Psi_2^*{\bQ}_2\Psi_2 \hs) - (\Psi_2^*{\bQ}_2\Psi_2 \hf)\otimes \hs\]
where $[\Psi_2^*{\bQ}_2\Psi_2 \hs](\w)=-\ln(\w)\hs(\w)$.
Therefore
\[\bT_2 F(\w',\w) = \big( \ln(\w')-\ln(\w) \big)F(\w',\w).\]

The operator $\bT_1$ is given by
\[{\bT}_1(\hf\otimes \hs) = \hf\otimes(i\frac{\partial}{\partial \w} \hs) - A_1(\bT_2)(i\frac{\partial}{\partial \w'} \hf)\otimes \hs\]
or
\[[\bT_1(\hf\otimes \hs)](\w',\w)= \Big(i\frac{\partial}{\partial \w} - e^{\bT_2}(-i)\frac{\partial}{\partial \w'}\Big)\hs(\w)\overline{\hf(\w')}.\]
Therefore
\[\bT_1 F(\w',\w)= \Big(i\frac{\partial}{\partial \w} + \frac{\w'}{\w}i\frac{\partial}{\partial \w'}\Big)F(\w',\w).\]

The operators $\exp(it\bT_1)$ on $F\in\cW\otimes\cS$  are written as
\[\exp(it\bT_1) F(\w',\w) = F(\w',\w;t)\]
where $F$ satisfies the first order linear PDE
\[\frac{\partial}{\partial t}{ F}(\w',\w;t) =\Big(-\frac{\partial}{\partial \w} - \frac{\w'}{\w}\frac{\partial}{\partial \w'}\Big)
F(\w',\w;t). \]
Denote the integral lines by ${\bf r}(t;\w'_0,\w_0)$, with curve parameter $t$ and initial position $(\w'_0,\w_0)$.
Moreover, denote ${\bf r}=(\w',\w):\RR\rightarrow\RR^2$. the curves ${\bf r}$ satisfy
\[\dot{{\bf r}}(t)=(\dot{\w'}(t),\dot{\w}(t))= \Big(\frac{\w'}{\w},1\Big)\]
with initial condition
\[\Big(\w'(0),\w(0)\Big)= \Big(\w'_0,\w_0\Big).\]
Thus
\begin{equation}
\w(t)=t+\w_0.
\label{eq:scale_x_flow11}
\end{equation}
\begin{equation}
\w'(t) = \frac{\w'_0}{\w_0}(t+\w_0).
\label{eq:scale_x_flow22}
\end{equation}

The above calculations of $\bT_{2}$ allow the application of perfect $scale$ bandpass filters, projecting upon the scale band $\{(g_1,g_2)\ |\ g_2\in[a,b]\}$. 
Namely,
\[P^2_{[a,b]}(\w',\w) = \left\{  \begin{array}{ccc}
	1 & , & a\leq \ln(\frac{\w'}{\w})\leq b\\
	0 & , & {\rm otherwise}
\end{array}
\right.\]
\[ = \left\{  \begin{array}{ccc}
	1 & , & e^a\leq \frac{\w'}{\w}\leq e^b\\
	0 & , & {\rm otherwise}
\end{array}
\right.\]

The operators $\exp({it{\bT}_1})$ allow the calculation of trigonometric polynomial approximate periodic band-pass filters in $time$. Namely, time-pass are convolution along slopes given by (\ref{eq:scale_x_flow11})(\ref{eq:scale_x_flow22}).

%$ $

%\noindent
%\emph{The STFT}
\subsection{The STFT}

Similarly to the 1D wavelet transform, it can be shown that time-pass filters are convolutions along diagonals 
\begin{equation}
\w(t)=t+\w_0.
\label{eq:scale_x_flow11o}
\end{equation}
\begin{equation}
\w'(t) = t+\w'_0.
\label{eq:scale_x_flow22o}
\end{equation}
Frequency-pass filters are multiplicative operators that restrict functions $F\in\cW\otimes\cS$ to a band of diagonals.

%$ $

%\noindent
%\emph{Rotation-dilation wavelet transform}
\subsection{Rotation-dilation wavelet transform}

In this section the signal and window spaces are assumed to be frequency domains. We denote the two frequency variables by $(\x,\y)$.
Consider the $angle$$-$$scale$ transform  
\[{\boldsymbol\Psi}_2=(\Psi_2^1,\Psi_2^2): L^2(\RR^2)\rightarrow L^2(e^{i\RR}\times\RR)\]
 of (\ref{eq:angle-scale}). The corresponding observable $\boldsymbol{\Psi}_2^*{\bf\bQ}_2\boldsymbol{\Psi}_2: L^2(\RR^2)\rightarrow L^2(\RR^2)^2$ is given by
\[[\boldsymbol{\Psi}_2^*{\bf\bQ}_2\boldsymbol{\Psi}_2 \hf](\x,\y)=\Big(\abs{(\x,\y)}^{-1}( \x + i\y ),-\ln\abs{(\x,\y)}\Big)\hf(\x,\y).\]

According to formula (\ref{eq:pull1_unitary}), the $angle$ observables $\bT_2$ is given by
\[{\bT}_2(\hf\otimes \hs) = (\Psi_2^*{\bQ}^1_2\Psi_2 \hf)\otimes(\Psi_2^*{\bQ}^1_2\Psi_2 \hs)\]
so
\[[\bT_2(\hf\otimes \hs)](\x',\y',\x,\y)= \frac{ \x + i\y }{\abs{(\x,\y)}}\hs(\x,\y)\overline{\frac{ \x' + i\y' }{\abs{(\x',\y')}}\hf(\x',\y')}.\]
Therefore
\[\bT_2 F(\x',\y',\x,\y) = \frac{(\x + i\y)(\x' - i\y')}{\abs{(\x,\y)}\abs{(\x',\y')}} F(\x',\y',\x,\y).\]

The $scale$ observable is given by (\ref{eq:pull1_self_adjoint}) as
\[{\bT}_3(\hf\otimes \hs) = \hf\otimes(\Psi_2^*{\bQ}^2_2\Psi_2 \hs) - (\Psi_2^*{\bQ}^2_2\Psi_2 \hf)\otimes \hs\]
or
\[[\bT_3(\hf\otimes \hs)](\x',\y',\x,\y)= \Big(\ln(\abs{(\x',\y')})-\ln(\abs{(\x,\y)})\Big)\hs(\x,\y)\overline{\hf(\x',\y')}.\]
Therefore
\[\bT_3 F(\x',\y',\x,\y)= \Big(\ln(\abs{(\x',\y')})-\ln(\abs{(\x,\y)})\Big)F(\x',\y',\x,\y).\]

To derive a formula for $angle$-pass filters, note that the space of functions $L^{\infty}(e^{i\RR})$, applied on $\bT_2$, is exhausted by the space of functions of the form
\[\big\{ \eta\big(\frac{(\x + i\y)(\x' - i\y')}{\abs{(\x,\y)}\abs{(\x',\y')}}\big) \ |\ \eta\in L^{\infty}(e^{i\RR})\big\}.\]
For $scale$-pass filters, note that the space of functions $L^{\infty}(\RR)$, applied on $\bT_3$, is exhausted by the space of functions of the form
\[\big\{ \eta(\frac{\abs{(\x',\y')}}{\abs{(\x,\y)}}) \ |\ \eta\in L^{\infty}(\RR)\big\}.\]

To calculate ${\bf\bT}_1$,
we use
\[{\bf A}_1(g_2,g_3) = \bD(g_2,g_3)  =e^{g_3}\left(
\begin{array}{cc}
	 \Re(g_2) & -\Im(g_2) \\
	\Im(g_2) & \Re(g_2)
\end{array}
\right)\]
the get
\[{\bf A}_1(\bT_2,\bT_3) =\frac{\abs{(\x',\y')}}{\abs{(\x,\y)}}\left(
\begin{array}{cc}
	 \frac{\x\x'+\y\y'}{\abs{(\x,\y)}\abs{(\x',\y')}} & -\frac{-\x\y'+\x'\y}{\abs{(\x,\y)}\abs{(\x',\y')}} \\
	\frac{-\x\y'+\x'\y}{\abs{(\x,\y)}\abs{(\x',\y')}} & \frac{\x\x'+\y\y'}{\abs{(\x,\y)}\abs{(\x',\y')}}
\end{array}
\right).\]
Therefore, by (\ref{eq:pull1_self_adjoint}),
\[{\bf\bT}_1 (\hf\otimes \hs)=\hf\otimes(\Psi_1^*{\bf\bQ}_1^x\Psi_1 \hs) - A_1(\bT_2,\bT_3)\big((\Psi_1^*{\bf\bQ}^x\Psi_1 \hf)\otimes \hs\big)\]
or
\begin{equation}
\begin{split}
& {\bf\bT}_1(\hf\otimes \hs) (\x',\y',\x,\y)=  \left(\begin{array}{c}
	i\frac{\partial}{\partial \x}\\
	i\frac{\partial}{\partial \y}
\end{array}\right)\overline{\hf(x',y')}\hs(\x,\y) \\
& \quad \quad - \frac{\abs{(\x',\y')}}{\abs{(\x,\y)}}\left(
\begin{array}{cc}
	 \frac{\x\x'+\y\y'}{\abs{(\x,\y)}\abs{(\x',\y')}} & -\frac{-\x\y'+\x'\y}{\abs{(\x,\y)}\abs{(\x',\y')}} \\
	\frac{-\x\y'+\x'\y}{\abs{(\x,\y)}\abs{(\x',\y')}} & \frac{\x\x'+\y\y'}{\abs{(\x,\y)}\abs{(\x',\y')}}
\end{array}
\right)\left(\begin{array}{c}
	-i\frac{\partial}{\partial \x'}\\
	-i\frac{\partial}{\partial \y'}
\end{array}\right)\overline{\hf(\x',\y')}\hs(\x,\y).
\end{split}
\end{equation}
Thus
\[
\begin{split}
 & {\bf\bT}_1 F(\x',\y',\x,\y) = \\
 & \left(
\begin{array}{c}
	 i\frac{\partial}{\partial \x}+ \frac{\abs{(\x',\y')}}{\abs{(\x,\y)}}\frac{\x\x'+\y\y'}{\abs{(\x,\y)}\abs{(\x',\y')}}i\frac{\partial}{\partial \x'} - \frac{\abs{(\x',\y')}}{\abs{(\x,\y)}}\frac{-\x\y'+\x'\y}{\abs{(\x,\y)}\abs{(\x',\y')}}i\frac{\partial}{\partial \y'} \\
	i\frac{\partial}{\partial \y} + \frac{\abs{(\x',\y')}}{\abs{(\x,\y)}}\frac{-\x\y'+\x'\y}{\abs{(\x,\y)}\abs{(\x',\y')}}i\frac{\partial}{\partial \x'}  + \frac{\abs{(\x',\y')}}{\abs{(\x,\y)}} \frac{\x\x'+\y\y'}{\abs{(\x,\y)}\abs{(\x',\y')}}i\frac{\partial}{\partial \y'}
\end{array}
\right)F(\x',\y',\x,\y). 
\end{split}\]

The exponential $\exp(it{\bf\bT}_1)$ of $F\in\cW\otimes\cS$ is written as
\[\exp(it{\bf\bT}_1) F(\x',\y',\x,\y) = {\bf F}(\x',\y',\x,\y;t)\]
where ${\bf F}(\ \cdot\ ;t)\in (\cW\otimes\cS)^2$ for every $t\in\RR$,
with initial condition
\[{\bf F}(\x',\y',\x,\y;0)=\Big(F(\x',\y',\x,\y),F(\x',\y',\x,\y)\Big).\]
The pair of functions ${\bf F}$ satisfies the two first order linear PDEs
\[\frac{\partial}{\partial t}{\bf F}(\x',\y',\x,\y;t) =
\left(
\begin{array}{c}
	 -\frac{\partial}{\partial \x}- \frac{\x\x'+\y\y'}{\abs{(\x,\y)}^2}\frac{\partial}{\partial \x'} + \frac{-\x\y'+\x'\y}{\abs{(\x,\y)}^2}\frac{\partial}{\partial \y'} \\
	-\frac{\partial}{\partial \y} - \frac{-\x\y'+\x'\y}{\abs{(\x,\y)}^2}\frac{\partial}{\partial \x'}  - \frac{\x\x'+\y\y'}{\abs{(\x,\y)}^2}\frac{\partial}{\partial \y'}
\end{array}
\right){\bf F}(\x',\y',\x,\y;t) \]

Denote the corresponding integral lines by ${\bf r}_{1,2}(t;\x'_0,\y'_0,\x_0,\y_0)$, with parameter $t$ and initial position $(\x'_0,\y'_0,\x_0,\y_0)$.
Moreover, denote ${\bf r}_{1,2}=(\x',\y',\x,\y):\RR\rightarrow \RR^4$. The curve ${\bf r}_1$ satisfies
\[\dot{{\bf r}}_1(t)=(\dot{\x}'(t),\dot{\y}'(t),\dot{\x}(t),\dot{\y}(t))\]
\[= \Big(\frac{\x\x'+\y\y'}{\abs{(\x,\y)}^2}  ,-\frac{-\x\y'+\x'\y}{\abs{(\x,\y)}^2} ,1, 0 \Big)\]
with initial condition
\[\Big(\x'(0),\y'(0),\x(0),\y(0)\Big)= \Big(\x'_0,\y'_0,\x_0,\y_0\Big).\]

The solution of these ODE's are
\begin{equation}
\x(t)=t+\x_0
\label{eq:Shear_x_flow100}
\end{equation}
\begin{equation}
\y(t)=\y_0.
\label{eq:Shear_x_flow2}
\end{equation}
\begin{equation}
\x'(t) = \frac{(\x_0\x_0'+\y_0\y_0')(\x_0+t)+\y_0^2\x_0'-\x_0\y_0\y'_0}{\x_0^2 +\y_0^2}.
\label{eq:angle_scale_x_prime_trans}
\end{equation}
\begin{equation}
\y'(t)= \frac{(\x_0\y_0'-\y_0\x_0')(\x_0+t) + \y_0^2\y_0' + \x_0\x_0'\y_0}{\x_0^2 +\y_0^2}
\label{eq:angle_scale_y_prime_trans}
\end{equation}

For ${\bf r}_2$, we have
\[\dot{\bf r}_2(t) = \Big( \frac{-\x\y' +\x'\y}{\x^2+\y^2} , \frac{\x \x' + \y\y'}{\x^2+\y^2},  0 , 1 \Big)\]
The solution is $\x(t)=\x_0$, $\y(t)=\y_0+t$, and
\begin{equation}
\x'(t) = \frac{(\y_0+t)(\y_0\x_0'-\x_0\y_0')  +\x_0^2\x_0' + \x_0\y_0\y_0'}{\x_0^2 +\y_0^2}.
\label{eq:angle_scale_x_prime_trans_gy}
\end{equation}
\begin{equation}
\y'(t)= \frac{(\y_0+t)(\y_0\y_0'+\x_0\x_0') + \x_0^2\y_0'-\x_0\y_0\x_0'}{\x_0^2 +\y_0^2}
\label{eq:angle_scale_y_prime_trans00}
\end{equation}

The above calculation of $\exp({it{\bf\bT}_1})$ allows the application of trigonometric polynomial $position$-pass filters.

%$ $

%\noindent
%\emph{Shearlet transform}
\subsection{Shearlet transform}

In this example we assume that the supports of $\hs$ and $\hf$ are in $\RR^2_+=\{(\x,\y)\in\RR^2\ |\ \x>0\}$. By this, we can ignore the reflection subgroup $N_4$ of $G$.

Using the $scale$ transform  (\ref{eq:scale_t}), it is straightforward to show
\[[\Psi_3^*\bQ_3\Psi_3 \hs](\x,\y)=-\ln(\x)\hs(\x,\y).\]
According to (\ref{eq:pull1_self_adjoint}), $\bT_{3}$ is given by
\[\bT_3(\hf\otimes \hs) = \hf\otimes(\Psi_3^*\bQ_3\Psi_3 \hs) - (\Psi_3^*\bQ_3\Psi_3 \hf)\otimes \hs\]
or
\[[\bT_3(\hf\otimes \hs)](\x',\y',\x,\y)= -\ln(\x)\hs(\x,\y)\overline{\hf(\x',\y')} + \ln(\x')\hs(\x,\y)\overline{\hf(\x',\y')}.\]
Therefore
\[\bT_3 F(\x',\y',\x,\y) = (\ln(\x')-\ln(\x))F(\x',\y',\x,\y).\]
To derive a formula for $scale$-pass filters, note that the space of functions $L^{\infty}(\RR)$, applied on $\bT_3$, is exhausted by the space of functions 
\[\{ \eta(\frac{\x'}{\x}) \ |\ \eta\in L^{\infty}(\RR)\}.\]

To calculate $\bT_2$,
we use
\[A_2(g_3)=e^{\frac{1}{2}g_3}\]
to get
\[A_2(\bT_3)F(\x',\y',\x,\y)=e^{\frac{1}{2}(\ln(\x')-\ln(\x))}F(\x',\y',\x,\y)={\x}^{\prime\frac{1}{2}}\x^{-\frac{1}{2}}F(\x',\y',\x,\y).\]
Therefore
\[\bT_2 (\hf\otimes \hs)=\hf\otimes(\Psi_2^*\bQ_2\Psi_2 \hs) - A_2(\bT_3)\big((\Psi_2^*\bQ_2\Psi_2 \hf)\otimes \hs\big)\]
is given by
\[\bT_2 (\hf\otimes \hs)(\x',\y',\x,\y)= -\frac{\y}{\x}\overline{\hf(\x',\y')}\hs(\x,\y) + {\x}^{\prime\frac{1}{2}}\x^{-\frac{1}{2}}\frac{\y'}{\x'}\overline{\hf(\x',\y')}\hs(\x,\y).\]
Thus
\[\bT_2 F(\x',\y',\x,\y)=\Big( -\frac{\y}{\x} +{\x}^{\prime-\frac{1}{2}}\x^{-\frac{1}{2}}\y'\Big)F(\x',\y',\x,\y).\]
To derive a formula of $slope$-pass filters, note that the space of functions $L^{\infty}(\RR)$, applied on $\bT_2$, is exhausted by the space of functions 
\[\big\{ \eta\big( -\frac{\y}{\x} +{\x}^{\prime-\frac{1}{2}}\x^{-\frac{1}{2}}\y'\big) \ |\ \eta\in L^{\infty}(\RR)\big\}.\]

To calculate ${\bf\bT}_1$,
we use
\[{\bf A}_1(g_2,g_3) = D(g_2,g_3)= \left(
\begin{array}{cc}
	e^{g_3} & g_2e^{\frac{1}{2}g_3}\\
	   0    &  e^{\frac{1}{2}g_3}
\end{array}
\right)\]
to get
\[\begin{split}
{\bf A}_1(\bT_2,\bT_3) = &  \left(
\begin{array}{cc}
	e^{(\ln(\x')-\ln(\x))} & ( -\frac{\y}{\x} +{\x}^{\prime-\frac{1}{2}}\x^{-\frac{1}{2}}\y')e^{\frac{1}{2}(\ln(\x')-\ln(\x))}\\
	   0    &  e^{\frac{1}{2}(\ln(\x')-\ln(\x))}
\end{array}
\right)\\
 = & \left(
\begin{array}{cc}
	\x'\x^{-1} &  -\y\x^{\prime \frac{1}{2}}\x^{-\frac{3}{2}} +\x^{-1}\y'\\
	   0    &  \x^{\prime \frac{1}{2}}\x^{-\frac{1}{2}}
\end{array}
\right).
\end{split}\]

Therefore
\[{\bf\bT}_1 (\hf\otimes \hs)=\hf\otimes(\Psi_1^*{\bf\bQ}_1\Psi_1 \hs) - A_1(\bT_2,\bT_3)\big((\Psi_1^*{\bf\bQ}\Psi_1 \hf)\otimes \hs\big)\]
so
\[\begin{split}
\bT_1 (\hf\otimes \hs) (\x,\y)= & \left(\begin{array}{c}
	i\frac{\partial}{\partial \x}\\
	i\frac{\partial}{\partial \y}
\end{array}\right)\overline{\hf(\x',\y')}\hs(\x,\y) \\
 & - \left(
\begin{array}{cc}
	\x'\x^{-1} &  -\y\x^{\prime \frac{1}{2}}\x^{-\frac{3}{2}} +\x^{-1}\y'\\
	   0    &  \x^{\prime \frac{1}{2}}\x^{-\frac{1}{2}}
\end{array}
\right)\left(\begin{array}{c}
	-i\frac{\partial}{\partial \x'}\\
	-i\frac{\partial}{\partial \y'}
\end{array}\right)\overline{\hf(\x',\y')}\hs(\x,\y) \\
& =\left(\begin{array}{c}
	i\frac{\partial}{\partial \x}\\
	i\frac{\partial}{\partial \y}
\end{array}\right)\overline{\hf(\x',\y')}\hs(\x,\y)  \\
& + \left(
\begin{array}{cc}
	\x'\x^{-1}i\frac{\partial}{\partial \x'} -  (\y\x^{\prime \frac{1}{2}}\x^{-\frac{3}{2}} +\x^{-1}\y')i\frac{\partial}{\partial \y'}\\
	        \x^{\prime \frac{1}{2}}\x^{-\frac{1}{2}}i\frac{\partial}{\partial \y'}
\end{array}
\right)\overline{\hf(\x',\y')}\hs(\x,\y).
\end{split}\]
Thus
\[\begin{split}
& {\bf\bT}_1 F(\x',\y',\x,\y) = \\
& \left(
\begin{array}{cc}
	i\frac{\partial}{\partial \x} + \x'\x^{-1}i\frac{\partial}{\partial \x'} -  (\y\x^{\prime \frac{1}{2}}\x^{-\frac{3}{2}} +\x^{-1}\y')i\frac{\partial}{\partial \y'}\\
	   i\frac{\partial}{\partial \y}  +  \x^{\prime \frac{1}{2}}\x^{-\frac{1}{2}}i\frac{\partial}{\partial \y'}
\end{array}
\right)F(\x',\y',\x,\y). 
\end{split}\]
The exponential $\exp(it{\bf\bT}_1)$ applied on $F\in\cW\otimes\cS$ is given by
\[e^{i t {\bf\bT}_1 } F(\x',\y',\x,\y) =: {\bf F}(\x',\y',\x,\y;t)\]
with initial condition
\[{\bf F}(\x',\y',\x,\y;0)=\Big(F(\x',\y',\x,\y),F(\x',\y',\x,\y)\Big).\]
The pair ${\bf T}$ satisfies the two first order linear PDEs
\[\begin{split}
 & \frac{\partial}{\partial t}{\bf F}(\x',\y',\x,\y;t)  \\
 & =\left(
\begin{array}{cc}
	-\frac{\partial}{\partial \x} - \x'\x^{-1}\frac{\partial}{\partial \x'} +  (\y\x^{\prime \frac{1}{2}}\x^{-\frac{3}{2}} +\x^{-1}\y')\frac{\partial}{\partial \y'}\\
	   -\frac{\partial}{\partial \y}  -  \x^{\prime \frac{1}{2}}\x^{-\frac{1}{2}}\frac{\partial}{\partial \y'}
\end{array}
\right){\bf F}(\x',\y',\x,\y;t) 
\end{split}\]

Denote the integral lines by ${\bf r}_{1,2}(t;\x'_0,\y'_0,\x_0,\y_0)$.
Moreover, denote ${\bf r}_{1,2}=(\x',\y',\x,\y)$. Then the curve ${\bf r}_1$
 satisfies
\[\begin{split}
\dot{{\bf r}}_1(t)= & \big(\dot{\x}'(t),\dot{\y}'(t),\dot{\x}(t),\dot{\y}(t)\big)\\
 = & \Big(-\x'(t)\x(t)^{-1} , \y(t)\x'(t)^{\frac{1}{2}}\x(t)^{-\frac{3}{2}} +\x(t)^{-1}\y'(t) , -1, 0 \Big)
\end{split}\]
with initial condition
\[\Big(\x'(0),\y'(0),\x(0),\y(0)\Big)= \Big(\x'_0,\y'_0,\x_0,\y_0\Big).\]
Thus,
\begin{equation}
\x(t)=t+\x_0
\label{eq:Shear_x_flow1}
\end{equation}
\begin{equation}
\y(t)=\y_0.
\label{eq:Shear_x_flow200}
\end{equation}
\begin{equation}
\x'(t) = \x'_0\x_0^{-1} (t+\x_0)=\x'_0\x_0^{-1} t + \x'_0.
\label{eq:Shear_x_flow3}
\end{equation}
and
\begin{equation}
\y'(t) = \frac{-\y_0\x_0^{\prime \frac{1}{2}}\x_0^{-\frac{1}{2}}t + \x_0\y_0'}{t+\x_0}
\label{eq:Shear_x_flow4}
\end{equation}

For ${\bf r}_2$ we have
\[\dot{{\bf r}}_2(t)=(\dot{\x'}(t),\dot{\y'}(t),\dot{\x}(t),\dot{\y}(t))\]
\[=\Big( 0 , \x'(t)^{\frac{1}{2}}\x(t)^{-\frac{1}{2}}, 0 , 1 \Big)\]
with solution
\begin{equation}
\x(t)=\x_0
\label{eq:Shear_y_flow1}
\end{equation}
\begin{equation}
\y(t)=t+\y_0
\label{eq:Shear_y_flow2}
\end{equation}
\begin{equation}
\x'(t)=\x'_0.
\label{eq:Shear_y_flow3}
\end{equation}
and
\begin{equation}
\y'(t)=\x_0^{\prime\frac{1}{2}}\x_0^{-\frac{1}{2}} t + \y'_0.
\label{eq:Shear_y_flow300}
\end{equation}

The formula for $\exp({it{\bf\bT}_1})$ allows the calculation of trigonometric polynomials in ${\bf\bT}_1$, and thus of approximate $position$-pass filters.

\section{\textcolor{black}{Implementation of the WP4 search algorithm}}
\label{Implementation of the Wavelet-Plancherel search algorithm}

When computing the filters of the search algorithms, splines are restricted to interval subdomain, summed with other splines, and dilated. We first explain how to implement these three operations.

Consider a spline
\[\hf =  \{\w'_k,\hf_n(\w'_k)\}_{k=1}^{K},\]
where $\{w'_k\}_{k=1}^{K}$ are the $K\in\NN$ knots of the spline. We define the \emph{restriction} of the spline to the interval $[a,b]\subset\RR$, denoted by $[a,b]\cap \hf$, as the following spline.
The first knot of $[a,b]\cap \hf$ is 
\begin{equation}
\w'_a=\left\{ 
\begin{array}{cc}
  a   & {\rm if\ } \w'_1 \leq a\\
  \w'_1   & {\rm if\ } \w'_1> a.
\end{array}
\right.
\label{eq_SP_res0}
\end{equation}
Similarly, the last knot is
\begin{equation}
\w'_b=
\left\{ 
\begin{array}{cc}
  b   & {\rm if\ } \w'_K \geq b\\
  \w'_K   & {\rm if\ } \w'_K < b.
\end{array}
\right.
    \label{eq_SP_res1}
\end{equation}
Then, $[a,b]\cap \hf$ is implemented, using (\ref{eq_SP_res0}) and (\ref{eq_SP_res1}) by
\begin{equation}
[a,b]\cap \hf = 
\left\{
\begin{array}{ccc}
   \emptyset  & {\rm if\ } & \w'_K\leq a {\rm\ or\ } \w'_1 \geq b \\
   \{\w'_a, \hf(\w'_a)\}\cup\{\w'_k,\hf_n(\w'_k)\}_{a<\w'_k<b} \cup \{\w'_b, \hf(\w'_b)\}  & {\rm else} &
\end{array}
\right.
    \label{eq_SP_res}
\end{equation}
where in (\ref{eq_SP_res}), $\cup$ is understood as the concatenation of an element to a sequence. Note that $[a,b]\cap \hf$ can be implemented in $O(K)$ operation.

To emphasize the  implementational aspect of summation, we denote by $\hf_1\uplus\hf_2$ the sum of the two splines $\hf_1=\{\w'_{1,k}, \hf_1(\w'_{1,k})\}_{k=1}^{K_1}$ and $\hf_2=\{\w'_{2,k}, \hf_2(\w'_{2,k})\}_{k=1}^{K_2}$. This is implemented by uniting the two knot sets, and summing the values of the splines in the united knot set. Namely,
\begin{equation}
\hf_1\uplus\hf_2 = \Big\{ \{\w'_{1,k}\}_{k=1}^{K_1} \cup \{\w'_{2,k}\}_{k=1}^{K_2} \ , \ [\hf_1+\hf_2]\big(\{\w'_{1,k}\}_{k=1}^{K_1} \cup \{\w'_{2,k}\}_{k=1}^{K_2}\big) \Big\}.
    \label{eq_SP_sum}
\end{equation}
Here, $[\hf_1+\hf_2]\big(\{\w'_{1,k}\}_{k=1}^{K_1} \cup \{\w'_{2,k}\}_{k=1}^{K_2}\big) $ is the sequence of evaluations of $\hf_1+\hf_2$ on the united sample set. Uniting the two knot sequences to an ordered sequence is implemented in $O(K_1+K_2)$ operations, which is thus also the complexity of $\uplus$. We denote
\[\biguplus_{n=1}^N \hf_n = \hf_1\uplus\hf_2\uplus\ldots \uplus\hf_N.\]

The dilation of a spline by the scale $\ln(a)$ is defined by
\[\hat{D}_a\hf =  \{a\w'_k,\hf_n(\w'_k)\}_{k=1}^{K}.\]

Using the operation $[a,b]\cap$, $\uplus$, and $\hat{D}_a$, we can now formulate scale-pass and time-pass filters of spline sequences.
Consider the spline sequence
\[\{\hf_n\}_{n=0}^N = \big\{  \{\w'_{n,k},\hf_n(\w'_{n,k})\}_{k=1}^{K_n}\big\}_{n=0}^N,\]
where $\{\w'_{n,k}\}_{k=1}^{K_n}$ are the $K_n\in\NN$ knots of the $n$th spline. Denote that spacing of the signal grid $\{\w_n\}_{n=0}^N$ by $h>0$.

%{more formal and short definition of the scale to freq change. In def env?}

As explained in Subsections \ref{1D continuous wavelet as a time-frequency transform} and \ref{Sampling resolution in phase space}, we consider the change of variable between scale $g_2$ and frequency $\kappa$, given by 
\begin{equation}
   \kappa = -\ln(g_2) 
   \label{eq:S-T_cov}
\end{equation}
When representing the time-scale phase space using the change of variable (\ref{eq:S-T_cov}), it becomes the time-frequency plane with the standard Euclidean measure as the Haar measure (see (\ref{eq:TF_Haar})).
We formulate the search algorithm in terms of frequency bisections instead of scale bisections. Hence, bisecting a frequency interval at its middle point produces two sub-intervals of equal measures.

A frequency-pass filter which restricts spline sequences to a slope interval $z=\frac{\w'}{\w}\in[\frac{1}{\k^{\rm r}},\frac{1}{\k^{\rm l}}]$ is denoted by $S([\k^{\rm l},\k^{\rm r}])=\chi_{[-\ln(\k^{\rm r}),-\ln(\k^{\rm l})]}(\bT_2)$, where $\chi_{[-\ln(\k^{\rm r}),-\ln(\k^{\rm l})]}$ is the characteristic function of the interval $[-\ln(\k^{\rm r}),-\ln(\k^{\rm l})]$.
Here,  $[\k^{\rm l},\k^{\rm r}]$ represents a frequency interval.
The operator $S([\k^{\rm l},\k^{\rm r}])$ is implemented by
\[S([\k^{\rm l},\k^{\rm r}])\{\hf_n\}_{n=0}^N  = \{[\frac{\w_n}{\k^{\rm r}},\frac{\w_n}{\k^{\rm l}}]\cap \hf_n\}_{n=0}^N.\]

A translation along the slope rays by $t\in\NN$ pixels, $\hat{T}_{\rm slope}(t) = \exp(ith\bT_1)$, is computed by translating along $\w$, and dilating along $\w'$.
Denote
\[\hat{T}_{\rm signal}(t) \{\hf_n\}_n = \{\hf_{n-t}\}_n.\]
Then the slope translation is given by
\[\hat{T}_{\rm slope}(t)\{\hf_n\}_{n=0}^N = \hat{T}_{\rm signal}(t)\{\hat{D}_{\frac{\w_n}{\w_n-ht}} \hf_n\}_{n=0}^N.\]
Hence, for a trigonometric polynomial
$R(x) = \sum_{l=-L}^L c_l e^{ilhx}$, 
the time-filter $R(\bT_1)$ is given by
\[R(\bT_1) = \biguplus_{l=-L}^L c_l \hat{T}_{\rm slope}(l)\{\hf_n\}_n.\]

In the search algorithm, the filters $R$ are taken as dilated versions of a trigonometric polynomial approximation of $\chi_{[-\pi,0]}$ in $L^2[-\pi,\pi]$, as explained in Subsection \ref{A search algorithm via the wavelet-Plancherel theory}.

Lastly, in the search algorithm we compute norms of spline sequences, which are implemented by
\begin{equation}
\norm{\{\w'_{n,k},\hf_n(\w'_{n,k})\}_{k=1}^{K_n}\big\}_{n=0}^N}^2 = \sum_{n=0}^N\int_{\w'_1}^{\w'_{K_n}}\abs{\hf_n(\w')}^2d\w'.
    \label{eq_SP_int}
\end{equation}
The integral of the piece-wise quadratic function $\abs{\hf_n(\w')}^2$ in (\ref{eq_SP_int}) has a closed form formula.

The search algorithm can now be formulated in pseudo-code.

\begin{algorithm}[Wavelet-Plancherel Search Algorithm]
\label{Wavelet-Plancherel Search Algorithm}
$ $ \newline
\emph{Input}: 

$\quad$ $\hf=\{\w'_k,\hf(\w'_k)\}_{k=1}^K$: spline window 

$\quad$ $\hs=\{\hs(\w_n)\}_{n=1}^N$: discrete signal  of spacing $h$ and resolution $N=2^J$, with $J\in\NN$. \newline
\emph{Output}: a point $\mathbf{g}$ in phase space with large coefficient $V_{\hf}[\hs](\mathbf{g})$. \newline
\emph{Notations}: 

$\quad$ $R(x) = \sum_{l=-L}^L c_l e^{ilhx}$: a trigonometric polynomial approximating $\chi_{[-\pi,0]}$ in $L^2[-\pi,\pi]$

$\quad$  $t_j^{\rm l}, t_j^r$: the left and right boundaries of the time band of length $2^{-j}(t_0^{\rm r}-t_0^{\rm l})$  at step $j$

$\quad$  ${\kappa}^{\rm l}_j, {\kappa}^{\rm r}_j$: the left and right boundaries of the frequency band of length $2^{-j}({\kappa}^{\rm r}_0-{\kappa}^{\rm l}_0)$ at step $j$

$\quad$ $F_j$: the approximate projection of $\hf\otimes\hs$ to the time-frequency rectangle   $[t_j^{\rm l},t_j^{\rm r}]\times [{\kappa}^{\rm l}_j, {\kappa}^{\rm r}_j]$.

\begin{itemize}
    \item
    Initialize
    \begin{itemize}
        \item 
        $F_0=\hf\otimes\hs$
        \item
    $t_0^{\rm l}=0$, $t_0^{\rm r}=N$
        \item
     ${\kappa}^{\rm l}_0= \frac{\w_0}{\w'_1}$, ${\kappa}^{\rm r}_0= \frac{\w_N}{\w'_K}$
    \end{itemize}
    \item
    For $j=1,\ldots,J$ repeat
    \begin{itemize}
        \item 
        Define $R_j(x) = \sum_{l=-L}^L c_l e^{il2^jhx}$
        \item
        Compute the four projections
        
        $R_{j-1}(\bT_1) \ S([{\kappa}_{j-1}^{\rm l},{\kappa}_{j-1}^{\rm l}+2^{-j}])\ F_{j-1}$, 
        
        $(1-R_{j-1}(\bT_1)) \ S([{\kappa}_{j-1}^{\rm l},{\kappa}_{j-1}^{\rm l}+2^{-j}])\ F_{j-1}$,
        
        $R_{j-1}(\bT_1) \ S([{\kappa}_{j-1}^{\rm r}-2^{-j},{\kappa}_{j-1}^{\rm r}])\ F_{j-1}$,
        
        $(1-R_{j-1}(\bT_1)) \ S([{\kappa}_{j-1}^{\rm r}-2^{-j},{\kappa}_{j-1}^{\rm r}])\ F_{j-1}$

        \item
        Compute the four norms of the above projections
        \item
        Set the new boundaries $t_j^{\rm l}, t_j^{\rm r},{\kappa}_j^{\rm l}, {\kappa}_j^{\rm r}$ as the boundaries of the projection with largest norm, and set $F_j$ to be the corresponding projected spline sequence.
    \end{itemize}
    \item
    Return $\mathbf{g}=\Big(t_j^{\rm l}+\frac{1}{2},-\ln\big({\kappa}_j^{\rm l}+2^{-j-1}({\kappa}^{\rm r}_0-{\kappa}^{\rm l}_0)\big)\Big)$.
\end{itemize}

\end{algorithm}

\end{document}